\documentclass[12pt, draftclsnofoot,onecolumn,letterpaper]{IEEEtran}
\usepackage{inputenc}
\usepackage{cite,acronym,paralist,xspace}
\usepackage[cmex10]{amsmath}
\usepackage{amssymb,amsfonts,amsthm}
\usepackage{algorithmic}
\usepackage{pgf}
\usepackage{textcomp}
\usepackage{xcolor}
\usepackage{subcaption}
\usepackage{stmaryrd} 
\usepackage{dsfont}
\usepackage{balance}
\usepackage{diagbox}
\usepackage{enumitem}
\usepackage{cite}

\usepackage[linesnumbered,ruled,vlined]{algorithm2e}
\newtheorem{definition}{Definition}
\newtheorem{theorem}{Theorem}%
\newtheorem{lemma}[theorem]{Lemma}%
\newtheorem{remark}{Remark}

\renewcommand{\leq}{\leqslant} 
\renewcommand{\geq}{\geqslant} 
\newcommand{\argmin}{\mathop{\text{argmin}}} 
\newcommand{\argmax}{\mathop{\text{argmax}}} 

\DeclareMathAlphabet{\eurm}{U}{eur}{m}{n}
\DeclareMathAlphabet{\mathbsf}{OT1}{cmss}{bx}{n}
\DeclareMathAlphabet{\mathssf}{OT1}{cmss}{m}{sl}
\DeclareMathAlphabet{\mathcsf}{OT1}{cmss}{sbc}{n}

\DeclareSymbolFont{bsfletters}{OT1}{cmss}{bx}{n}  
\DeclareSymbolFont{ssfletters}{OT1}{cmss}{m}{n}
\DeclareMathSymbol{\bsfGamma}{0}{bsfletters}{'000}
\DeclareMathSymbol{\ssfGamma}{0}{ssfletters}{'000}
\DeclareMathSymbol{\bsfDelta}{0}{bsfletters}{'001}
\DeclareMathSymbol{\ssfDelta}{0}{ssfletters}{'001}
\DeclareMathSymbol{\bsfTheta}{0}{bsfletters}{'002}
\DeclareMathSymbol{\ssfTheta}{0}{ssfletters}{'002}
\DeclareMathSymbol{\bsfLambda}{0}{bsfletters}{'003}
\DeclareMathSymbol{\ssfLambda}{0}{ssfletters}{'003}
\DeclareMathSymbol{\bsfXi}{0}{bsfletters}{'004}
\DeclareMathSymbol{\ssfXi}{0}{ssfletters}{'004}
\DeclareMathSymbol{\bsfPi}{0}{bsfletters}{'005}
\DeclareMathSymbol{\ssfPi}{0}{ssfletters}{'005}
\DeclareMathSymbol{\bsfSigma}{0}{bsfletters}{'006}
\DeclareMathSymbol{\ssfSigma}{0}{ssfletters}{'006}
\DeclareMathSymbol{\bsfUpsilon}{0}{bsfletters}{'007}
\DeclareMathSymbol{\ssfUpsilon}{0}{ssfletters}{'007}
\DeclareMathSymbol{\bsfPhi}{0}{bsfletters}{'010}
\DeclareMathSymbol{\ssfPhi}{0}{ssfletters}{'010}
\DeclareMathSymbol{\bsfPsi}{0}{bsfletters}{'011}
\DeclareMathSymbol{\ssfPsi}{0}{ssfletters}{'011}
\DeclareMathSymbol{\bsfOmega}{0}{bsfletters}{'012}
\DeclareMathSymbol{\ssfOmega}{0}{ssfletters}{'012}

\newcommand{\calX}{{\mathcal{X}}}

\SetCommentSty{mycommfont}

\SetKwInput{KwInput}{Input}                
\SetKwInput{KwInitial}{Initialization}                
\SetKwInput{KwOutput}{Output}              

\title{Covert Online Decision Making:\\ From Sequential Hypothesis Testing to Stochastic Bandits}
\author{\IEEEauthorblockN{Meng-Che Chang and Matthieu R. Bloch}}

\begin{document}
\maketitle
\begin{abstract}
    We study the problem of covert online decision making in which an agent attempts to identify a parameter governing a system by probing the system while escaping detection from an adversary. The system is modeled as Markov kernel whose input is controlled by the agent and whose two outputs are observed by the agent and the adversary, respectively. This problem is motivated by applications such as covert sensing or covert radar, in which one tries to perform a sensing task without arousing suspicion by an adversary monitoring the environment for the presence of sensing signals. Specifically, we consider two situations corresponding to different amounts of knowledge of the system. If the kernel is known but governed by an unknown fixed parameter, we formulate the problem as a sequential hypothesis testing problem. If the kernel determining the observations of the agent is unknown but the kernel determining those of the adversary is known, we formulate the problem as a best arm identification problem in a bandit setting. In both situations, we characterize the exponent of the probability of identification error. As expected because of the covertness requirement, the probability of identification error decays exponentially with the \emph{square-root} of the blocklength.
\end{abstract}

\IEEEpeerreviewmaketitle

\section{Introduction}
\label{sec:introduction}

Decision making encompasses many fundamental problems in the area of communication, control, sensing, and machine learning. 
A common objective in decision making problems consists in identifying the unknown but fixed parameters governing the distribution of observations. This type of decision making problem falls under the framework of {\it hypothesis testing} \cite{Chernoff1952,Hoeffding1994,Neyman1992,Cover_Thomas}. When the number of observations is fixed, the problem is called {\it fixed-length hypothesis} testing and the optimum performance has been thoroughly analyzed either in Stein's region \cite{Cover_Thomas} or in the Bayesian region \cite{Chernoff1952}, for which the detection error exponent is characterized by the relative entropy and the Chernoff information between distributions corresponding to difference parameters, respectively. 
When the number of observations varies, the problem is called {\it sequential hypothesis testing} \cite{Wald1948,Chernoff1959,Baum1994}. The detection error probability exponent has again been thoroughly investigated and is in general larger than in the fixed-length setting.
The extension of sequential hypothesis to situations in which the decision maker can causally choose one of several available actions to influence the distribution of its observations, has been analyzed in \cite{Nitinawarat2013,Naghshvar2013}, and named {\it controlled sensing} or {\it active hypothesis testing} in \cite{Nitinawarat2013} and \cite{Naghshvar2011}, respectively. 
Another objective of a decision maker can be to identify the action that results in the largest average rewards from its observations without  knowledge of output distributions. This problem falls in the framework of {\it best-arm identification} (BAI) in the multi-arm bandit (MAB) literature \cite{EvenDar2002,EvenDar2006,Gariver2016}. 
The problem formulation of BAI can be traced back to \cite{EvenDar2002}, where the authors analyze the bandit problem in the {\it probabilistic approximately correct} setting, i.e., the designed algorithm should identify the best arm (or action) with high probability as fast as possible. The problems of controlled sensing and BAI have very similar mathematical models. In controlled sensing, we are given a known set of distribution $\{\nu_{\theta}^x\}_{\theta\in\Theta,x\in\mathcal{X}}$, where $\Theta$ is the set of parameters and $\mathcal{X}$ is the set of actions. For any parameter $\theta\in\Theta$, the set of distributions $\{\nu_{\theta}^x\}_{x\in\mathcal{X}}$ can also be viewed as a bandit machine, where $\mathcal{X}$ is the set of arms. Because of their similarity, we analyze these two problems in a unified framework. 

Concurrently to advances in decision making, \emph{security} has become an increasingly important concern when designing an algorithm or a system because of the growing amount of sensitive data involved. A common security requirement is {\it secrecy} \cite{Bloch2011PhysicalLayerSF,Bloch2013}, by which the objective is to ensure that no information is leaked about the transmitted data through the observations of an adversary. While secrecy is relevant for many communication problems, the concept is not always directly applicable in the context of decision making because the actions carried out do not necessarily carry information by themselves but rather act as probing signals to identify an underlying parameter. As an example, \cite{Chang2020} analyzes secrecy in the context of controlled sensing by defining the performance metric as the ratio of error exponents between the legitimate receiver and an eavesdropper, hence capturing the ability of decision making policies to slow-down the decision making of an adversary. In addition to secrecy, {\it covertness} has emerged as a useful security concept. In covert communications, the goal is to hide the presence of the communication from an adversary. It has been shown in \cite{Bash2013,Bloch2016,Mehrdad2020} that, in order to achieve covertness, a \emph{square-root law} should be satisfied, which states that the number of effective channel uses should on the order of $\sqrt{n}$, where $n$ is the blocklength. Motivated by potential applications to undetectable radar and sensing~\cite{Gagatsos2019}, our goal is to explore the performance of covert decision making in sense that we make precise next. We note that there have been intriguing related works in the context of covert control~\cite{Bai2014,Bai2015,Barak2022} but our focus is on sensing rather than control. There exist conceptual similarities between our approach and recent work investigating covert communication from the perspective of change point detection~\cite{Huang2020,Huang2021}, although the metrics and results are fairly different because of our focus on sensing rather than communication. 

The first works studying covertness in decision making can be found in \cite{Goeckel2017,Gagatsos2019,Tahmasbi2020,Tahmasbi2020c,Tahmasbi2021}, motivated by quantum covert sensing. In particular,~\cite{Tahmasbi2020} analyzes the exponent of active hypothesis testing while maintaining certain covertness constraint in the {\it fixed-length} setting. Since the error exponent of active hypothesis testing is larger 
in the {\it sequential} setting~\cite{Nitinawarat2013, Li2022}, we seek here to investigate the performance of sequential hypothesis testing with covertness constraints.  
As an attempt to study the problem in presence of uncertainty regarding the hypotheses tested, we formulate a related best-arm identification problem, which we show admits to a similar characterization. There exist related works that study the performance of decision making algorithm under the existence of adversary. 
The authors in \cite{Lykouris18} and \cite{Gupta19a} design algorithms that minimize the regret under the influence of adversarial attack, while the problem of BAI with adversarial corruptions is analyzed in \cite{Zhong2021} and \cite{mukherjee2021}. Moreover, the influence of adversarial attacks on either observations or actions in controlled sensing can also be found in \cite{Chang2022_controlsensing} and \cite{Jin2021}. However, none of these work study the problem of BAI or sequential hypothesis testing with covertness constraints. 

Our main contributions are as follows.
\begin{itemize}
    \item We formally introduce the problem of covert sequential decision making, which includes sequential hypothesis testing and best-arm identification, and relate the stopping time and covertness constraints to the adversary's ability to identify the presence of decision making protocol.
    \item We characterize both lower bounds and  upper bounds for the detection error exponent. Our results show that the upper bounds and the lower bounds match for certain general classes of protocols.
    \item Our proofs differ from standard proofs by requiring a careful analysis of the decision making process in the presence of covertness constraints. In particular, we rely extensively on Freedman's inequality throughout our analysis. 
\end{itemize}
While the results presented here a largely of theoretical nature, they shed light on how to design decision making process that achieve their goals without arousing suspicion. Such results could already find applications in the context of quantum sensing and ranging~\cite{Hao2022Demonstration}.

The rest of the paper is organized as follows. After briefly reviewing notation in Section~\ref{sec:notation}, we formally introduce our problem in Section~\ref{sec:model} and state our main results together will illustrative numerical examples in Section~\ref{sec:main_results}. In Section~\ref{sec:math_tool}, we review mathematical tools that are frequently used throughout the paper. 
Finally, the proofs of our main theorems are relegated to Section~\ref{sec:prof_thm1}-Section~\ref{sec:prof_thm4}. 

\section{Notation}
\label{sec:notation}

Let $\mathcal{X}$ be any alphabet set. For any $n\in\mathbb{N}$, $x^n=(x_1,\cdots,x_n)$ is a sequence of variables in $\mathcal{X}^n$, and for any $i\leq j\leq n$, $x_i^j = (x_1,\cdots, x^j)$ is a sub-sequence of of $x^n$. 
The set of distributions on $\mathcal{X}$ is denoted by $\mathcal{P}_{\mathcal{X}}$. 
If $p\in\mathcal{P}_{\mathcal{X}}$ and $q\in\mathcal{P}_{\mathcal{X}}$ be two different distributions on $\mathcal{X}$, we define $\mathbb{V}(p\Vert q) \triangleq \frac{1}{2} \sum_{x\in\mathcal{X}}|p(x)-q(x)|$,  
$\mathbb{D}(p\Vert q)\triangleq\sum_x p(x)\log\frac{p(x)}{q(x)}$ and $\chi_2(p\Vert q)=\sum_{x\in\mathcal{X}}\frac{(p(x)-q(x))^2}{q(x)}$ as the total variational distance,the  relative entropy and the Chi-square distance between $p$ and $q$, respectively. We say $p$ is absolutely continuous with respect to (w.r.t.) $q$, denoted by $p \ll q$, if for all $x\in\calX$, $p(x)=0$ if $q(x)=0$. For any distribution $p\in\mathcal{P}_{\mathcal{X}}$, the mean of the distribution $p$ is denoted by $\mu(p)\triangleq \sum_{x\in\mathcal{X}}xp(x)$. We also denote by $\Vert p-q \Vert_{\infty}\triangleq \max_{x\in\mathcal{X}} |p(x)-q(x)|$. 
$\mathbb{N}^+\triangleq \{1,2,\cdots\}$ is the set of all positive integers. The Landau notation $g(n) = O_{n\rightarrow \infty}(f(n))$ means that $g(n) \leq C f(n)$ for some $0<C<\infty$ for all $n$ sufficiently large. Similarly, the notation $g(\delta) = O_{\delta\rightarrow 0}(f(\delta))$ means that $g(\delta) \leq C f(\delta)$ for some $0<C<\infty$ for all $\delta$ small enough. Other Landau notation are defined similarly and the tilde Landau’s notations, i.e. $\Tilde{O}, \Tilde{o}, \Tilde{\Omega}, \Tilde{\omega}$, and $\Tilde{\Theta}$, are the defined as the same way as conventional Landau's notation but ignoring logarithmic factors. For example, $g(n)=\Tilde{O}_{n\rightarrow\infty}(f(n))$ means that $g(n)\leq Cf(n)\times \textnormal{poly}(\log f(n))$ for some polynomial function and some constant $C$ for all $n$ sufficiently large.  

\section{Problem Formulation}
\label{sec:model}
\subsection{General Model Description}
We consider the following general model of an online decision making problem with an adversary. Let $\{\nu_{\theta}^{x}\}_{x\in[0;K],\theta\in\Theta}$ and $\{q_{\theta}^{x}\}_{x\in[0;K],\theta\in\Theta}$ be the collections of distributions of observations of the agent (Alice) and the adversary (Willie), respectively, where $\Theta$ is the known and finite set of hypotheses and $\mathcal{X}\triangleq[0;K]$ is the set of actions. 
We denote by $X_t\in \mathcal{X}$ the action chosen by the agent at each time $t\in\mathbb{N}^+$. The action $X_t$ determines the index of the distributions from which the observations of Alice and Willie are generated. Specifically, for any $t\in\mathbb{N}^+$, the observation $Y_t$ of Alice and the observation $Z_t$ of Willie are generated from the distribution $\nu^{X_t}_{\theta}$ and $q^{X_t}_{\theta}$, respectively, when the hypothesis is $\theta$. Among the set of actions $\mathcal{X}$, we denote by $0$ the null action, which corresponds to the situation in which no effective action is chosen. We assume that the distribution $\nu_{\theta}^0$ has a zero mean for all $\theta\in\Theta$ and $\mathbb{D}(\nu_{\theta}^0\Vert \nu_{\theta'}^0)=0$ for all $\theta\neq\theta'$, implying that taking the null action is useless in distinguishing different hypotheses. The fact that the null action has zero mean is also known by the agent before taking any actions. Moreover, for all $x\in \mathcal{X}\setminus\{0\}$, we also assume that $0<\mathbb{D}(\nu_{\theta}^x\Vert \nu_{\theta'}^x)<\infty$ for all $\theta\neq \theta'$ so that different hypotheses cannot be distinguished perfectly by any single action $x\in\mathcal{X}\setminus\{0\}$. 
The action $X_t$ is chosen according to some causal distribution $P_{X_t|X^{t-1},Y^{t-1}}$, which depends on past actions and observations. Note that if we fix the hypothesis $\theta$, the set of distributions $\mathcal{V}_{\theta}=\{\nu_{\theta}^x\}_{x\in\mathcal{X}}$ and $\mathcal{Q}_{\theta}=\{q_{\theta}^x\}_{x\in \mathcal{X}}$ can be viewed as stochastic bandit machines in which the set of actions $\mathcal{X}$ is the set of arms in the context of multi-arm bandits. For any $\theta\in\Theta$, we also denote by $x^*(\mathcal{V}_{\theta})$ the arm with which the distribution $\nu_{\theta}^{x^*(\mathcal{V}_{\theta})}$ has the largest mean among $\{\nu_{\theta}^x\}_{x\in\mathcal{X}}$.

Different problem formulations are possible when varying the assumptions in our general model. We list the general descriptions of the problems that we would like to analyze.  
\begin{enumerate}
    \item[(P1)] When the agent knows $\{\mathcal{V}_{\theta}\}_{\theta\in\Theta}$ and $\{\mathcal{Q}_{\theta}\}_{\theta\in\Theta}$ but not the hypothesis $\theta$, we  formulate a covert active hypothesis testing problem, in which the objective is to identify $\theta$ subject to a covertness constraint. 
    \item[(P2)] When the agent knows $\{\mathcal{Q}_{\theta}\}_{\theta\in\Theta}$ and $\theta$ but not $\{\mathcal{V}_{\theta}\}_{\theta\in\Theta}$, we formulate a covert best-arm identification problem, in which the objective is to identify the arm resulting in the largest expected reward (observation) for the stochastic bandit $\mathcal{V}_{\theta}$ subjecting to a covertness constraint. In this problem, we make the assumption that $\nu_{\theta}^x$ and $q_{\theta}^x$ are Gaussian distributions with variance $1$ for all $\theta\in\Theta$ and $x\in\mathcal{X}$. We denote by $\mathcal{E}_{\mathcal{N}}$ the set of Gaussian bandits with $|\mathcal{X}|-1$ non-null arms and one null arm, where the distribution corresponding to each arm is a Gaussian distribution and has a unit variance. 
\end{enumerate}
Formally, an online decision-making policy $\pi=(\phi,\varphi,\psi)$ in both (P1) and (P2) is composed of three elements, namely, 
\begin{inparaenum}
    \item a control policy $\phi=\{P_{X_t|X^{t-1},Y^{t-1}}\}_{t=1}^{\infty}$ that determines the actions,
    \item a stopping rule $\varphi$ that determines when the decision-making policy stops, and
    \item the final decision rule $\psi$ that identifies the estimated hypothesis in (P1) and the best arm of the bandit machine $\mathcal{V}_{\theta}$ in (P2).  
\end{inparaenum}
For each $t\in\mathbb{N}^+$, we denote by $S_{t}\in \mathcal{S}\triangleq\{\textnormal{stop},\textnormal{continue}\}$ the status indicating whether the policy stops or not. Then, a stopping rule $\varphi=\{\varphi_t\}_{t=1}^{\infty}$,  $\varphi_t:\mathcal{X}^{t}\times Y^{t} \mapsto \mathcal{S}$ for all $t\in\mathbb{N}^+$, is a function deciding on the status of $S_t$. The status of $S_t$ is stop whenever $S_{k}=\textnormal{stop}$ for any $k<t$. We also denote by $\tau\triangleq \inf\{t\in\mathbb{N}^+:S_t=\textnormal{stop}\}$ the stopping time of the decision-making policy $\pi$, where $\tau$ is adapted to the filtration $(\mathcal{F}_{t})_{t=0}^{\infty}$ with $\mathcal{F}_t=\sigma(X_1,Y_1,\cdots,X_t,Y_t)$ the $\sigma$-algebra generated by $(X_1,Y_1,\cdots,X_t,Y_t)$. 
By saying the decision making policy stops, we mean that null actions are chosen for all $t>\tau$, i.e., $P_{X_t|X^{t-1}=x^{t-1},Y^{t-1}=y^{t-1}}=\mathbf{1}(X_t=0)$ whenever $(x^{t-1},y^{t-1})$ contains a subsequence $(x^{k},y^{k})$ such that $\varphi_k(x^k,y^k)=\textnormal{stop}$ for some $k<t$.  
We let $\mathbb{P}_{\mathcal{V}_{\theta},\mathcal{Q}_{\theta},\pi}$ be the probability measure of the tuple of sequences $(x^{k},y^{k},z^{k})\in \mathcal{X}^{k}\times \mathcal{Y}^{k}\times \mathcal{Z}^{k}$ of any length $k\in\mathbb{N}^+$ under the decision making policy $\pi$ and the stochastic bandits $\mathcal{V}_{\theta}$ and $\mathcal{Q}_{\theta}$, i.e., 
\begin{align}
    \mathbb{P}_{\mathcal{V}_{\theta},\mathcal{Q}_{\theta},\pi}(x^k,y^k,z^k) = \prod_{i=1}^k P_{X_i|X^{i-1}=x^{i-1},Y^{i-1}=y^{i-1}}(x_i) \nu_{\theta}^{x_i}(y_i)q_{\theta}^{x_i}(z_i)
\end{align}
for any $k\in\mathbb{N}^+$. Similarly, $\mathbb{P}_{\mathcal{V}_{\theta},\pi}$ is the probability measure of the tuple of sequences $(x^k,y^k)\in\mathcal{X}^k\times \mathcal{Y}^k$ for any $k\in\mathbb{N}^+$ under the policy $\pi$ and the bandit $\mathcal{V}_{\theta}$. $\mathbb{E}_{\mathcal{V}_{\theta},\pi}$ and $\mathbb{E}_{\mathcal{V}_{\theta},\mathcal{Q}_{\theta},\pi}$ denote the expectation under $\mathbb{P}_{\mathcal{V}_{\theta},\mathcal{Q}_{\theta},\theta}$ and $\mathbb{P}_{\mathcal{V}_{\theta},\theta}$, respectively. Finally, for any random variable $U$ such that $\sigma(U) \subset (\mathcal{F}_t)_{t=0}^{\infty}$, we denote by $\mathbb{E}_{U;\mathcal{V}_{\theta},\pi}$ the expectation of $U$ under $\mathbb{P}_{\mathcal{V}_{\theta},\pi}$. 
We are now ready to define our problem formally.

\subsection{Covert Active Hypothesis Testing (P1)}
In the context of sequential hypothesis testing, the stopping time $\tau$ should satisfy certain time budget constraints. In this paper, we consider a $\emph{probabilistic}$ time constraint, i.e., the probability that the stopping time exceed the budget $n\in\mathbb{N}^+$ decreases to $0$ asymptotically when $n\rightarrow\infty$,  
\begin{align}
    \lim_{n\rightarrow\infty}\max_{\theta\in\Theta}  \mathbb{P}_{\mathcal{V}_{\theta},\pi}(\tau>n) = 0. 
    \label{eqn:time_budget_1}
\end{align}
In addition to the time budget constraint, we also require the decision making policy $\pi$ to be \emph{covert} with respective to Willie. 
Before elaborating on the covertness constraint, we need to define the status of Alice formally. Let $\mathcal{I}=\{\textnormal{active},\textnormal{idle}\}$ be the set of status of Alice. Alice performs the decision making policy $\pi$ when she is active. When Alice is idle, the null decision making policy $\pi_0$ among which the control policy $\phi=\{\mathbf{1}(X_t=0)\}_{t=1}^{\infty}$ is applied so that the observation $Z_t$ is generated from the distributions $q_{\theta}^0$ for any $t\in\mathbb{N}^+$ when the true hypothesis is $\theta$. Note that the stopping rule $\varphi$ and the final decision rule $\psi$ of the null policy $\pi_0$ can be defined arbitrarily, and they do not affect our analysis. 
Then, the covertness constraint in (P1) is
\begin{align}
    \lim_{n\rightarrow \infty} \mathbb{D}(P_{Z^{n};\theta}\Vert (q_{\theta}^{0})^{\otimes n})\leq\eta \textnormal{ for all }\theta\in\Theta,
    \label{eqn:covertness}
\end{align}
where $P_{Z^n;\theta}(z^n) \triangleq \mathbb{P}_{\mathcal{V}_{\theta},\mathcal{Q}_{\theta},\pi}(z^n)$ for any $z^n\in\mathcal{Z}^n$, and $\eta$ is the parameter governing how covert the strategy should be. The relationship between the covertness constraint in \eqref{eqn:covertness} and Willie's capability to identify the decision making policy $\pi$ is discovered in Remark~\ref{rem:1}.

\begin{remark}
    \label{rem:1}
   For all $k\in\mathbb{N}^+$, let $\rho_k:\mathcal{Z}^k\mapsto\mathcal{I}$ be the decision function of Willie to determine whether Alice is active or not by using $k$ observations. Fix any $\theta\in\Theta$ and for any $k\in\mathbb{N}^+$, we define $\alpha_{\theta,k}\triangleq \mathbb{P}_{\mathcal{V}_{\theta},\mathcal{Q}_{\theta},\pi}(\rho_k(Z^k)=\textnormal{idle})$ and $\beta_{\theta,k}\triangleq \mathbb{P}_{\mathcal{V}_{\theta},\mathcal{Q}_{\theta},\pi_0}(\rho_k(Z^k)=\textnormal{active})$ as the two kinds of error probability. We require Willie's performance for identifying the active policy to be close to the performance of a random guess for all decision functions $\rho_k$ for all $k\in\mathbb{N}^+$. We ensure this by enforcing a lower bound on $\alpha_{\theta,k}+\beta_{\theta,k}$ for all $\theta\in\Theta$ and $k\in\mathbb{N}^+$. For all $k\leq n$,
   \begin{align}
       \alpha_{\theta,k} + \beta_{\theta,k} &= \mathbb{P}_{\mathcal{V}_{\theta},\mathcal{Q}_{\theta},\pi}(\rho_k(Z^k)=\textnormal{idle}) + \mathbb{P}_{\mathcal{V}_{\theta},\mathcal{Q}_{\theta},\pi_0}(\rho_k(Z^k)=\textnormal{active})\\
       &\geq 1 - \mathbb{V}(P_{Z^k;\theta}\Vert (q_{\theta}^0)^{\otimes k}) \label{eqn:rmk_1_0}\\
       &\geq 1 - \sqrt{\mathbb{D}\left(P_{Z^k;\theta}\middle\Vert (q_{\theta}^0)^{\otimes  k}\right)}\label{eqn:rmk_1_1}\\
       &\geq 1 -\sqrt{\mathbb{D}\left(P_{Z^n;\theta}\middle\Vert (q_{\theta}^0)^{\otimes  n}\right)} 
       \label{eqn:rmk_1_2},
   \end{align}
   where \eqref{eqn:rmk_1_0} follows from the definition of total variational distance, \eqref{eqn:rmk_1_1} follows from Pinsker's inequality, and \eqref{eqn:rmk_1_2} follows from the monotonicity of relative entropy. 
   Similarly, for all $k>n$, 
   \begin{align}
       \alpha_{\theta,k} + \beta_{\theta,k} &= \mathbb{P}_{\mathcal{V}_{\theta},\mathcal{Q}_{\theta},\pi}(\rho_k(Z^k)=\textnormal{idle}) + \mathbb{P}_{\mathcal{V}_{\theta},\mathcal{Q}_{\theta},\pi_0}(\rho_k(Z^k)=\textnormal{active})\\
       &\geq  \mathbb{P}_{\mathcal{V}_{\theta},\mathcal{Q}_{\theta},\pi}(\rho_k(Z^k)=\textnormal{idle},\tau\leq n) + \mathbb{P}_{\mathcal{V}_{\theta},\mathcal{Q}_{\theta},\pi_0}(\rho_k(Z^k)=\textnormal{active}),
    \end{align}
where 
    \begin{align}
        &\mathbb{P}_{\mathcal{V}_{\theta},\mathcal{Q}_{\theta},\pi}(\rho_k(Z^k)=\textnormal{idle},\tau\leq n) \nonumber\\
        &=
       \sum_{z^k}\sum_{(x^n,y^n)} \mathbb{P}_{\mathcal{V}_{\theta},\pi}\left(x^n,y^n\right) \mathbf{1}\left(\varphi(x^n,y^n)=\textnormal{stop}\right)\left(\prod_{i=1}^n q_{\theta}^{x_i}(z_i)\prod_{i=n+1}^k q_{\theta}^{0}(z_i) \right)\mathbf{1}(\rho_k(z^k)=\textnormal{idle})\nonumber\\
       &\geq \sum_{z^k}\sum_{(x^n,y^n)} \mathbb{P}_{\mathcal{V}_{\theta},\pi}\left(x^n,y^n\right)\left(\prod_{i=1}^n q_{\theta}^{x_i}(z_i)\prod_{i=n+1}^k q_{\theta}^{0}(z_i)\right)\mathbf{1}(\rho_k(z^k)=\textnormal{idle}) - \mathbb{P}_{\mathcal{V}_{\theta},\pi}(\tau>n). 
   \end{align}
   By defining $\check{P}_{Z^k;\theta}(z^k)=\sum_{(x^n,y^n)} \mathbb{P}_{\mathcal{V}_{\theta},\pi}\left(x^n,y^n\right)\left(\prod_{i=1}^n q_{\theta}^{x_i}(z_i)\prod_{i=n+1}^k q_{\theta}^{0}(z_i)\right)$ for any $z^k\in\mathcal{Z}^k$, we have 
   \begin{align}
       \alpha_{\theta,k} + \beta_{\theta,k} &\geq 1 - \sqrt{\mathbb{D}\left(\check{P}_{Z^k;\theta}\middle\Vert (q_{\theta}^0)^{\otimes k}\right)} - \mathbb{P}_{\mathcal{V}_{\theta},\pi}(\tau>n) \\
       &= 1 - \sqrt{\mathbb{D}\left(P_{Z^n;\theta}\middle\Vert (q_{\theta}^0)^{\otimes n}\right)} - \mathbb{P}_{\mathcal{V}_{\theta},\pi}(\tau>n) 
       \label{eqn:rmk_1_3}
   \end{align}
   for all $k>n$, where in \eqref{eqn:rmk_1_3} we use the fact that $\check{P}_{Z^{k};\theta}(z_n^k)=(q_{\theta}^0)^{\otimes k-n}(z_n^k)$ for any $z_n^k\in\mathcal{Z}^{k-n}$. 
   The probability $\mathbb{P}_{\mathcal{V}_{\theta},\pi}(\tau>n)$ should be arbitrarily small by \eqref{eqn:time_budget_1} when $n$ is sufficiently large. 
   Therefore, for all $k\in\mathbb{N}^+$, we have 
   \begin{align}
       \alpha_{\theta,k} + \beta_{\theta,k} &\geq 1 - \sqrt{\eta}
   \end{align}
   when $n$ is sufficiently large by applying our stopping time constraint in \eqref{eqn:time_budget_1} and the covertness constraint in \eqref{eqn:covertness}. 

\end{remark}

Fix any $\eta>0$, we define $\Lambda_1(\eta)$ as the sets of policies that satisfy the covertness constraint in \eqref{eqn:covertness} and the time budget constraint in \eqref{eqn:time_budget_1}, respectively. Then, given any time budget $n\in\mathbb{N}^+$ and any $\eta>0$, the error probability of any policy $\pi \in\Lambda_1(\eta)$ in (P1) is defined as 
\begin{align}
    P_{\textnormal{err},1}^{(n)}(\pi) \triangleq \max_{\theta\in\Theta} \mathbb{P}_{\mathcal{V}_{\theta},\pi}(\psi(Y^{\tau},X^{\tau})\neq \theta). 
\end{align}
Moreover, we define the error exponent as  
\begin{align}
    \gamma_{1}(\pi) \triangleq \liminf_{n\rightarrow \infty }\frac{-\log P_{\textnormal{err},1}^{(n)}(\pi)}{\sqrt{n}} .
    \label{eqn:def_1}
\end{align} 
\begin{definition}[Achievability in Covert Sequential Testing]
    For any $\eta>0$, we say that the exponent $r$ is achievable with $\eta$-covertness in (P1) if there exists a policy $\pi\in \Lambda_{1}(\eta)$ such that $\gamma_{1}(\pi)>r$. 
\end{definition}
The objective of (P1) is to analyze the supremum of all achievable exponents, i.e., 
\begin{align}
    \gamma^{*}_{1} &= \sup_{\pi \in \Lambda_{1}(\eta)}   \gamma_{1}(\pi). 
\end{align}

\subsection{ Covert Best-Arm Identification (P2)}
In (P2), the hypothesis $\theta\in\Theta$ is assumed known. Therefore, we drop the subscript $\theta$ in $\mathcal{V}_{\theta}$ and $\mathcal{Q}_{\theta}$ to simplify the notation, and we also use the notation $\nu_x$ and $q_{x}$ to represent $\nu_{\theta}^x$ and $q_{\theta}^x$, respectively.  
In contrast to (P1), in which a predefined time budget $n$ exists, we restrict the  probability of incorrectly identifying the optimal arm. Specifically, for each  $\delta>0$, 
we define 
\begin{align}
    P_{\textnormal{err},2}^{(\delta)}(\pi) \triangleq \mathbb{P}_{\mathcal{V},\pi}(\psi(Y^{\tau},X^{\tau})\neq x^*(\mathcal{V})). 
\end{align}
Then, the {\it confidence constraint} requires that 
\begin{align}
    P_{\textnormal{err},2}^{(\delta)}(\pi) \leq \delta,
    \label{eqn:confidence_constraint}
\end{align}
where $\delta>0$ is the predefined value. Fix some small $\kappa>0$, we also define 
\begin{align}
    \tau_{\textnormal{sup}}^{(\delta)} \triangleq \inf\{a\in\mathbb{R}: \mathbb{P}_{\mathcal{V},\pi}\left(\tau>a\right) < \kappa \} 
\end{align}
when the confidence constraint is $\delta$. The definition of $\tau_{\textnormal{sup}}^{(\delta)}$ is similar to essential supremum of $\tau$, except for the non-zero value of $\kappa$. 
Then, the {\it covertness constraint} in this problem is defined as
\begin{align}
    \lim_{\delta \rightarrow 0} \mathbb{D}\left(P_{Z^{\tau_{\textnormal{sup}}^{(\delta)}}}\middle\Vert (q_{}^{0})^{\otimes \tau_{\textnormal{sup}}^{(\delta)}}\right)\leq\eta,
    \label{eqn:covert_constraint_bandit}
\end{align}
where $\eta>0$ is some predefined covertness constraint. 
As discussed in Remark~\ref{rem:2}, the definition of $\tau^{(\delta)}_{\textnormal{sup}}$ and the covertness constraint allow us to analyze the performance of Willie's ability to identify the presence of a best arm identification policy.

\begin{remark}
        For any $k\in\mathbb{N}^+$, let $\rho_k:\mathcal{Z}^{k}\mapsto \mathcal{I}$ be the decision function of Willie to determine the existence of the policy by using $k$ observations. Fix the bandit $\mathcal{V}$, $\mathcal{Q}$ and the policy $\pi$ in (P2), we define $\alpha_k\triangleq \mathbb{P}_{\mathcal{V},\mathcal{Q},\pi}(\rho_k(Z^k)=\textnormal{idle})$ and $\beta_k\triangleq\mathbb{P}_{\mathcal{V},\mathcal{Q},\pi_0}(\rho_k(Z^k)=\textnormal{active})$ as the miss detection and the false alarm probability of Willie's detection policy $\rho_k$, respectively. We are able to lower bound $\alpha_k+\beta_k$ for any $k\in\mathbb{N}$ as done in Remark~\ref{rem:1}. Specifically, if $k\leq \tau^{(\delta)}_{\textnormal{sup}}$, 
        \begin{align}
            \alpha_{k} + \beta_{k} &= \mathbb{P}_{\mathcal{V},\mathcal{Q},\pi}(\rho_k(Z^k)=\textnormal{idle}) + \mathbb{P}_{\mathcal{V},\mathcal{Q},\pi_0}(\rho_k(Z^k)=\textnormal{active})\label{eqn:rmk_2_0}\\
            &\geq 1 - \mathbb{V}(P_{Z^k}\Vert (q_0)^{\otimes k})\\
            &\geq 1 - \sqrt{\mathbb{D}\left(P_{Z^k}\middle\Vert (q_0)^{\otimes  k}\right)}\label{eqn:rmk_2_1}\\
             &\geq 1 -\sqrt{\mathbb{D}\left(P_{Z^{\tau_{\textnormal{sup}}^{(\delta)}}}\middle\Vert (q_0)^{\otimes  {\tau_{\textnormal{sup}}^{(\delta)}}}\right)} 
            \label{eqn:rmk_2_2},
        \end{align}
    where \eqref{eqn:rmk_2_0} follows from the definition of total variational distance, \eqref{eqn:rmk_2_1} follows from Pinsker's inequality, and \eqref{eqn:rmk_2_2} follows from the monotonicity of relative entropy. 
    Similarly, for all $k>\tau_{\textnormal{sup}}^{(\delta)}$, 
   \begin{align}
       \alpha_{k} + \beta_{k} &\geq  \mathbb{P}_{\mathcal{V},\mathcal{Q},\pi}(\rho_k(Z^k)=\textnormal{idle},\tau\leq \tau_{\textnormal{sup}}^{(\delta)}) + \mathbb{P}_{\mathcal{V},\mathcal{Q},\pi_0}(\rho_k(Z^k)=\textnormal{active}),
    \end{align}
where 
    \begin{align}
        &\mathbb{P}_{\mathcal{V},\mathcal{Q},\pi}(\rho_k(Z^k)=\textnormal{idle},\tau\leq \tau_{\textnormal{sup}}^{(\delta)}) \nonumber\\
        &=
       \sum_{z^k}\sum_{(x^{\tau_{\textnormal{sup}}^{(\delta)}},y^{\tau_{\textnormal{sup}}^{(\delta)}})} \mathbb{P}_{\mathcal{V},\pi}\left(x^{\tau_{\textnormal{sup}}^{(\delta)}},y^{\tau_{\textnormal{sup}}^{(\delta)}}\right) \mathbf{1}\left(\varphi(x^{\tau_{\textnormal{sup}}^{(\delta)}},y^{\tau_{\textnormal{sup}}^{(\delta)}})=\textnormal{stop}\right)\nonumber\\
       &\hspace{5cm}\times \Bigg(\prod_{i=1}^{\tau_{\textnormal{sup}}^{(\delta)}} q_{x_i}(z_i)\prod_{i=\tau_{\textnormal{sup}}^{(\delta)}+1}^k q_{0}(z_i) \Bigg)\mathbf{1}(\rho_k(Z^k)=\textnormal{idle})\nonumber\\
       &\geq \sum_{z^k}\sum_{(x^{\tau_{\textnormal{sup}}^{(\delta)}},y^{\tau_{\textnormal{sup}}^{(\delta)}})} \mathbb{P}_{\mathcal{V},\pi}\left(x^{\tau_{\textnormal{sup}}^{(\delta)}},y^{\tau_{\textnormal{sup}}^{(\delta)}}\right)\left(\prod_{i=1}^{\tau_{\textnormal{sup}}^{(\delta)}} q_{x_i}(z_i)\prod_{i=\tau_{\textnormal{sup}}^{(\delta)}+1}^k q_{0}(z_i)\right)\mathbf{1}(\rho_k(Z^k)=\textnormal{idle}) \nonumber\\
       &\quad - \mathbb{P}_{\mathcal{V},\pi}(\tau>\tau_{\textnormal{sup}}^{(\delta)}). 
   \end{align}
   By defining $\check{P}_{Z^k}(z^k)=\sum_{(x^{\tau_{\textnormal{sup}}^{(\delta)}},y^{\tau_{\textnormal{sup}}^{(\delta)}})} \mathbb{P}_{\mathcal{V},\pi}\left(x^{\tau_{\textnormal{sup}}^{(\delta)}},y^{\tau_{\textnormal{sup}}^{(\delta)}}\right)\left(\prod_{i=1}^{\tau_{\textnormal{sup}}^{(\delta)}} q_{x_i}(z_i)\prod_{i=\tau_{\textnormal{sup}}^{(\delta)}+1}^k q_{0}(z_i)\right)$ for any $z^k\in\mathcal{Z}^k$, we have 
   \begin{align}
       \alpha_{k} + \beta_{k} &\geq 1 - \sqrt{\mathbb{D}\left(\check{P}_{Z^k}\middle\Vert (q_0)^{\otimes k}\right)} - \mathbb{P}_{\mathcal{V},\pi}(\tau>\tau_{\textnormal{sup}}^{(\delta)}) \\
       &= 1 - \sqrt{\mathbb{D}\left(P_{Z^{\tau_{\textnormal{sup}}^{(\delta)}}}\middle\Vert (q_0)^{\otimes \tau_{\textnormal{sup}}^{(\delta)}}\right)} - \mathbb{P}_{\mathcal{V},\pi}(\tau>\tau_{\textnormal{sup}}^{(\delta)}) 
       \label{eqn:rmk_2_3}\\
       &\geq 1 - \sqrt{\mathbb{D}\left(P_{Z^{\tau_{\textnormal{sup}}^{(\delta)}}}\middle\Vert (q_0)^{\otimes \tau_{\textnormal{sup}}^{(\delta)}}\right)} - \kappa 
       \label{eqn:rmk_2_4}
   \end{align}
   for all $k>\tau_{\textnormal{sup}}^{(\delta)}$, where in \eqref{eqn:rmk_2_3} we use the fact that $\check{P}_{Z^{k}}(z_{\tau_{\textnormal{sup}}^{(\delta)}}^k)=(q_0)^{\otimes k-\tau_{\textnormal{sup}}^{(\delta)}}(z_{\tau_{\textnormal{sup}}^{(\delta)}}^k)$ for any $z_{\tau_{\textnormal{sup}}^{(\delta)}}^k\in\mathcal{Z}^{k-\tau_{\textnormal{sup}}^{(\delta)}}$, and \eqref{eqn:rmk_2_4} follows from the definition of $\tau_{\textnormal{sup}}^{(\delta)}$. Therefore, for all $k\in\mathbb{N}^+$, we have 
   \begin{align}
       \alpha_k + \beta_k \geq 1 - \sqrt{\eta} - \kappa
   \end{align}
   when $\delta$ is sufficiently small by the covertness constraint in \eqref{eqn:covert_constraint_bandit}. 
    \label{rem:2}
\end{remark}
We also denote by $\Lambda_2(\eta)$ the set of decision making policies that satisfy \eqref{eqn:confidence_constraint} and \eqref{eqn:covert_constraint_bandit}. 
For any policy $\pi\in\Lambda_2(\eta)$, the error exponent is then defined as
\begin{align}
    \gamma_{2}(\pi) \triangleq \liminf_{\delta\rightarrow 0} \frac{-\log \delta}{\sqrt{\tau_{\textnormal{sup}}^{\delta}}}
\end{align}
\begin{definition}[Achievability in Covert Best Arm Identification]
    For any $\eta>0$, we say that the exponent $r$ is achievable with $\eta$-covertness in (P2) if there exists a policy $\pi\in \Lambda_{2}(\eta)$ such that $\gamma_{2}(\pi)>r$. 
\end{definition}

The objective of (P2) is to analyze the supremum of all achievable exponents, i.e., 
\begin{align}
    \gamma^{*}_{2} &= \sup_{\pi \in \Lambda_{2}(\eta)}   \gamma_{2}(\pi). 
\end{align}

\begin{remark}
    Our problem formulation in (P2) belongs to the fixed confidence setting for best arm identification. However, instead of analyzing the averaged stopping time \cite{Gariver2016}, the exponent is defined as the asymptotic ratio between $-\log \delta$ and the square-root of $\tau_{\textnormal{sup}}^{(\delta)}$ when $\delta\rightarrow 0$. We define the exponent in this way because $\tau_{\textnormal{sup}}^{(\delta)}$ is the quantity that helps us define a meaningful covertness constraint as mentioned in Remark~\ref{rem:2}.
    \label{rem:3}
\end{remark}

\section{Main Results}
\label{sec:main_results}

\subsection{Main Results for Covert Active Hypothesis Testing}
Our first theorem gives a lower bound on the optimal exponent $\gamma_1^*$.  
\begin{theorem}
Let $\Theta$ be the set of parameters that are indistinguishable from another parameter by choosing the null action $0$, i.e. $\mathbb{D}(\nu_{\theta}^0\Vert \nu_{\theta'}^0)=0$ for all $\theta\neq \theta'$. For all $\theta\in\Theta$, we assume that no distribution $\Bar{P}_X$ over $\mathcal{X}\setminus\{0\}$ is such that $\sum_{x\neq 0}\Bar{P}_X(x)q_{\theta}^x = q_{\theta}^0$. Then, 
we have
\begin{align}
    \gamma_{1}^* \geq \sqrt{2\eta} \min_{\theta\in\Theta} \max_{\Bar{P}_X\in\mathcal{P}_{\mathcal{X}\setminus\{0\}}}\min_{\theta'\neq\theta}  \frac{\sum_{x\neq 0}\Bar{P}_X(x)\mathbb{D}(\nu_{\theta}^x\Vert \nu_{\theta'}^x)}{\sqrt{\chi_2\left(\sum_{x\neq 0}\Bar{P}_X(x)q_{\theta}^x\Vert q_{\theta}^0\right)}}.
    \label{eqn:main_result1}
\end{align}
\label{thm:main_result1}
\end{theorem}
We illustrate Theorem~\ref{thm:main_result1} with the following example. Let $\mathcal{X}=\{0,1,2\}$, $\mathcal{Y}=\{0,1\}$, and $\Theta=\{a,b,c\}$ so that $\{\nu_{\theta}^x\}_{x\in\mathcal{X},\theta\in\Theta}$ and $\{q_{\theta}^x\}_{x\in\mathcal{X},\theta\in\Theta}$ are sets of Bernoulli distributions, where the parameters of Bernoulli distributions are given in the following tables.

\begin{table}[htp]
\centering
\caption{$\nu_{\theta}^x(1)$ for all $x\in\mathcal{X}$ and $\theta\in\Theta$}
\begin{tabular}{|c|c|c|c|c|}
\hline
\backslashbox{$\theta$}{$x$}  & 0                        & 1                        & 2                                            \\ \hline
a & 0.0                      & 0.9                      & 0.6                             \\ \hline
b & 0.0                      & 0.9                      & 0.9                             \\ \hline
c & 0.0 & 0.6 & 0.9 \\ \hline
\end{tabular}
\label{table:1}
\end{table}

\begin{table}[htp]
\centering
\caption{$q_{\theta}^x(1)$ for all $x\in\mathcal{X}$ and $\theta\in\Theta$}
\begin{tabular}{|c|c|c|c|c|}
\hline
\backslashbox{$\theta$}{$x$}  & 0                        & 1                        & 2                                            \\ \hline
a & 0.0                      & 0.6                      & 0.9                             \\ \hline
b & 0.0                      & 0.6                      & 0.9                             \\ \hline
c & 0.0 & 0.6 & 0.9 \\ \hline
\end{tabular}
\label{table:2}
\end{table}
Note that the distributions given by Table~\ref{table:1} and Table~\ref{table:2} satisfy the assumption that hypotheses cannot be distinguished by the null action $0$ and for all $\theta\in\Theta$, and   there is no distribution $\Bar{P}_X\in\mathcal{P}_{\mathcal{X}\setminus\{0\}}$ such that $\sum_{x\in\mathcal{X}\setminus\{0\}}\Bar{P}_X(x) q_{\theta}^x = q_{\theta}^0$. When there is no covertness constraint, the optimal value of the exponent $\gamma^{\#}_1(\pi) \triangleq \lim_{n\rightarrow \infty}\frac{-\log P_{\textnormal{err},1}^{(n)}(\pi)}{n}$ defined in \cite{Li2020,Nitinawarat2013} is known to be 
\begin{align}
    \min_{\theta\in\Theta}\max_{\Bar{P}_X\in\mathcal{P}_{\mathcal{X}\setminus\{0\}}} \sum_{x\in\mathcal{X}\setminus\{0\}} \Bar{P}_X(x) \min_{\theta'\neq \theta} \mathbb{D}(\nu_{\theta}^x\Vert \nu_{\theta'}^x),
    \label{eqn:optimal_exponent_wo_covert}
\end{align}
where the policy $\pi$ and $P_{\textnormal{err},1}^{(n)}(\pi)$ are defined similarly in Section~\ref{sec:model} but without covertness constraint. Note that one can choose the actions $x=2$ and $x=1$ to distinguish the state $a$ and the state $b$ from others, respectively, but there is no single action that  distinguishes the state $b$ from other states. Specifically, 
the arguments of $\theta$ and $\Bar{P}_X$ that solve the min-max optimization in \eqref{eqn:optimal_exponent_wo_covert} are given by $\theta=b$ and $\Bar{P}_{X}(x) = \frac{1}{2} \mathbf{1}(x=1) + \frac{1}{2} \mathbf{1}(x=2)$. In contrast, with a covertness constraint, the arguments of $\theta$ and $\Bar{P}_X$ that solve the min-max statement in \eqref{eqn:main_result1} are $\theta=b$ and $\Bar{P}_X(x) = 0.67\times \mathbf{1}(x=1) + 0.33\times \mathbf{1}(x=2)$, in which the policy has a higher probability to choose $x=1$. The phenomenon comes from the fact that $x=1$ is the action that makes the output distribution of Willie more closely resemble the distribution generated by the null action, i.e., $\chi_2(q_{\theta}^{1}\Vert q_{\theta}^0) < \chi_2(q_{\theta}^{2}\Vert q_{\theta}^0)$ for all $\theta\in\Theta$. 

We next provide an upper bound on $\gamma_1(\pi)$ when the policy $\pi$ satisfies additional assumptions. Specifically, we assume that for each time $t\in\mathbb{N}^+$, the action $X_t$ is generated from a distribution that is a function of the ML estimate $\hat{\theta}_{\textnormal{ML}}(t-1)\triangleq \argmax_{\theta'\in\Theta} \mathbb{P}_{\mathcal{V}_{\theta'},\pi}(X^{t-1},Y^{t-1})$ and of the stopping decision $S_{t-1}$. Specifically, for each $t\in\mathbb{N}^+$, we assume that 
\begin{align}
    P_{X_t|X^{t-1},Y^{t-1}}(x) &=
    \begin{cases}
        P_{X;\hat{\theta}_{\textnormal{ML}(t-1)}}(x) \quad \textnormal{ if } S_{t-1}\neq \textnormal{stop}\\
      \mathbf{1}(x=0) \quad \textnormal{ if } S_{t-1}= \textnormal{stop}
    \end{cases}
    \label{eqn:control_policy_assumption_1}
\end{align}
for any $x\in\mathcal{X}$ for some $\{P_{X;\theta}\}_{\theta\in\Theta}$ so that 
$P_{X_t|X^{t-1},Y^{t-1}}=P_{X_t|\hat{\theta}_{\textnormal{ML}}(t-1),S_{t-1}}$. 

\begin{theorem}
Let $\Theta$ be the set of parameters that are indistinguishable from another parameter by choosing the null action $0$, i.e. $\mathbb{D}(\nu_{\theta}^0\Vert \nu_{\theta'}^0) = 0$ for all $\theta\neq \theta'$. For all $\theta\in\Theta$, we assume that no distribution $\Bar{P}_X$ over $\mathcal{X}\setminus\{0\}$ is such that $\sum_{x\in \mathcal{X}\setminus\{0\}}\Bar{P}_X(x)q_{\theta}^{x} = q_{\theta}^0$. For all decision making policies $\pi\in\Lambda_1(\eta)$ in which the control policy at each time $t\in\mathbb{N}^+$ satisfies \eqref{eqn:control_policy_assumption_1}, then 
\begin{align}
    \gamma_{1}(\pi) \leq \sqrt{2\eta} \min_{\theta\in\Theta} \max_{\Bar{P}_X\in\mathcal{P}_{\mathcal{X}\setminus\{0\}}}\min_{\theta'\neq\theta}  \frac{\sum_{x\neq 0}\Bar{P}_X(x)\mathbb{D}(\nu_{\theta}^x\Vert \nu_{\theta'}^x)}{\sqrt{\chi_2\left(\sum_{x\neq 0}\Bar{P}_X(x)q_{\theta}^x\Vert q_{\theta}^0\right)}}.
\end{align}
\label{thm:main_result2}
\end{theorem}

\subsection{Main Results for Covert Best Arm Identification}
The exponent $\gamma_2^*$ is lower bounded by the following Theorem. 
\begin{theorem}
    Fix any bandit $\mathcal{V}$ in (P2). 
    If there is no distribution $\Bar{P}_X$ over $\mathcal{X}\setminus\{0\}$ such that $\sum_{x\in\mathcal{X}\setminus\{0\}}\Bar{P}_X(x)q_x = q_0$, then 
we have
\begin{align}
    \gamma_{2}^* \geq \sqrt{2\eta}\max_{\Bar{P}_X\in \mathcal{P}_{\mathcal{X}\setminus\{0\}}}\frac{\min_{\nu'\in\mathcal{E}_{\text{Alt}}(\mathcal{V}_{})}\sum_{x\in \mathcal{X}\setminus\{0\}}\Bar{P}_X(x)\mathbb{D}(\nu_{x}\Vert \nu_{x}')}{\sqrt{\chi_2(\sum_{x\in\mathcal{X}\setminus\{0\}}\Bar{P}_X(x)q_x\Vert q_0)}}, 
    \label{eqn:main_result3}
\end{align}
where $\mathcal{E}_{\text{Alt}}(\mathcal{V}) = \{\mathcal{V}'\triangleq \{\nu_x'\}_{x\in\mathcal{X}}\in\mathcal{E}_{\mathcal{N}}: x^*(\mathcal{V})\cap x^{*}(\mathcal{V}')=\emptyset\}$. 
\label{thm:main_result3}
\end{theorem}
In this work, we assume that $\mathcal{V}$ is a Gaussian bandit with variance $1$. Without loss of generality we assume that $\nu_1$ has the largest mean among all arms, i.e., $\mu(\nu_1)\geq \mu(\nu_x)$ for all $x\neq 1$, then for all $\Bar{P}_X\in\mathcal{P}_{\mathcal{X}\setminus\{0\}}$, 
\begin{align}
    \min_{\nu'\in\mathcal{E}_{\text{Alt}}(\mathcal{V}_{})}\sum_{x\in \mathcal{X}\setminus\{0\}}\Bar{P}_X(x)\mathbb{D}(\nu_{x}||\nu_{x}') = \frac{1}{2}\min_{x\in \mathcal{X}\setminus\{0,1\}} \frac{\Bar{P}_X(x)\Bar{P}_X(1)(\mu(\nu_1)-\mu(\nu_x)^2)}{\Bar{P}_X(1) + \Bar{P}_X(x)} 
    \label{eqn:eqn:avg_divergence_expression}
\end{align}
by \cite[Problem 33.4.(a)]{Lattimore2020}, and 
\begin{align}
    \chi_2\left(\sum_{x\in\mathcal{X}\setminus\{0\}}\Bar{P}_X(x)q_x\middle\Vert q_0\right) = e^{\left(\sum_{x\in\mathcal{X}\setminus\{0\}}\Bar{P}_X(x) \mu(q_x)\right)^2} - 1,
\end{align}
where we have use the fact that $\mu(q_0)=0$. Therefore, \eqref{eqn:main_result3} can be simplified as 
\begin{align}
    \gamma_2^* \geq \frac{\sqrt{2\eta}}{2} \max_{\Bar{P}_X\in \mathcal{P}_{\mathcal{X}\setminus\{0\}}} \frac{\min_{x\in \mathcal{X}\setminus\{0,1\}} \frac{\Bar{P}_X(x)\Bar{P}_X(1)(\mu(\nu_1)-\mu(\nu_x)^2)}{\Bar{P}_X(1) + \Bar{P}_X(x)}}{\sqrt{e^{\left(\sum_{x\in\mathcal{X}\setminus\{0\}}\Bar{P}_X(x) \mu(q_x)\right)^2} - 1}}. 
    \label{eqn:exponent_BAI_simplified}
\end{align}
We provide an example to illustrate how the covertness constraint affects the optimal strategy. 
Let the bandits $\mathcal{V}$ and $\mathcal{Q}$ be as given in Table~\ref{table:3}. It is shown in \cite{Gariver2016,Lattimore2020} that the optimal value of the exponent $\gamma_2^{\#}(\pi) \triangleq \lim_{\delta \rightarrow 0} \frac{ - \log \delta}{\mathbb{E}_{\mathcal{V},\pi}[\tau]}$ in the conventional setting without covertness constraint is given by 
\begin{align}
    \max_{\Bar{P}_X\in \mathcal{P}_{\mathcal{X}\setminus\{0\}}} \inf_{\mathcal{V}'\in\mathcal{V}_{\textnormal{Alt}(\mathcal{V})}} \sum_{x\in\mathcal{X}\setminus\{0\}}\Bar{P}_X(x)\mathbb{D}(\nu_x\Vert \nu_x').
    \label{eqn:exponent_of_bandit_wo_covert}
\end{align}

\begin{table}[htp]
\centering
\caption{Means of Gaussian bandits $\mathcal{V}$ and $\mathcal{Q}$ }
\begin{tabular}{|c|c|c|c|c|}
\hline
   & x = 0                        & x = 1                        & x = 2                                            \\ \hline
$\mu(\nu_x)$ &  0                      & 1                      & 0.5                             \\ \hline
$\mu(q_x)$ & 0                      & 1                      & 0.5                             \\ \hline
\end{tabular}
\label{table:3}
\end{table}

In our example, the distribution $\Bar{P}_X$ that maximizes \eqref{eqn:exponent_of_bandit_wo_covert} is $\Bar{P}_X(x) = \frac{1}{2}\mathbf{1}(x=1) + \frac{1}{2}\mathbf{1}(x=2)$. In contrast, the distribution $\Bar{P}_X$ that maximizes \eqref{eqn:exponent_BAI_simplified} is $\Bar{P}_X(x) = 0.3\times \mathbf{1}(x=1) + 0.7\times \mathbf{1}(x=2)$, incurred by the fact that pulling arm $1$ makes it easier for Willie to detect the existence of a policy. 

We now provide a upper bound on $\gamma_2^*$ when $\pi$ satisfies additional assumptions. 
Fix any confidence constraint $\delta>0$ and assume that 
\begin{align}
    P_{X_t|X^{t-1},Y^{t-1}}(x) =
    \begin{cases}
        P_{\hat{\mathcal{V}}(t-1)}(x) \quad \textnormal{ if  }S_{t-1} = \textnormal{continue}\\
     \mathbf{1}(x=0) \quad \textnormal{ if  }S_{t-1} = \textnormal{stop}
    \end{cases}
    \label{eqn:control_policy_assumption_3}
\end{align}
for all $x\in\mathcal{X}$ and for all $t\in\mathbb{N}^+$, i.e, the distribution of $X_t$ is a function of the estimated bandit $\hat{\mathcal{V}}(t-1)$ and the stopping decision $S_{t-1}$. We also assume that $\{P_{\mathcal{V}'}\}_{\mathcal{V}'}$ is a set of  distributions on $\mathcal{X}$ that are
continuous functions of the bandit, and that the probability of choosing any $x\in\mathcal{X}\setminus\{0\}$ decreases approximately with the same speed with $|\log \delta|$, i.e., there exists some $0<D<\infty$ such that 
\begin{align}
    \left|\frac{P_{\mathcal{V}'}(x)}{P_{\mathcal{V}''}(x)} - 1\right| \leq D||\mathcal{V}'-\mathcal{V}''||_{\infty} \quad \textnormal{ and } \quad \left|\frac{P_{\mathcal{V}'}(x)}{P_{\mathcal{V}'}(x')} - 1\right| \leq D
    \label{eqn:control_policy_assumption_4}
\end{align}
for any $x,x'\in\mathcal{X}\setminus\{0\}$, $\mathcal{V}'$,  $\mathcal{V}''$ and any confidence level $\delta>0$, where for any $\mathcal{V}=\{\nu'_x\}_{x\in\mathcal{X}}$ and  $\mathcal{V}'=\{\nu''_x\}_{x\in\mathcal{X}}$ the infinity norm $\Vert \mathcal{V}'-\mathcal{V}''\Vert_{\infty} $ is defined as 
\begin{align}
    \Vert \mathcal{V}'-\mathcal{V}''\Vert_{\infty} \triangleq \max_{x\in \mathcal{X}} |\mu(\nu_x') - \mu(\nu_x'') |. 
\end{align}
Note that $\{P_{\mathcal{V}'}\}_{\mathcal{V}'}$ depends on the confidence level $\delta$ implicitly, and \eqref{eqn:control_policy_assumption_4} implies that there exists some $\alpha>0$ such that
\begin{align}
    \max_{x\neq 0} P_{\mathcal{V}'}(x) = \Tilde{\Theta}_{\delta\rightarrow 0}(|\log \delta|^{-\alpha}) \quad \textnormal{ and }\quad 
    \min_{x\neq 0} P_{\mathcal{V}'}(x) = \Tilde{\Theta}_{\delta\rightarrow 0}(|\log \delta|^{-\alpha})
    \label{eqn:control_policy_assumption_5}
\end{align}
for all $\mathcal{V}'$. Finally, we also assume that the stopping time $\tau$ of the policy $\pi$ concentrates, i.e., 
\begin{align}
    \lim_{\delta\rightarrow 0} \mathbb{P}_{\mathcal{V},\pi}\left(\left|\tau - \mathbb{E}_{\mathcal{V},\pi}[\tau] \right| \geq \epsilon \mathbb{E}_{\mathcal{V},\pi}[\tau]\right) = 0
    \label{eqn:tau_concentrate}
\end{align}
for all $\mathcal{V}$ and $\epsilon>0$. 
\begin{theorem}
    Fix any bandit $\mathcal{V}$ in (P2). 
    For all decision making policies $\pi\in\Lambda_2(\eta)$ in which the control policy at each time $t\in\mathbb{N}^+$ satisfies the assumption in \eqref{eqn:control_policy_assumption_3}, \eqref{eqn:control_policy_assumption_4}, \eqref{eqn:control_policy_assumption_5} and \eqref{eqn:tau_concentrate}, we have 
\begin{align*}
    \gamma_{2}(\pi) \leq \sqrt{2\eta}\max_{\Bar{P}_X\in\mathcal{X}\setminus\{0\}}\frac{\min_{\mathcal{V}'\in\mathcal{E}_{\text{Alt}}(\mathcal{V}_{})}\sum_{x\in \mathcal{X}\setminus\{0\}}\Bar{P}_X(x)\mathbb{D}(\nu_{x}||\nu_{x}')}{\sqrt{\chi_2(\sum_{x\in\mathcal{X}\setminus\{0\}}\Bar{P}_X(x)q_x||q_0)}}. 
\end{align*}
\label{thm:main_result4}
\end{theorem}

\section{Mathematical Tools}
\label{sec:math_tool}

\begin{lemma}[Bretagnolle–Huber's Inequality \cite{Bretagnolle}]
Let $P$ and $P'$ be two different probability measure on a common sigma algebra $\Sigma$ and the sample space $\Omega$. Let $\mathcal{E}\in\Sigma$, then it holds that 
\begin{align}
    P(\mathcal{E}^{c}) + P'(\mathcal{E}) \geq \frac{1}{2}\exp\left(-\mathbb{D}(P\Vert P')\right),
\end{align}
where $\mathcal{E}^{c}=\Omega\setminus\mathcal{E}$.
\end{lemma}

\begin{lemma}[Relative Entropy Decomposition Lemma \cite{Lattimore2020}] Let $\mathcal{V}=\{\nu_x\}_{x\in\mathcal{X}}$ and $\mathcal{V}'=\{v_x'\}_{x\in\mathcal{X}}$ be two sets of distributions on $\mathcal{Y}$. Given some policy $\pi$, for each $t\in\mathbb{N}$, the distributions of $X_t$ and $Y_t$ follow the structure defined in Section~\ref{sec:model}. If we define $\mathcal{V}_{\pi}$ and $\mathcal{V}_{\pi}'$ as the distributions on $\mathcal{X}^{\infty}\times \mathcal{Y}^{\infty}$ under the policy $\pi$ when the set of distributions on $\mathcal{Y}$ is given by $\mathcal{V}$ and $\mathcal{V}'$, respectively, then 
\begin{align}
    \mathbb{D}(\mathcal{V}_{\pi}\Vert \mathcal{V}'_{\pi}) = \sum_{i\in\mathcal{X}}\mathbb{E}[T_x(\tau)]\mathbb{D}(\nu_x\Vert \nu_x'),
\end{align}
where $T_x(\tau) = \sum_{t=1}^{\infty}\mathbf{1}(X_t=x,S_{t-1}\neq \textnormal{stop})$. 

\end{lemma}

\begin{lemma}[Bernstein's Inequality \cite{Devroye2001}] 
    Let $U_1, \cdots, U_n$ be independent zero mean random variables. Suppose that $|U_i|\leq M$ almost surely for all $i\in[1;n]$. Then, for any $\mu >0$, it holds that 
    \begin{align}
        \mathbb{P}\left(\sum_{i=1}^n U_i  \geq \mu \right) \leq \exp\left(-\frac{\frac{1}{2}\mu^2}{\sum_{i=1}^n \mathbb{E}[U_i^2]+\frac{1}{3}M\mu}\right). 
    \end{align}
\end{lemma}

\begin{lemma}[Freedman's Inequality \cite{Freedman1975,Joel2011}]
    {\it 
    Consider a real valued martingales $\{V_i\}_{i=0}^{\infty}$ adapted to the filtration $\{\mathcal{F}_i\}_{i=0}^{\infty}$ with difference sequence $\{U_i\}_{i=0}^{\infty}$. Assume that the difference sequence is uniformly bounded, i.e., 
    \begin{align*}
        U_i \leq C \quad \textnormal{ \it for all }i. 
    \end{align*}
    For any $k\in\mathbb{N}$, we define the predictable quadratic variation process of the martingale as 
    \begin{align*}
        W_k = \sum_{i=1}^k \mathbb{E}[U_i^2|\mathcal{F}_{i-1}].
    \end{align*}
    Then, for all $\delta>0$ and $\sigma^2>0$, 
    \begin{align*}
        \mathbb{P}\left(\exists k\geq 0: V_k \geq \delta \textnormal{ and }W_k\leq \sigma^2\right) \leq \exp\left(-\frac{\delta^2}{\sigma^2+C\delta/3}\right). 
    \end{align*}
    }
    \label{lem:Freedman}
\end{lemma}

\section{Proof of Theorem~\ref{thm:main_result1}}
\label{sec:prof_thm1}
\subsection{Construction of Covert Policy}
\label{sec:policy}
For any $\theta\in\Theta$ and $t\in\mathbb{N}^{+}$, we define the generalized log likelihood ratio under any policy $\pi\in\Lambda_1(\eta)$ as 
\begin{align}
    A_{\theta}(t) = \log \frac{\mathbb{P}_{\mathcal{V}_{\theta},\pi}(X^t,Y^t)}{\max_{\theta'\neq \theta} \mathbb{P}_{\mathcal{V}_{\theta'},\pi}(X^t,Y^t)}. 
\end{align}
Similarly, for any $\theta\in\Theta$, $\theta'\neq \theta$ and $t\in\mathbb{N}^+$, we define the ordinary log likelihood ratio as
\begin{align}
    A_{\theta,\theta'}(t) = \log \frac{\mathbb{P}_{\mathcal{V}_{\theta},\pi}(X^t,Y^t)}{ \mathbb{P}_{\mathcal{V}_{\theta'},\pi}(X^t,Y^t)}.
\end{align}
We then specify the policy $\pi$ in the achievability proof as follows. 
\paragraph{Stopping rule $\phi$}
The policy stops at $t\in\mathbb{N}^+$ if there exists some $\theta\in\Theta$ such that the log liklihood ratio $A_{\theta,\theta'}(t)$ is greater than a threshold $\Gamma_{\theta,\theta'}$ for all $\theta'\neq \theta$. The stopping time is therefore  
\begin{align}
    \tau = \inf \left\{t\in\mathbb{N}^+: \exists\:\theta\in\Theta\textnormal{ s.t. }  A_{\theta,\theta'}(t)\geq \Gamma_{\theta,\theta'} \:\forall\: \theta'\neq \theta \right\},
    \label{eqn:stop_condition}
\end{align}
where for all $\theta\in\Theta$ and $\theta'\neq \theta$, the threshold $\Gamma_{\theta,\theta'}$ is defined as 
\begin{align}
    \Gamma_{\theta,\theta'} = n\alpha_{\theta}\left(\sum_{x\neq 0}\Bar{P}_{X;\theta}(x)\mathbb{D}(\nu_{\theta}^x\Vert \nu_{\theta'}^x)-\zeta\right),
    \label{eqn:stop_threshold}
\end{align}
\begin{align}
    \Bar{P}_{X;\theta} = \argmax_{\Bar{P}_X\in\mathcal{P}_{\mathcal{X}\setminus \{0\}}} \min_{\theta''\neq\theta} \frac{\sum_{x\in\mathcal{X}\setminus \{0\}} \Bar{P}_X(x)\mathbb{D}(\nu_{\theta}^x \Vert \nu_{\theta''}^x)}{\sqrt{\chi_2\left(\sum_{x\in\mathcal{X}\setminus\{0\}}\Bar{P}_X(x)q_{\theta}^x\middle\Vert q_{\theta}^0\right)}}
    \label{eqn:def_of_P_theta}
\end{align}
and 
\begin{align}
    \alpha_{\theta} = \frac{\sqrt{2\eta}}{\sqrt{n}}\frac{1}{\sqrt{\chi_2\left(\sum_{x\in\mathcal{X}\setminus\{0\}}\Bar{P}_{X;\theta}(x)q_{\theta}^x\middle\Vert q_{\theta}^0\right)}},
    \label{eqn:def_of_alpha}
\end{align}
and $\zeta>0$ is some small value. 

\paragraph{Control policy $\varphi$} 

tFor any $t\in\mathbb{N}^+$, we define $\hat{\theta}_{\textnormal{ML}}(t)$ as the maximum likelihood estimate of the hypothesis $\theta$, i.e., $\hat{\theta}_{\textnormal{ML}}(t) = \argmax_{\theta'\in\Theta} \prod_{i=1}^t \nu_{\theta'}^{X_i}(Y_i)$. Moreover, for all $\theta\in\Theta$, we define the distribution $P_{X;\theta}\in\mathcal{P}_{\mathcal{X}}$ as 
\begin{align}
    P_{X;\theta}(x) = 
    \begin{cases}
        1 - \alpha_{\theta} \quad \textnormal{if }x=0\\
        \alpha_{\theta} \Bar{P}_{X;\theta}(x) \quad \textnormal{if }x\neq 0
    \end{cases}
\end{align}
for all $x\in\mathcal{X}$, where $\Bar{P}_{X;\theta}$ and $\alpha_{\theta}$ are defined in \eqref{eqn:def_of_P_theta} and \eqref{eqn:def_of_alpha}, respectively. 
Then, the control policy $P_{X_t|X^{t-1},Y^{t-1}}$ is given by 
\begin{align}
    P_{X_t|Y^{t-1},X^{t-1}}(x) 
    = 
    \begin{cases}
        P_{X;\hat{\theta}_{\textnormal{ML}}(t-1)}(x) \quad \textnormal{if } S_{t-1} \neq \textnormal{stop}\\
        \mathbf{1}(x=0) \quad \textnormal{if } S_{t-1} = \textnormal{stop}
    \end{cases}
    \label{eqn:control_policy_1}
\end{align}
for all $x\in\mathcal{X}$. 

\paragraph{Final decision rule $\psi$} When the decision making process stops, the estimated hypothesis is given by $\psi(X^{\tau},Y^{\tau}) = \hat{\theta}_{\textnormal{ML}}(\tau)$.

For convenience of analysis, 
for any policy $\pi$ satisfying \eqref{eqn:control_policy_assumption_1}, 
we also define a corresponding dummy policy $\tilde{\pi}$ that never stops and has the control policy 
\begin{align}
    P_{X_t|Y^{t-1},X^{t-1}}(x) = P_{X;\hat{\theta}_{\textnormal{ML}}(t-1)}(x)
    \label{eqn:dummy_control_policy}
\end{align}
for all $x$ and $t\in\mathbb{N}^+$.

\subsection{Analysis of Covert Policy}
\paragraph{Stopping time $\tau$}
For any $\theta\in\Theta$, $\theta'\neq \theta$ and $i\in\mathbb{N}^+$, we define  
\begin{align}
    L_{\theta,\theta'}(i) = \log \frac{\mathbb{P}_{\mathcal{V}_{\theta},\pi}(Y_i|X_i)}{\mathbb{P}_{\mathcal{V}_{\theta'},\pi}(Y_i|X_i)} = \log \frac{\nu_{\theta}^{X_i}(Y_i)}{\nu_{\theta'}^{X_i}(Y_i)}
\end{align}
so that 
\begin{align*}
    A_{\theta,\theta'}(t) = \sum_{i=1}^t L_{\theta,\theta}(i) 
\end{align*}
for any $t\in\mathbb{N}^+$. 
We also define the random variable $N_{\theta}$ as the earliest time such that the ML estimation about the hypothesis is correct for all $t> N_{\theta}$ when the true hypothesis is $\theta\in\Theta$, i.e., 
\begin{align}
    N_{\theta} = \sup\{t\in\mathbb{N}^+:\hat{\theta}_{\textnormal{ML}}(t) \neq \theta \}. 
\end{align}
The standard analysis of the stopping time $\tau$ in the sequential hypothesis testing given in the literature \cite{Nitinawarat2013, Li2022} relies on the fact that the estimate of the true hypothesis is incorrect for only finitely many time steps. 
However, this is not true in our setting because effective actions are selected with a probability shrinking with $n$, i.e., $\alpha_{\theta}=\Theta_{n\rightarrow\infty}(n^{-1/2})$ for all $\theta\in\Theta$.   
Nevertheless, we can still show that $N_{\theta}$ grows much slower than $n$ through the following lemma, proved in Appendix~\ref{apx:A}. 

\begin{lemma}
    Let $\pi$ be some policy that satisfies \eqref{eqn:control_policy_assumption_1}, and $\sum_{x\in\mathcal{X}\setminus\{0\}} P_{X;\theta'}(x) = \Tilde{\Theta}_{n\rightarrow \infty}(n^{-\alpha})$ for all $\theta'\in\Theta$ for some $0<\alpha<1$. Then, 
    \begin{align}
       \mathbb{P}_{\mathcal{V}_{\theta},\Tilde{\pi}}\left(N_{\theta} \geq n^{\alpha+\epsilon}\right) = O_{n\rightarrow\infty}(n^{-\beta})
    \end{align}
    \label{lem:1}
    for arbitrarily large $\beta>0$, $\theta\in\Theta$ and $\epsilon>0$, where $\Tilde{\pi}$ is the corresponding dummy policy of $\pi$ with control policy defined in \eqref{eqn:dummy_control_policy}. 
\end{lemma}
By inspecting \eqref{eqn:def_of_alpha} and \eqref{eqn:control_policy_1}, the policy constructed in this section satisfies the assumptions in Lemma~\ref{lem:1} with $\alpha=1/2$. 
The proof of Lemma~\ref{lem:1} requires Lemma~\ref{lem:Freedman}, proved by Freedman in \cite{Freedman1975}, which is a Bernstein-style concentration bound on martingales. 
Lemma~\ref{lem:1} states that the ML estimate $\hat{\theta}_{\textnormal{ML}}(t)$ of the hypothesis is correct for any $t\geq n^{1/2+\epsilon}$ with high probability under the dummy policy $\Tilde{\pi}$ of $\pi$, where $\pi$ is the policy defined in Section~\ref{sec:policy}. 

We are now ready to upper bound the  probability that $\tau$ is greater than $n$ as follows. 
\begin{align}
    &\mathbb{P}_{\mathcal{V}_{\theta},\pi}(\tau>n) \\
    &\leq  \mathbb{P}_{\mathcal{V}_{\theta},\Tilde{\pi}}\left(\forall t\in[1;n]\textnormal{ and }\:\forall \theta'\in\Theta \:\exists\: \theta''\neq \theta' \textnormal{ s.t } A_{\theta',\theta''}(t)<\Gamma_{\theta',\theta''}\right)\label{eqn:4_A_2_0}\\
    &\leq \sum_{\theta'\neq\theta }\mathbb{P}_{\mathcal{V}_{\theta},\Tilde{\pi}} \left( \: A_{\theta,\theta'}(n) < \Gamma_{\theta,\theta'}\right) \label{eqn:4_A_2_1}\\
    &\leq \sum_{\theta'\neq\theta }\mathbb{P}_{\mathcal{V}_{\theta},\Tilde{\pi}} \left( \: N_{\theta}> n^{1/2+\epsilon}\right) + \sum_{\theta'\neq\theta }\mathbb{P}_{\mathcal{V}_{\theta},\Tilde{\pi}}\left(\sum_{i=1}^{n^{1/2+\epsilon}}\mathbf{1}(X_i\neq 0) > n^{2\epsilon}  \right) \nonumber\\
    &+\sum_{\theta'\neq\theta } \mathbb{P}_{\mathcal{V}_{\theta},\Tilde{\pi}} \left( \: A_{\theta,\theta'}(n) < \Gamma_{\theta,\theta'},N_{\theta}\leq n^{1/2+\epsilon}, \sum_{i=1}^{n^{1/2+\epsilon}}\mathbf{1}(X_i\neq 0) \leq n^{2\epsilon}\right),
    \label{eqn:stopping_time_analysis}
\end{align}
where in \eqref{eqn:4_A_2_0} we use the definition of the stopping time and the fact that the control policies of $\pi$ and $\Tilde{\pi}$ are the same before the policy stops,  \eqref{eqn:4_A_2_1} comes from the fact that 
\begin{align*}
    \left\{\forall t\in[1;n]\textnormal{ and }\:\forall \theta'\in\Theta \:\exists\: \theta''\neq \theta' \textnormal{ s.t } A_{\theta',\theta''}(t)<\Gamma_{\theta',\theta''}\right\} \subset \left\{ \exists\: \theta' \neq \theta \textnormal{ s.t } A_{\theta,\theta'}(n)<\Gamma_{\theta,\theta'}\right\} 
\end{align*}
and the union bound, and \eqref{eqn:stopping_time_analysis} comes from the law of total probability.  
The first term on the right hand side of  \eqref{eqn:stopping_time_analysis} goes to zero when $n\rightarrow \infty$ by Lemma~\ref{lem:1}, and the second term on the right hand side of \eqref{eqn:stopping_time_analysis} can be upper bounded by $e^{-\Omega_{n\rightarrow \infty}(n^{2\epsilon})}$ by applying Freedman's inequality for martingales as follows. By defining the martingale sequence $\{V_{t}\}_{t=1}^{\infty}$ with 
\begin{align}
    V_t \triangleq \sum_{i=1}^{t}\mathbf{1}(X_i\neq 0) - \sum_{i=1}^{t}\sum_{x\in\mathcal{X}\setminus\{0\}} P_{X_i|Y^{i-1},X^{i-1}}(x)
\end{align}
and $U_t = V_t - V_{t-1}$ for all $t\in\mathbb{N}^+$, 
we have 
\begin{align}
    &\mathbb{P}_{\mathcal{V}_{\theta},\Tilde{\pi}}\left(\sum_{i=1}^{n^{1/2+\epsilon}}\mathbf{1}(X_i\neq 0) > n^{2\epsilon}  \right) \nonumber\\
    &\quad \leq \mathbb{P}_{\mathcal{V}_{\theta},\Tilde{\pi}}\left(\sum_{i=1}^{n^{1/2+\epsilon}}\mathbf{1}(X_i\neq 0) - \sum_{i=1}^{n^{1/2}+\epsilon}\sum_{x\in\mathcal{X}\setminus\{0\}} P_{X_i|Y^{i-1},X^{i-1}}(x) > \Omega_{n\rightarrow \infty}(n^{2\epsilon})  \right)\\
    &\quad \leq \exp\left(-\frac{\Omega_{n\rightarrow\infty}(n^{4\epsilon})}{O_{n\rightarrow \infty}(n^{\epsilon}) + O_{n\rightarrow \infty}(n^{2\epsilon})}\right),
\end{align}
where we use the fact that 
\begin{align}
    n^{2\epsilon} - \sum_{i=1}^{n^{1/2}+\epsilon}\sum_{x\in\mathcal{X}\setminus\{0\}} P_{X_i|Y^{i-1},X^{i-1}}(x) = \Omega_{n\rightarrow \infty}(n^{2\epsilon}),
\end{align}
and 
\begin{align}
    \sum_{i=1}^{n^{1/2+\epsilon}}\mathbb{E}_{\mathcal{V}_{\theta},\Tilde{\pi}}[U_i^2|\mathcal{F}_{i-1}] \leq \sum_{i=1}^{n^{1/2+\epsilon}}\mathbb{E}_{\mathcal{V}_{\theta},\tilde{\pi}}[\mathbf{1}(X_i\neq 0)|\mathcal{F}_{i-1}] \leq O_{n\rightarrow \infty}(n^{\epsilon}). 
\end{align}
To analyze the third term on the right hand side of \eqref{eqn:stopping_time_analysis}, we first define the following events for any $\theta\in\Theta$ and $\theta'\neq \theta$.   
\begin{align*}
    \mathcal{D}_{\theta,\theta'} & \triangleq \left\{(x^{n},y^{n})\in \mathcal{X}^n\times \mathcal{Y}^n:A_{\theta,\theta'}(n)<\Gamma_{\theta,\theta'}, N_{\theta}\leq n^{1/2+\epsilon}, \sum_{i=1}^{n^{1/2+\epsilon}} \mathbf{1}(X_i\neq 0)\leq n^{2\epsilon}\right\}.\\
    \mathcal{D}'_{\theta,\theta'} &\triangleq \left\{(x^n_{n^{1/2+\epsilon}+1},y^n_{n^{1/2+\epsilon}+1}): \sum_{t=n^{1/2+\epsilon}+1}^n L_{\theta,\theta'}(t)<\Gamma_{\theta,\theta'} - n^{2\epsilon}\min_{x\in\mathcal{X}\setminus\{0\},y\in\mathcal{Y}} \log\frac{{\nu}^x_{\theta}(y)}{{\nu}^x_{\theta'}(y)}\right\}.
\end{align*}
Then, 
\begin{align}
    &\mathbb{P}_{\mathcal{V}_{\theta},\Tilde{\pi}} \left( \: A_{\theta,\theta'}(n) < \Gamma_{\theta,\theta'},N_{\theta}\leq n^{1/2+\epsilon}, \sum_{i=1}^{n^{1/2+\epsilon}}\mathbf{1}(X_i\neq 0) \leq n^{2\epsilon}\right) \nonumber\\
    &= \sum_{(x^n,y^n)\in \mathcal{D}_{\theta,\theta'}} \mathbb{P}_{\mathcal{V}_{\theta},\Tilde{\pi}}(x^n,y^n)  \label{eqn:4_A_2_2}\\
    &= \sum_{(x^n,y^n)\in \mathcal{D}_{\theta,\theta'}} \mathbb{P}_{\mathcal{V}_{\theta},\Tilde{\pi}}(x^{n^{1/2+\epsilon}},y^{n^{1/2+\epsilon}}) \prod_{t=n^{1/2+\epsilon}+1}^n P_{X;\theta}(x_t) \nu_{\theta}^{x_t}(y_t)  \label{eqn:4_A_2_3}\\
    &\leq \sum_{\left(x^n_{n^{1/2+\epsilon}+1},y^n_{n^{1/2+\epsilon}+1}\right)\in\mathcal{D}_{\theta,\theta'}'} \prod_{t=n^{1/2+\epsilon}+1}^n P_{X;\theta}(x_t) \nu_{\theta}^{x_t}(y_t)  \label{eqn:4_A_2_4}\\
    &= \mathbb{P}_{(P_{X;\theta}\circ \nu_{\theta}^x)^{\otimes (n - n^{1/2+\epsilon})}}\left( \sum_{t=n^{1/2+\epsilon}+1}^n L_{\theta,\theta'}(t) < \Gamma_{\theta,\theta'} - n^{2\epsilon}\min_{x\in\mathcal{X}\setminus\{0\},y\in\mathcal{Y}}\log \frac{\nu_{\theta}^x(y)}{\nu_{\theta'}^x(y)}\right)
    \label{eqn:4_A_2_5},
\end{align}
where \eqref{eqn:4_A_2_2} follows from the definition of $\mathcal{D}_{\theta,\theta'}$, \eqref{eqn:4_A_2_3} follows from the definition of the control policy and the fact that $N_{\theta}\leq n^{1/2+\epsilon}$ when the event $\mathcal{D}_{\theta,\theta'}$ holds, \eqref{eqn:4_A_2_4} follows by lower bounding $L_{\theta,\theta'}(t)$ by $\min_{x\in\mathcal{X}\setminus\{0\},y\in\mathcal{Y}} \log\frac{\nu^x_{\theta}(y)}{\nu_{\theta'}^x(y)}$ for all $t\leq n^{1/2+\epsilon}$ satisfying $X_t\neq 0$ and marginalizing over the sequences $(x^{n^{1/2+\epsilon}},y^{n^{1/2+\epsilon}})\in\mathcal{X}^{n^{1/2+\epsilon}}\times \mathcal{Y}^{n^{1/2+\epsilon}}$, and in \eqref{eqn:4_A_2_5} we use the notation \\ $\mathbb{P}_{(P_{X;\theta}\circ \nu_{\theta}^x)^{\otimes (n-n^{1/2+\epsilon})}}$ to emphasize that the tuples $\{(X_t,Y_t)\}_{t=n^{1/2+\epsilon}+1}^n$ in \eqref{eqn:4_A_2_5} are generated i.i.d from the distribution $P_{X;\theta} \circ \nu_{\theta}^x$, where $P_{X;\theta}\circ \nu_{\theta}^x(x,y) \triangleq P_{X;\theta}(x)\nu_{\theta}^x(y)$ for all $x\in\mathcal{X}$ and $y\in\mathcal{Y}$. Therefore, $\{L_{\theta,\theta'}(t)\}_{t=n^{1/2+\epsilon}+1}^n$ are generated independently in \eqref{eqn:4_A_2_5}. 
The event in the right hand side of \eqref{eqn:4_A_2_5} can be written as 
\begin{align*}
    &\sum_{t=n^{1/2+\epsilon}+1}^n L_{\theta,\theta'}(t) - (n-n^{1/2+\epsilon})\left(\sum_{x\in\mathcal{X}}P_{X;\theta}(x)\mathbb{D}(\nu_{\theta}^x\Vert \nu_{\theta'}^x) \right) \\
    &< \Gamma_{\theta,\theta'} - (n-n^{1/2+\epsilon})\left(\sum_{x\in\mathcal{X}}P_{X;\theta}(x)\mathbb{D}(\nu_{\theta}^x\Vert \nu_{\theta'}^x) \right) - n^{2\epsilon}\min_{x\in\mathcal{X}\setminus\{0\},y\in\mathcal{Y}}\log \frac{\nu_{\theta}^x(y)}{\nu_{\theta'}^x(y)},
\end{align*}
where 
\begin{align*}
    &\Gamma_{\theta,\theta'} - (n-n^{1/2+\epsilon})\left(\sum_{x\in\mathcal{X}}P_{X;\theta}(x)\mathbb{D}(\nu_{\theta}^x\Vert \nu_{\theta'}^x) \right) - n^{2\epsilon}\min_{x\in\mathcal{X},y\in\mathcal{Y}}\log \frac{\nu_{\theta}^x(y)}{\nu_{\theta'}^x(y)} \\
    &= -n\alpha_{\theta}\zeta + n^{1/2+\epsilon}\left(\sum_{x\in\mathcal{X}}P_{X;\theta}(x)\mathbb{D}(\nu_{\theta}^x\Vert p_{\theta'}^x) \right) - n^{2\epsilon}\min_{x\in\mathcal{X},y\in\mathcal{Y}}\log \frac{\nu_{\theta}^x(y)}{\nu_{\theta'}^x(y)} \\
    &\leq - C_1 \zeta  n^{1/2}
\end{align*}
for some $C_1>0$ for all $n$ sufficiently large by plugging in the definition of $\Gamma_{\theta,\theta'}$ and choosing $\epsilon$ arbitrarily small so that the term $-n\alpha_{\theta}\zeta$ in the definition of $\Gamma_{\theta,\theta'}$ dominates in the above expression. 
Then, we apply Bernstein's inequality to upper bound the right hand side of \eqref{eqn:4_A_2_5} by
\begin{align}
    \exp\left(-\frac{1}{2}\frac{C_1^2\zeta^2 n}{\sum_{t=n^{1/2+\epsilon}+1}^n \mathbb{E}_{P_{X;\theta}\circ \nu_{\theta}^x}[L_{\theta,\theta'}(t)^2]+\frac{1}{3}M C_1\zeta n^{1/2}}\right),
    \label{eqn:4_A_2_6}
\end{align}
where $0<M<\infty$ is some constant such that $|L_{\theta,\theta'}(t)|<M$ almost surely for all $t\in\mathbb{N}^+$ and for all $n^{1/2+\epsilon}< t < n$, and 
\begin{align}
    \mathbb{E}_{P_{X;\theta}\circ \nu_{\theta}^x}[L_{\theta,\theta'}(t)^2] = \sum_{x\in\mathcal{X}\setminus\{0\}}\alpha_{\theta}\Bar{P}_{X;\theta}(x)\sum_{y\in\mathcal{Y}} \nu_{\theta}^x(y) \left(\log \frac{\nu_{\theta}^x(y)}{\nu_{\theta'}^x(y)}\right)^2 &\leq O_{n\rightarrow\infty}(n^{-1/2}), 
\end{align}
where we use the fact that $\nu_{\theta}^0(y)\left(\log \frac{\nu_{\theta}^0(y)}{\nu_{\theta'}^0(y)}\right)^2$ is $0$ for all $y\in\mathcal{Y}$. Therefore, \eqref{eqn:4_A_2_6} becomes $e^{-\Omega_{n\rightarrow\infty}(n^{1/2})}$, which decrease to zero when $n\rightarrow 0$. Finally, combining \eqref{eqn:stopping_time_analysis} and \eqref{eqn:4_A_2_6}, we conclude that 
\begin{align}
    \lim_{n\rightarrow\infty} \mathbb{P}_{\mathcal{V}_{\theta},\pi }(\tau>n) = 0
    \label{eqn:4_A_2_7}
\end{align}
for any $\theta\in\Theta$ and $\zeta>0$ when the stopping time is defined in \eqref{eqn:stop_condition} with the threshold given by \eqref{eqn:stop_threshold}. 

\paragraph{Relative Entropy }
We denote the output distribution of Willie under the dummy policy $\Tilde{\pi}$ of $\pi$ by  $\Tilde{P}_{Z^n;\theta}(z^n)\triangleq \mathbb{P}_{\mathcal{V}_{\theta},\mathcal{Q}_{\theta},\Tilde{\pi}}(z^n)$ for any $z^n\in\mathcal{Z}^n$, and $\Tilde{P}_{Z^n;\theta}(z^n) = \prod_{i=1}^n \Tilde{P}_{Z^{i}|Z^{i-1};\theta}(z_i)$ is the factorization $\Tilde{P}_{Z^n;\theta}$ by the product of conditional probabilities under the dummy policy $\tilde{\pi}$. Since $\Tilde{P}_{Z^n;\theta}$ is the output distribution of $Z^n$ when the policy does not stop, it holds that 
\begin{align}
    \mathbb{D}(P_{Z^{n};\theta}\Vert (q_{\theta}^{0})^{\otimes n})\leq \mathbb{D}(\Tilde{P}_{Z^{n};\theta}\Vert (q_{\theta}^{0})^{\otimes n}). 
\end{align}
Given any $\theta\in\Theta$, by the chain rule, the relative entropy $\mathbb{D}(\Tilde{P}_{Z^{n};\theta}\Vert (q_{\theta}^{0})^{\otimes n})$ can be expressed as 
\begin{align}
    \mathbb{D}(\Tilde{P}_{Z^{n};\theta}\Vert (q_{\theta}^{0})^{\otimes n}) 
    &= \sum_{i=1}^n \mathbb{E}_{Z^{i-1};\mathcal{V}_{\theta},\mathcal{Q}_{\theta},\Tilde{\pi}}\left[\mathbb{D}(\Tilde{P}_{Z_i|Z^{i-1};\theta}\Vert q_{\theta}^0)\right] \\
    &\leq \sum_{i=1}^n \mathbb{E}_{Z^{i-1};\mathcal{V}_{\theta},\mathcal{Q}_{\theta},\Tilde{\pi}} \mathbb{E}_{\hat{\theta}_{\textnormal{ML}}(i-1)|Z^{i-1};\mathcal{V}_{\theta} ,\mathcal{Q}_{\theta},\Tilde{\pi}} \left[\mathbb{D}\left(\Tilde{P}_{Z_i|Z^{i-1},\hat{\theta}_{\textnormal{ML}}(i-1);\theta}\middle\Vert q_{\theta}^0\right)\label{eqn:div_analysis_1} \right]\\
    &= \sum_{i=1}^n \mathbb{E}_{\hat{\theta}_{\textnormal{ML}}(i-1);\mathcal{V}_{\theta},\mathcal{Q}_{\theta},\Tilde{\pi}} \mathbb{D}\left(\Tilde{P}_{Z_i|\hat{\theta}_{\textnormal{ML}}(i-1);\theta}\middle\Vert q_{\theta}^0\right) \label{eqn:div_analysis_2} \\
    &\leq \sum_{i=1}^n \mathbb{P}_{\mathcal{V}_{\theta},\Tilde{\pi}}(\hat{\theta}_{\textnormal{ML}}(i-1)\neq \theta) \max_{\theta'\in\Theta}\mathbb{D}(\Tilde{P}_{Z_i|\hat{\theta}_{\textnormal{ML}}(i-1)=\theta';\theta}\Vert q_{\theta}^0) \nonumber\\
    &\quad+ \sum_{i=1}^n \mathbb{P}_{\mathcal{V}_{\theta},\Tilde{\pi}}(\hat{\theta}_{\textnormal{ML}}(i-1) =\theta) \mathbb{D}\left(\sum_{x\in\mathcal{X}}{P}_{X;\theta}(x)q^x_{\theta}\middle\Vert q_{\theta}^0\right)\\
    &\leq O_{n\rightarrow\infty}(n^{-1}) \left(n^{1/2+\epsilon}+\sum_{i=n^{1/2+\epsilon}+1}^{n} \mathbb{P}_{\mathcal{V}_{\theta},\tilde{\pi}}(N_{\theta} \geq i) \right) \nonumber\\
    &\quad+ n \mathbb{D}\left(\sum_{x\in\mathcal{X}}{P}_{X;\theta}(x)q^x_{\theta}\middle\Vert q_{\theta}^0\right) 
    \label{eqn:div_analysis_3}\\
    &\leq o_{n\rightarrow\infty}(1) + n\left(\frac{(\alpha_{\theta})^2} {2}\chi_2\left(\sum_{x\in\mathcal{X}\setminus\{0\}}\Bar{P}_{X;\theta}(x)q_{\theta}^x\middle\Vert q_{\theta}^0\right) + o_{n\rightarrow \infty}((\alpha_{\theta})^2) \right)\label{eqn:div_analysis_4}\\
    &\leq \eta + o_{n\rightarrow\infty}(1),
\end{align}
where \eqref{eqn:div_analysis_1} follows from the convexity of the relative entropy, Jensen's inequality and the fact that 
\begin{align}
    \Tilde{P}_{Z_i|Z^{i-1};\theta}(z) & = \mathbb{P}_{\mathcal{V}_{\theta},\mathcal{Q}_{\theta},\Tilde{\pi}}(z|Z^{i-1}) \\
    &= \sum_{\theta'\in\Theta}\mathbb{P}_{\mathcal{V}_{\theta},\mathcal{Q}_{\theta},\Tilde{\pi}}(\hat{\theta}_{\textnormal{ML}}(i-1)=\theta'|Z^{i-1})\mathbb{P}_{\mathcal{V}_{\theta},\mathcal{Q}_{\theta},\Tilde{\pi}}(z|\hat{\theta}_{\textnormal{ML}}(i-1)=\theta',Z^{i-1})\\
    &= \mathbb{E}_{\hat{\theta}_{\textnormal{ML}}(i-1)|Z^{i-1};\mathcal{V}_{\theta},\mathcal{Q}_{\theta},\Tilde{\pi}} \left[ \Tilde{P}_{Z_i|\hat{\theta}_{\textnormal{ML}}(i-1),Z^{i-1};\theta}(z)\right]
\end{align}
for any $z\in\mathcal{Z}$, 
\eqref{eqn:div_analysis_2} follows because $\Tilde{P}_{Z_i|Z^{i-1},\hat{\theta}_{\textnormal{ML}}(i-1);\theta}=\Tilde{P}_{Z_i|
\hat{\theta}_{\textnormal{ML}(i-1)};\theta}$ by our construction of the policy and the fact that $Z_i$ is independent of $Z^{i-1}$ conditioned on $X_i$,  \eqref{eqn:div_analysis_3} follows by upper bounding $\max_{\theta'\in\Theta}\mathbb{D}(\Tilde{P}_{Z_i|\hat{\theta}(i-1)=\theta';\theta}\Vert q_{\theta}^0)$  by $O_{n\rightarrow\infty}(n^{-1})$ because 
\begin{align}
    \Tilde{P}_{Z_i|\hat{\theta}(i-1)=\theta';\theta} &= \sum_{x\in\mathcal{X}} P_{X;\theta'}(x) q_{\theta}^x 
\end{align}
for any $\theta'\in\Theta$, 
$\sum_{x\neq 0}P_{X;\theta'}(x)=O_{n\rightarrow\infty}(n^{-1/2})$ for any $\theta'\in\Theta$, and we apply the result from \cite[Lemma 1]{Bloch2016}, and in \eqref{eqn:div_analysis_4} we use the result from Lemma~\ref{lem:1} so that $\sum_{i=n^{1/2+\epsilon}}^n \mathbb{P}_{\mathcal{V}_{\theta},\Tilde{\pi}}(N_{\theta}\geq i) \leq o_{n\rightarrow\infty}(1)$ and apply \cite[Lemma 1]{Bloch2016} again on the term $\mathbb{D}\left(\sum_{x\in\mathcal{X}}{P}_{X;\theta}(x)q^x_{\theta}\middle\Vert q_{\theta}^0\right)$. 
Therefore, we conclude that 
\begin{align}
    \lim_{n\rightarrow \infty} \mathbb{D}(P_{Z^{n};\theta}\Vert (q_{\theta}^{0})^{\otimes n}) \leq \eta
\end{align}
for all $\theta\in\Theta$ and the covertness constraint is satisfied. 

\paragraph{Estimation error}
    For any $\theta\in\Theta$, the decision error probability can be upper bounded by 
    \begin{align}
        \mathbb{P}_{\mathcal{V}_{\theta},\pi}(\psi(X^{\tau},Y^{\tau})\neq\theta) &\leq \sum_{\theta'\neq\theta} \mathbb{P}_{\mathcal{V}_{\theta},\pi}\left(\psi(X^{\tau},Y^{\tau}) = \theta'\right)\\
        &\leq \sum_{\theta'\neq\theta} \mathbb{E}_{\mathcal{V}_{\theta},\pi}\left[\mathbf{1}(A_{\theta',\theta}(\tau)\geq \Gamma_{\theta',\theta})\right] \\
        &= \sum_{\theta'\neq\theta} \mathbb{E}_{\mathcal{V}_{\theta},\pi}\left[\mathbf{1}(A_{\theta,\theta'}(\tau)\leq -\Gamma_{\theta',\theta})\right]\\
        &= \sum_{\theta'\neq\theta} \mathbb{E}_{\mathcal{V}_{\theta'},\pi}\left[e^{A_{\theta,\theta'}(\tau)}\mathbf{1}(A_{\theta,\theta'}(\tau)\leq -\Gamma_{\theta',\theta})\right]\\
        &\leq \sum_{\theta'\neq\theta} e^{-\Gamma_{\theta',\theta}}\\
        &\leq |\Theta| e^{-\min_{\theta'\neq\theta}\Gamma_{\theta',\theta}}. 
    \end{align}
    Therefore,
    \begin{align}
        \max_{\theta\in\Theta}\mathbb{P}_{\mathcal{V}_{\theta},\pi}(\psi(X^{\tau},Y^{\tau})\neq\theta) &\leq |\Theta|e^{-\min_{\theta\in\Theta}\min_{\theta'\neq\theta}\Gamma_{\theta',\theta}}\\
        &= |\Theta|e^{-\min_{\theta\in\Theta}\min_{\theta'\neq\theta} n\alpha_{\theta'}\left(\sum_{x\neq 0}\Bar{P}_{X;\theta'}(x)\mathbb{D}(\nu_{\theta'}^x\Vert \nu_{\theta}^x)-\zeta\right)}. 
        \label{eqn:error_probability}
    \end{align}

    \paragraph{Achievable exponent}
    By \eqref{eqn:error_probability},  the definition of $\gamma_{1}(\pi)$ and by making $\zeta>0$ arbitrarily small, we have 
    \begin{align}
        \gamma_{1}^*  &\geq  \sqrt{2\eta} \min_{\theta\in\Theta}\min_{\theta'\neq\theta}  \frac{\sum_{x\neq 0}\Bar{P}_{X;\theta'}(x)\mathbb{D}(\nu_{\theta'}^x\Vert \nu_{\theta}^x)}{\sqrt{\chi_2\left(\sum_{x\neq 0}\Bar{P}_{X;\theta'}(x)q_{\theta'}^x\Vert q_{\theta'}^0\right)}}\\
        &= \sqrt{2\eta} \min_{\theta'\in\Theta}\min_{\theta\neq\theta'}  \frac{\sum_{x\neq 0}\Bar{P}_{X;\theta'}(x)\mathbb{D}(\nu_{\theta'}^x\Vert \nu_{\theta}^x)}{\sqrt{\chi_2\left(\sum_{x\neq 0}\Bar{P}_{X;\theta'}(x)q_{\theta'}^x\Vert q_{\theta'}^0\right)}}\\
        &= \sqrt{2\eta} \min_{\theta'\in\Theta} \max_{\Bar{P}_X\in\mathcal{P}_{\mathcal{X}\setminus\{0\}}}\min_{\theta\neq\theta'}  \frac{\sum_{x\neq 0}\Bar{P}_{X}(x)\mathbb{D}(\nu_{\theta'}^x\Vert \nu_{\theta}^x)}{\sqrt{\chi_2\left(\sum_{x\neq 0}\Bar{P}_X(x)q_{\theta'}^x\Vert q_{\theta'}^0\right)}}
        \label{eqn:claim_exponent},
    \end{align}
    where \eqref{eqn:claim_exponent} follows from the definition of $\Bar{P}_{X;\theta'}$ in \eqref{eqn:def_of_P_theta}, which completes the proof of Theorem~\ref{thm:main_result1}.

\section{Proof of Theorem~\ref{thm:main_result2}}
\label{sec:prof_thm2}

In the converse, we assume that the decision making policy $\pi$ satisfies the property described in \eqref{eqn:control_policy_assumption_1}. 
We assume without loss of generality that the policy $\pi$ can achieve a positive detection error exponent, i.e., $\gamma_{1}(\pi)>0$. If $\gamma_{1}(\pi)=0$, our converse result holds trivially. 
We first need to build up the relationship between achievable exponent and the log-likelihood ratio, and \cite{Li2020} provides a neat way to do so. We rewrite part of their proofs in the following equations for completeness. 
For any $\theta\in\Theta$, $\theta'\neq \theta$ and $\lambda>0$, it holds that 
\begin{align}
    &\mathbb{P}_{\mathcal{V}_{\theta'},\pi}(\psi(X^{\tau},Y^{\tau})=\theta') - e^{\sqrt{n}\lambda} \mathbb{P}_{\mathcal{V}_{\theta},\pi}(\psi(X^{\tau},Y^{\tau})=\theta') \nonumber\\
    &= \mathbb{E}_{\mathcal{V}_{\theta'},\pi}\left[ \mathbf{1}\left(\psi(X^{\tau},Y^{\tau})=\theta'\right)\right] - \mathbb{E}_{\mathcal{V}_{\theta'},\pi}\left[ e^{\sqrt{n}\lambda+A_{\theta,\theta'}(\tau)_{\theta,\theta'}}\mathbf{1}\left(\psi(X^{\tau},Y^{\tau})=\theta'\right)\right] \nonumber\\
    &\leq \mathbb{E}_{\mathcal{V}_{\theta'},\pi}\left[ \mathbf{1}(\psi(X^{\tau},Y^{\tau})=\theta') \mathbf{1}(A_{\theta,\theta'}(\tau)\leq -\sqrt{n}\lambda)\right]
    \label{eqn:conv_1}\\
    &\leq \mathbb{P}_{\mathcal{V}_{\theta'},\pi}(A_{\theta,\theta'}(\tau)\leq -\sqrt{n}\lambda),
    \label{eqn:conv_2}
\end{align}
where \eqref{eqn:conv_1} follows because $1-e^{\sqrt{n}\lambda+A_{\theta,\theta'}(\tau)} \leq 1$ for all $A_{\theta,\theta'}(\tau) \leq -\sqrt{n}\lambda$ and $1-e^{\sqrt{n}\lambda+A_{\theta,\theta'}(\tau)} < 0 $ for all $A_{\theta,\theta'}(\tau) > -\sqrt{n}\lambda$. By re-arranging \eqref{eqn:conv_2}, we have that 
\begin{align}
    \mathbb{P}_{\mathcal{V}_{\theta},\pi}(\psi(X^{\tau},Y^{\tau})=\theta') &\geq e^{-\sqrt{n}\lambda} \left(\mathbb{P}_{\mathcal{V}_{\theta'},\pi}(\psi(X^{\tau},Y^{\tau})=\theta') - \mathbb{P}_{\mathcal{V}_{\theta'},\pi}(A_{\theta,\theta'}(\tau)\leq -\sqrt{n}\lambda)\right) \\
    &\geq e^{-\sqrt{n}\lambda} \left(1-e^{-\sqrt{n}\xi_2} - \mathbb{P}_{\mathcal{V}_{\theta'},\pi}\left(A_{\theta,\theta'}(\tau)\leq -\sqrt{n}\lambda\right)\right)
\end{align}
for some $\xi_2>0$ when $n$ is sufficiently large by the assumption that $\gamma_{1}(\pi)>0$. Then, 
\begin{align}
    \frac{-1}{\sqrt{n}}\log \mathbb{P}_{\mathcal{V}_{\theta},\pi}(\psi(X^{\tau},Y^{\tau})=\theta') \leq \lambda - \frac{1}{\sqrt{n}} \log \left(1-o_{n\rightarrow \infty}(1) - \mathbb{P}_{\mathcal{V}_{\theta'},\pi}\left(A_{\theta,\theta'}(\tau)\leq -\sqrt{n}\lambda\right)\right),
    \label{eqn:conv_3}
\end{align}
which implies that the detection error exponent $\gamma_{1}(\pi)$ is upper bounded by $\lambda$ whenever there exist some $\theta\in\Theta$ and $\theta'\neq \theta$ such that 
\begin{align}
    \lim_{n\rightarrow\infty} \mathbb{P}_{\mathcal{V}_{\theta'},\pi}\left(A_{\theta,\theta'}(\tau)\leq -\sqrt{n}\lambda\right) < 1. 
    \label{eqn:conv_4}
\end{align}
Equation~\eqref{eqn:conv_4} implies that if the detection error exponent $\gamma_{1}(\pi)$ is greater than some $\gamma>0$, then it holds that  
\begin{align}
    \lim_{n\rightarrow\infty} \mathbb{P}_{\mathcal{V}_{\theta},\pi}\left(A_{\theta,\theta'}(\tau) \geq \sqrt{n}\kappa\gamma\right) = 1
    \label{eqn:conv_lowerbound_llr}
\end{align}
for all $0<\kappa<1$ for any $\theta\in\Theta$ and $\theta'\neq \theta$. 
By using the above equality, we are able to derive a lower bound on the stopping time as shown in Lemma~\ref{lem:2} if $\gamma_{1}(\pi)$ can achieve the exponent $\gamma$. 
\begin{lemma}
    If the policy $\pi$ that satisfies \eqref{eqn:control_policy_assumption_1} achieves the detection error exponent $\gamma>0$ in (P1), then it holds that 
    \begin{align}
        \lim_{n\rightarrow\infty }\mathbb{P}_{\mathcal{V}_{\theta},\pi}\left(\tau < \frac{ \rho\gamma \sqrt{n} }{\sum_{x}P_{X;\theta}(x)\mathbb{D}(\nu_{\theta}^x\Vert \nu_{\theta'}^x)}\right) = 0 
    \end{align}
    for all $0<\rho<1$ and for all $\theta\in\Theta$ and $\theta'\neq \theta$. 
    \label{lem:2}
\end{lemma}
\begin{proof}
    The details of the proof are shown in Appendix~\ref{apx:prof_lem6}. The overall idea of the proof is as follows. From \eqref{eqn:conv_lowerbound_llr}, we know that the likelihood ratio $A_{\theta,\theta'}(\tau)$ at the stopping time is greater than $\sqrt{n}\kappa \gamma$ for any $0<\kappa<1$, $\theta'\neq \theta$ when the true hypothesis is $\theta$. Moreover, we  can always find a time index $\ell \leq \tau$ such that $A_{\theta,\theta'}(\ell)<O_{n\rightarrow\infty}(1)$ and $A_{\theta,\theta'}(t)>0$ for all $t\geq \ell$, implying that $X_{t}$ is generated from the distribution $P_{X;\theta}$ for all $t\geq \ell$. We can then show by Freedman's inequality on martingales that the amount of time required for the likelihood ratio $A_{\theta,\theta'}(t)$ to reach $\sqrt{n}\kappa\gamma$ from the value $A_{\theta,\theta'}(\ell)$ is at least $\frac{ (1-\delta)\kappa\gamma \sqrt{n} }{\sum_{x}P_{X;\theta}(x)\mathbb{D}(\nu_{\theta}^x\Vert \nu_{\theta'}^x)} $ for any $\delta>0$ with high probability. Then, by redefining $(1-\delta)\kappa$ as $\rho$, we complete the proof. 
\end{proof}
For any $\theta\in\Theta$ and $\theta'\neq\theta$, we define 
\begin{align}
    \Tilde{n}_{\theta,\theta'} \triangleq \frac{ \rho\gamma \sqrt{n} }{\sum_{x}P_{X;\theta}(x)\mathbb{D}(\nu_{\theta}^x\Vert \nu_{\theta'}^x)}.
    \label{eqn:n_tilde}
\end{align}
We know that $\Tilde{n}_{\theta,\theta'}\leq n$ for any $\theta\in\Theta$ and $\theta'\neq\theta$ by the stopping time constraint and Lemma~\ref{lem:2}. This also implies that 
\begin{align}
    \sum_{x\neq 0}P_{X;\theta}(x) \geq \frac{\rho \gamma}{\sqrt{n} \sum_{x\neq 0}\Bar{P}_{X;\theta}(x)\mathbb{D}(\nu_{\theta}^x\Vert \nu_{\theta'}^x)}
    \label{eqn:lower_bound_alpha}
\end{align}
for all $0<\rho<1$ for any $\theta\in\Theta$ and $\theta'\neq \theta$ if the exponent $\gamma$ is achievable by rearranging \eqref{eqn:n_tilde}, where for any $\theta\in\Theta$ and $x\in\mathcal{X}\setminus\{0\}$
\begin{align}
    \Bar{P}_{X;\theta}(x) \triangleq  \frac{P_{X;\theta}(x)}{\sum_{x\in\mathcal{X}\setminus\{0\}}P_{X;\theta}(x)}. 
\end{align}
Note that we lower bound $\sum_{x\in\mathcal{X}\setminus\{0\}}P_{X;\theta}(x)$ in \eqref{eqn:lower_bound_alpha}, which is the probability of selecting the non-null action when the ML estimate of the hypothesis is $\theta$. One can see that the value of $\sum_{x\in\mathcal{X}\setminus\{0\}}P_{X;\theta}(x)$ can not be too high in order to satisfy the covertness constraint. Our next step is to establish the connection between $P_{X;\theta}$ and the covertness constraint $\eta$. 
Therefore, we proceed by lower bounding the relative entropy $\mathbb{D}(P_{Z^{n};\theta}||(q_{\theta}^{0})^n) $ as follows.
For any $\theta'\neq \theta$, we have
\begin{align}
    \mathbb{D}(P_{Z^{n};\theta}||(q_{\theta}^{0})^{\otimes n}) &\geq \mathbb{D}(P_{Z^{\Tilde{n}_{\theta,\theta'}};\theta}||(q_{\theta}^{0})^{\otimes \Tilde{n}_{\theta,\theta'}})
    \label{eqn:prof_thm2_divergencebound1}\\
    &\geq \Tilde{n}_{\theta,\theta'} \mathbb{D}(\Bar{P}_{Z;\theta}||q_{\theta}^{0})
    \label{eqn:prof_thm2_divergencebound2}, 
\end{align}
where in \eqref{eqn:prof_thm2_divergencebound1} we use monotonicity of the relative entropy and the fact that $n\geq \tilde{n}_{\theta,\theta'}$, and \eqref{eqn:prof_thm2_divergencebound2} follows from \cite[Equation (13)]{Wang2016b}, where for any $z\in\mathcal{Z}$, 
\begin{align*}
    \Bar{P}_{Z;\theta}(z) &\triangleq \frac{1}{\Tilde{n}_{\theta,\theta'}} \sum_{i=1}^{\Tilde{n}_{\theta,\theta'}} P_{Z_i;\theta}(z)\\
    &= \frac{1}{\Tilde{n}_{\theta,\theta'}} \sum_{i=1}^{\Tilde{n}_{\theta,\theta'}} \left(\sum_{\theta'}\mathbb{P}_{\mathcal{V}_{\theta},\pi}(\hat{\theta}_{\textnormal{ML}}(i-1)=\theta',\tau\geq i)\left(\sum_{x\in\mathcal{X}}P_{X;\theta'}(x)q^x_{\theta}(z)\right)+\mathbb{P}_{\mathcal{V}_{\theta},\pi}(\tau<i)q_{\theta}^0(z)\right) 
\end{align*}
is the time averaged distribution on $\mathcal{Z}$. 
Note that we can define another distribution $\widehat{P}_{Z;\theta}$ as
\begin{align*}
    \widehat{P}_{Z;\theta} &\triangleq \frac{1}{\Tilde{n}_{\theta,\theta'}} \sum_{i=1}^{\Tilde{n}_{\theta,\theta'}} \Bigg(\mathbb{P}_{\mathcal{V}_{\theta},\pi}(\hat{\theta}_{\textnormal{ML}}(i-1)=\theta,\tau\geq i)\left(\sum_{x\in\mathcal{X}}P_{X;\theta}(x)q^x_{\theta}(z)\right) \nonumber\\
    &\hspace{6cm}+\mathbb{P}_{\mathcal{V}_{\theta},\pi}(\hat{\theta}_{\textnormal{ML}}(i-1)\neq \theta \textnormal{ or }\tau<i)q_{\theta}^0(z)\Bigg)\\
    &= \left( \widehat{\alpha}\left(\sum_{x\neq 0}\widehat{P}_{X;\theta}(x)q_{\theta}^x(z)\right) + (1-\widehat{\alpha})q_{\theta}^{0}(z)\right),
\end{align*}
where 
\begin{align}
    \widehat{\alpha} \triangleq 
    \frac{1}{\Tilde{n}_{\theta,\theta'}}\sum_{i=1}^{\Tilde{n}_{\theta,\theta'}}\mathbb{P}_{\mathcal{V}_{\theta},\pi}(\hat{\theta}_{\textnormal{ML}}(i-1)=\theta,\tau \geq i) \sum_{x\neq 0}P_{X;\theta}(x),
    \label{eqn:def_alpha_tilde}
\end{align}
and for all $x\in\mathcal{X}\setminus\{0\}$ 
\begin{align}
    \widehat{P}_{X;\theta}(x) &\triangleq \frac{ \frac{1}{\Tilde{n}_{\theta,\theta'}}\sum_{i=1}^{\Tilde{n}_{\theta,\theta'}}\mathbb{P}_{\mathcal{V}_{\theta},\pi}(\hat{\theta}_{\textnormal{ML}}(i-1)=\theta,\tau \geq i) P_{X;\theta}(x)}{\widehat{\alpha}} \\
    &= \frac{P_{X;\theta}(x)}{\sum_{x\in\mathcal{X}\setminus\{0\}}P_{X;\theta}(x)}\\
    &= \Bar{P}_{X;\theta}(x). 
\end{align}
$\widehat{P}_{Z;\theta}$ is therefore
the time averaged distribution on $\mathcal{Z}$ when the control policy is given by
\begin{align}
    P_{X_t|X^{t-1},Y^{t-1}}(x) &= \mathbf{1}(\hat{\theta}_{\textnormal{ML}}(t-1)=\theta,\tau\geq i)P_{X;\theta}(x) \\
    &+ \left(1-\mathbf{1}(\hat{\theta}_{\textnormal{ML}}(t-1)=\theta,\tau\geq i)\right) \mathbf{1}(x=0),
\end{align}
which is the control policy that determines the action from the non-trivial policy $P_{X;\theta}(x)$ only when $\hat{\theta}_{\textnormal{ML}}(t-1)=\theta$ and the decision making policy does not stop, so that 
\begin{align}
    \mathbb{D}(\Bar{P}_{Z;\theta}||q_{\theta}^{0}) \geq \mathbb{D}(\widehat{P}_{Z;\theta}||q_{\theta}^{0}). 
    \label{eqn:prof_thm2_divergencebound3}
\end{align}
Moreover, the following lemma proved in Appendix~\ref{apx:proof_lem3} implies that $\Tilde{n}_{\theta,\theta'}=\Omega(n)$ for all $\theta\in\Theta$ and $\theta'\neq \theta$. 
\begin{lemma} 
    It holds that $\sum_{x\neq 0}P_{X;\theta}(x)=O_{n\rightarrow \infty}(n^{-1/2})$ for any $\theta\in\Theta$. 
    \label{lem:3}
\end{lemma}
From Lemma~\ref{lem:2}, we know that 
$\sum_{x\neq 0}P_{X;\theta}(x) = \Omega_{n\rightarrow \infty}(n^{-1/2})$ for any $\theta\in\Theta$, otherwise, the policy fails to meet the stopping time constraint, i.e., $\Tilde{n}_{\theta,\theta'}$ becomes $\omega(n)$. Therefore, $\sum_{x\neq 0}P_{X;\theta}(x) = \Theta(n^{-1/2})$. 
By our assumption on $P_{X_t|X^{t-1},Y^{t-1}}$ and the fact $\sum_{x\neq 0}P_{X;\theta}(x) = \Theta_{n\rightarrow \infty}(n^{-1/2})$ for any $\theta\in\Theta$, we can apply Lemma~\ref{lem:1} and claim that $\mathbb{P}_{\mathcal{V}_{\theta},\Tilde{\pi}}\left(N_{\theta} \geq n^{1/2+\epsilon}\right) = O_{n\rightarrow \infty}(n^{-\beta})$ for any $\epsilon>0$ and $\beta>0$. 
Then, for any $\theta$ and $\theta'\neq \theta$, it holds that 
\begin{align}
    &\frac{1}{\Tilde{n}_{\theta,\theta'}}\sum_{i=1}^{\Tilde{n}_{\theta,\theta'}}\mathbb{P}_{\mathcal{V}_{\theta},\pi}(\hat{\theta}_{\textnormal{ML}}(i-1)=\theta,\tau \geq i) \nonumber\\
    &\quad \geq 1 - \frac{1}{\Tilde{n}_{\theta,\theta'}}\sum_{i=1}^{\Tilde{n}_{\theta,\theta'}} \left(\mathbb{P}_{\mathcal{V}_{\theta},\pi}(\hat{\theta}_{\textnormal{ML}}(i-1)\neq \theta,\tau \geq i) + \mathbb{P}_{\mathcal{V}_{\theta},\pi}(\tau< i)\right)\\
    &\quad \geq 1 - \frac{1}{\Tilde{n}_{\theta,\theta'}}\left(\sum_{i=1}^{n^{1/2+\epsilon}} 1  + \sum_{i=n^{1/2+\epsilon}+1}^{\Tilde{n}_{\theta,\theta'}} \mathbb{P}_{\mathcal{V}_{\theta},\Tilde{\pi}}\left(N_{\theta} \geq i-1 \right)\right) - o_{n\rightarrow \infty}(1) \label{eqn:conv_5}\\
    &\quad \geq 1 - o_{n\rightarrow \infty}(1) \label{eqn:conv_6}, 
\end{align}
where \eqref{eqn:conv_5} follows from Lemma~\ref{lem:2} and \eqref{eqn:conv_6} follows from  the fact that $\mathbb{P}_{\mathcal{V}_{\theta},\Tilde{\pi}}\left(N_{\theta} \geq i \right) = O_{n\rightarrow \infty}(n^{-\beta})$ for any $i\geq n^{1/2+\epsilon}$, $\epsilon>0$ and $\beta>0$. 
Equation \eqref{eqn:conv_6} and the definition of $\widehat{\alpha}$ imply that 
\begin{align}
     \widehat{\alpha} \geq \left(\sum_{x\neq 0}P_{X;\theta}(x)\right)(1-o_{n\rightarrow \infty}(1)).
     \label{eqn:prof_thm2_alpha_tilede_aprox}
\end{align}
Then, by plugging \eqref{eqn:prof_thm2_alpha_tilede_aprox} in the covertness constraint \eqref{eqn:prof_thm2_divergencebound2}, we obtain
\begin{align}
    \eta &\geq \Tilde{n}_{\theta,\theta'} \mathbb{D}(\widehat{P}_{Z;\theta}||q_{\theta}^{0}) \\
    &\geq \Tilde{n}_{\theta,\theta'} \frac{\widehat{\alpha}^2}{2} \chi_2\left(\sum_{x\neq 0}\widehat{P}_{X;\theta}(x)q_{\theta}^x(z)\middle\Vert q_{\theta}^0\right)(1-o_{n\rightarrow \infty}(1)) 
    \label{enq:prof_thm2_bound_relative_entropy_1}\\
    &\geq \left(\frac{ \rho\gamma \sqrt{n} }{\sum_{x}P_{X;\theta}(x)\mathbb{D}(\nu_{\theta}^x\Vert \nu_{\theta'}^x)}\right)\frac{\left(\sum_{x\neq 0}P_{X;\theta}(x)\right)^2}{2}\chi_2\left(\sum_{x\neq 0}\widehat{P}_{X;\theta}(x)q_{\theta}^x(z)\middle\Vert q_{\theta}^0\right)(1-o_{n\rightarrow \infty}(1))
    \label{enq:prof_thm2_bound_relative_entropy_2}\\
    &\geq \frac{\rho \gamma \sqrt{n} \left(\sum_{x\neq 0}P_{X;\theta}(x)\right)\chi_2\left(\sum_{x\neq 0}\Bar{P}_{X;\theta}(x)q_{\theta}^x(z)\middle\Vert q_{\theta}^0\right)}{2\sum_{x\neq 0}\Bar{P}_{X;\theta}(x)\mathbb{D}(\nu_{\theta}^x\Vert \nu_{\theta'}^x)}(1-o_{n\rightarrow \infty}(1))
    \label{enq:prof_thm2_bound_relative_entropy_3}\\
    &\geq \frac{\rho^2 \gamma^2 \chi_2\left(\sum_{x\neq 0}\Bar{P}_{X;\theta}(x)q_{\theta}^x(z)\middle\Vert q_{\theta}^0\right)}{2\left(\sum_{x\neq 0}\Bar{P}_{X;\theta}(x)\mathbb{D}(\nu_{\theta}^x\Vert \nu_{\theta'}^x)\right)^2}(1-o_{n\rightarrow \infty}(1))
    \label{enq:prof_thm2_bound_relative_entropy_4}
\end{align}
for any $\theta\in\Theta$ and $\theta'\neq \theta$, where \eqref{enq:prof_thm2_bound_relative_entropy_1} follows from \cite[Lemma~1]{Bloch2016}, \eqref{enq:prof_thm2_bound_relative_entropy_2} follows from \eqref{eqn:prof_thm2_alpha_tilede_aprox} and the definition of $\Tilde{n}_{\theta,\theta'}$, \eqref{enq:prof_thm2_bound_relative_entropy_3} follows from the definition of $\Bar{P}_{X;\theta}$, and in \eqref{enq:prof_thm2_bound_relative_entropy_4} we use \eqref{eqn:lower_bound_alpha}.  So, by taking $n$ arbitrarily large and $\rho$ arbitrarily close to $1$, we have  
\begin{align}
    \gamma \leq \sqrt{2\eta} \frac{ \sum_{x\neq 0}\Bar{P}_{X;\theta}(x)\mathbb{D}(\nu_{\theta}^x\Vert \nu_{\theta'}^x)}{\sqrt{\chi_2\left(\sum_{x\neq 0}\Bar{P}_{X;\theta}(x)q_{\theta}^x(z)\middle\Vert q_{\theta}^0\right)}} . 
\end{align}
for any $\theta\in\Theta$ and $\theta'\neq \theta$ for the specific set of $\{P_{X;\theta}\}_{\theta\in\Theta}$. 
By taking the maximum over all $\Bar{P}_{X;\theta}\in\mathcal{P}_{\mathcal{X}\setminus\{0\}}$, we conclude that 
\begin{align}
    \gamma_{1}(\pi) \leq \min_{\theta\in\Theta}\max_{\Bar{P}_{X}\in\mathcal{P}_{\mathcal{X}\setminus\{0\}}}\min_{\theta'\neq \theta} \sqrt{2\eta} \frac{ \sum_{x\neq 0}\Bar{P}_{X}(x)\mathbb{D}(\nu_{\theta}^x\Vert \nu_{\theta'}^x)}{\sqrt{\chi_2\left(\sum_{x\neq 0}\Bar{P}_{X}(x) q_{\theta}^x\middle\Vert q_{\theta}^0\right)}}.
\end{align}

\section{Proof of Theorem~\ref{thm:main_result3}}
\label{sec:prof_thm3}
\subsection{Construction of Policy}
We first specify the policy $\pi=(\phi,\varphi,\psi)$. 
Let $\hat{\mathcal{V}}(t)$ be the estimated bandit of the bandit $\mathcal{V}$ at the time $t$, i.e., $\hat{\mathcal{V}}(t)=\{\hat{\nu}_x(t)\}_{x\in\mathcal{X}}$, where $\hat{\nu}_x(t)$ is the Gaussian distribution with mean 
\begin{align}
    \mu(\hat{\nu}_x(t)) = \frac{1}{T_x(t)}\sum_{\ell=1}^t Y_{\ell} \mathbf{1}(X_{\ell} = x),
\end{align}
for all $x\in\mathcal{X}$, where 
\begin{align}
    T_x(t) = \sum_{\ell=1}^t \mathbf{1}(X_{\ell}=x)
\end{align}
is the number of pulls on the arm $x\in\mathcal{X}$. We assume that the variance is $1$ and is known, and hence, 
\begin{align*}
     \hat{\nu}_x(t) = \mathcal{N}\left( \mu(\hat{\nu}_x(t)),1\right),
\end{align*}
where $\mathcal{N}(a,b)$ is the Gaussian distribution with $a$ and $b$ the value of the mean and the variance, respectively. 
We also denote by 
$\mathcal{E}_{\text{Alt}}(\hat{\mathcal{V}}(t)) \triangleq \{\mathcal{V}'\in\mathcal{E}_{\mathcal{N}}: x^*(\hat{\mathcal{V}}(t))\cap x^{*}(\mathcal{V}')=\emptyset\}$ the set of bandits whose best arm is different from the estimated bandit $\hat{\mathcal{V}}(t)$. 
Then, we define 
\begin{align}
    R_t \triangleq \inf_{\mathcal{V}'\in\mathcal{E}_{\text{Alt}}(\hat{\mathcal{V}}(t))} \sum_{x\in\mathcal{X}\setminus\{0\}} T_x(t) \mathbb{D}(\hat{\nu}_x(t)||\nu_x'). 
\end{align}
The policy $\pi$ is then constructed through the following steps. 
\subsubsection{Stopping Rule $\phi$}
The stopping time is defined as 
\begin{align}
    \tau = \inf\left\{t: R_t> \Gamma_t\right\}, 
\end{align}
where 
\begin{align}
    \Gamma_t = K \log\left((T_{\mathcal{X}\setminus\{0\}}(t))^2 + T_{\mathcal{X}\setminus\{0\}}(t) \right) + f^{-1}(\delta),
\end{align}
\begin{align}
    T_{\mathcal{X}\setminus\{0\}}(t) \triangleq  \sum_{\ell=1}^t \mathbf{1}(X_{\ell}\neq 0)
\end{align}
is the number of pulls on effective arms up to the time $t$, and 
for all $a>0$, 
\begin{align}
    f(a) \triangleq \exp(K-a)\left(\frac{a}{K}\right)^K.  
\end{align}
Note that $f(a)$ is an exponentially decreasing function of the input $a>0$, and this makes $f^{-1}(\delta)= |\log \delta|(1+o_{\delta\rightarrow 0}(1))$, which grows with speed $|\log \delta|$ when $\delta$ is shrinking. 
\subsubsection{Control Policy $\varphi$}
It is known that for the best arm identification algorithm to work properly, each arm should be chosen enough number of times. There are several ways to deal with this. The most straightforward one is introducing an initial phase, in which each arm is picked uniformly. In this work, we consider a different method named $C$-tracking in \cite{Gariver2016}. The idea is to ensure that the control policy at each time $t$ has a non-zero probability in choosing all arms in $\mathcal{X}\setminus\{0\}$. 
Specifically, 
fix any $\zeta>0$, we define 
\begin{align}
    \mathcal{P}_{\mathcal{X}\setminus\{0\}}^{\zeta} \triangleq \{\Bar{P}\in\mathcal{P}_{\mathcal{X}\setminus\{0\}}:\Bar{P}(x) \geq \zeta \textnormal{ for all }x\in\mathcal{X}\setminus\{0\}\}
\end{align}
so that for any $\Bar{P}\in \mathcal{P}_{\mathcal{X}\setminus\{0\}}^{\zeta}$, $\Bar{P}(x)$ is non-zero for all $x\in\mathcal{X}\setminus\{0\}$. 
For any bandit $\mathcal{V}'=\{\nu'_x\}_{x\in\mathcal{X}}$, we define 
\begin{align}
    \Bar{P}_{\mathcal{V}'}^{\zeta} = \argmax_{\Bar{P}\in\mathcal{P}_{\mathcal{X}\setminus \{0\}}^{\zeta}} \frac{\inf_{\mathcal{V}''\in \mathcal{E}_{\text{Alt}}(\mathcal{V}')}\sum_{x\in\mathcal{X}\setminus\{0\}} \Bar{P}(x)\mathbb{D}(\nu_x'||\nu_x'') }{\sqrt{\chi_2\left(\sum_{x\in\mathcal{X}\setminus\{0\}}\Bar{P}(k) q_x||q_0\right)}}
    \label{eqn:def_of_P_theta_bandit}
\end{align}
and 
\begin{align}
    \alpha_{\mathcal{V}'}^{\zeta} = \frac{2\eta }{\chi_2\left(\sum_{x\in\mathcal{X}\setminus\{0\}}\Bar{P}_{\mathcal{V}'}^{\zeta}(x)q_x\middle\Vert q_0\right)} \times \frac{\inf_{\mathcal{V}''\in \mathcal{E}_{\text{Alt}}(\mathcal{V}')}\sum_{x\in\mathcal{X}\setminus\{0\}} \Bar{P}_{\mathcal{V}'}^{\zeta}(x)\mathbb{D}(\nu_x'||\nu_x'') }{|\log \delta|}.
    \label{eqn:def_of_alpha_bandit}
\end{align}
Moreover, for any bandit $\mathcal{V}'$, we also define the distribution $P_{\mathcal{V}'}^{\zeta}\in\mathcal{P}_{\mathcal{X}}^{\zeta}$ as 
\begin{align}
    P_{\mathcal{V}'}^{\zeta}(x) = \begin{cases}
        \alpha_{\mathcal{V}'}^{\zeta}\Bar{P}_{\mathcal{V}'}^{\zeta}(x) \quad \textnormal{if }x \neq 0\\
        1 - \alpha_{\mathcal{V}'}^{\zeta} \quad \textnormal{if }x = 0. 
    \end{cases}
\end{align}
Then, the control policy is given by 
\begin{align}
    P_{X_t|Y^{t-1},X^{t-1}}(x) = 
    \begin{cases}
        P_{\hat{\mathcal{V}}(t-1)}^{\zeta}(x) \textnormal{ if }S_{i-1} \neq \textnormal{stop}\\
        \mathbf{1}(x=0) \quad \textnormal{if } S_{i-1}=\textnormal{stop} 
    \end{cases}
\end{align}
It can be observed effective arms are chosen 
with the probability $\Theta_{\delta\rightarrow 0}(1/|\log \delta|)$ when the policy does not stop. 
\subsubsection{Final Decision Rule $\psi$}
The estimated best arm is given by $x^*(\hat{\mathcal{V}}(\tau))$. 

Similar to what we have done in Section~\ref{sec:prof_thm1}, 
for any policy $\pi$ satisfying \eqref{eqn:control_policy_assumption_3}, 
we also define its corresponding dummy policy as $\tilde{\pi}$, which never stops and has the control policy 
\begin{align}
    P_{X_t|Y^{t-1},X^{t-1}}(x) = P_{\hat{\mathcal{V}}(t-1)}^{}(x)
    \label{eqn:dummy_control_policy2}
\end{align}
for all $x$ and $t\in\mathbb{N}^+$. 

\subsection{Analysis of the Policy}
\subsubsection{Stopping time $\tau$}

Fix any $\zeta>0$ and any $\epsilon>0$, we first define 
\begin{align}
    N_{\mathcal{V}}(\epsilon) &\triangleq \sup\left\{t: \left\Vert\hat{\mathcal{V}}(t)-\mathcal{V}\right\Vert_{\infty} > \epsilon \textnormal{ or } \left\Vert\Bar{P}_{\hat{\mathcal{V}}(t)}^{\zeta} -\Bar{P}_{\mathcal{V}}^{\zeta} \right\Vert_{\infty} > \epsilon \textnormal{ or } \left\vert\frac{\alpha_{\hat{\mathcal{V}}(t)}^{\zeta}}{\alpha_{\mathcal{V}}^{\zeta}}-1\right\vert>\epsilon\right\}. 
\end{align}
The following lemma proved in Appendix~\ref{apx:D} shows that $N_{\mathcal{V}}(\epsilon)$ is upper bounded by $|\log\delta|^{1+\gamma}$ with high probability for any $\epsilon>0$ and $\gamma>0$ under the dummy policy $\Tilde{\pi}$. 

\begin{lemma}
    Let $\pi$ be a policy satisfying \eqref{eqn:control_policy_assumption_4}, and \eqref{eqn:control_policy_assumption_5} is satisfied with some $0<\alpha<1$. 
    Then, for all $\epsilon>0$ and any $\gamma>0$, it holds that 
    \begin{align}
        \mathbb{P}_{\mathcal{V},\Tilde{\pi}}\left( N_{\mathcal{V}}(\epsilon) \geq |\log \delta|^{\alpha+\gamma}\right) = O_{\delta \rightarrow 0}(|\log \delta|^{-\beta})
    \end{align}  
    for arbitrarily large $\beta>0$, where $\Tilde{\pi}$ is the dummy policy of $\pi$ defined in \eqref{eqn:dummy_control_policy2}.  
    \label{lem:bandit_0}
\end{lemma}
One can also observe that 
\begin{align}
    \Gamma_t \geq f^{-1}(\delta) = \Omega_{\delta\rightarrow 0}(|\log \delta|) 
\end{align}
so that $\Gamma_t > a|\log \delta|$ for some $a>0$ for all $\delta$ small enough. 
By using the fact that the policy only pick effective arms with probability $O_{\delta \rightarrow 0}(1/|\log \delta|)$ for any estimated bandit $\hat{\mathcal{V}}(t)$ for all $t\in\mathbb{N}$, the stopping time $\tau$ is lower bounded by $\Omega_{\delta \rightarrow 0}(|\log \delta|^2)$ with high probability as shown in the following lemma. 
\begin{lemma}
    There exists some $b>0$ such that 
    \begin{align}
        \mathbb{P}_{\mathcal{V},\pi}\left(\tau \leq b|\log \delta|^2\right) = O_{\delta\rightarrow 0}(|\log \delta|^{-\beta})
    \end{align}
    for arbitrarily large $\beta>0$. 
    \label{lem:bandit_1}
\end{lemma}
Lemme~\ref{lem:bandit_1} implies that $\tau =\Omega_{\delta\rightarrow 0}(|\log\delta|^2)$ with high probability. We proceed to analyze the value of $R_t$ for all $t\geq \Omega_{\delta\rightarrow 0}(|\log \delta|^2)$. 
Fix any $\epsilon>0$, whenever $N_{\mathcal{V}}(\epsilon)<|\log \delta|^{1+\gamma}$ for some $0<\gamma<1$, there exists some $\xi_1(\epsilon)$, which has the property that $\lim_{\epsilon\rightarrow 0}\xi_1(\epsilon) = 0$, such that 
\begin{align}
    \inf_{\mathcal{V}'\in\mathcal{E}_{\text{Alt}}(\hat{\mathcal{V}}(t))} \sum_{x\in\mathcal{X}\setminus\{0\}} T_x(t) \mathbb{D}(\hat{\nu}_x(t)||\nu_x') \geq \inf_{\mathcal{V}'\in\mathcal{E}_{\text{Alt}}(\mathcal{V})} \sum_{x\in\mathcal{X}\setminus\{0\}} T_x(t) \mathbb{D}(\nu_x||\nu_x')(1-\xi_1(\epsilon))
    \label{eqn:avg_divergence_continuity}
\end{align}
for all $t\geq \Omega_{\delta\rightarrow 0}(|\log \delta|^2)$
when $\delta$ is sufficiently small, 
where we have used the definition of $N_{\mathcal{V}}(\epsilon)$, the properly that the relative entropy is a continuous function of its input, and the equality in \eqref{eqn:eqn:avg_divergence_expression}.  
Fix any $\zeta>0$. By defining $\tau^*_{\zeta}$ as 
\begin{align}
    \tau^*_{\zeta} = \frac{|\log \delta|}{\inf_{\mathcal{V}'\in \mathcal{E}_{\text{Alt}}(\mathcal{V})}\sum_{x\in\mathcal{X}\setminus\{0\}} \Bar{P}^{\zeta}_{\mathcal{V}}(k)\mathbb{D}(\nu_k||\nu_k')}\times \frac{1}{\alpha^{\zeta}_{\mathcal{V}}},
\end{align}
the probability that $\tau>\tau^*_{\zeta}(1+\epsilon')$ for any $\epsilon'>0$ can be written as follows. 
\begin{align}
    &\mathbb{P}_{\mathcal{V},\pi}\left(\tau > \tau^*_{\zeta}(1+\epsilon') \right) \nonumber\\
    &= \mathbb{P}_{\mathcal{V},\Tilde{\pi}}\left(  R_t < \Gamma_t \textnormal{ for all }1\leq t\leq \tau^*_{\zeta}(1+\epsilon')\right) \\
    &\leq \mathbb{P}_{\mathcal{V},\Tilde{\pi}}\left(  R_{\tau^*_{\zeta}(1+\epsilon')} < \Gamma_{\tau^*_{\zeta}(1+\epsilon')} \right)\\
    &\leq \mathbb{P}_{\mathcal{V},\Tilde{\pi}}\Bigg(  \inf_{\mathcal{V}'\in\mathcal{E}_{\text{Alt}}(\hat{\mathcal{V}}(\tau^*_{\zeta}(1+\epsilon')))} \sum_{x\in\mathcal{X}\setminus\{0\}} T_x(\tau^*_{\zeta}(1+\epsilon')) \mathbb{D}(\hat{\nu}_x(\tau^*_{\zeta}(1+\epsilon'))||\nu_x') <  \Gamma_{\tau^*_{\zeta}(1+\epsilon')}\Bigg)\\
    &\leq \mathbb{P}_{\mathcal{V},\Tilde{\pi}}\Bigg(  \inf_{\mathcal{V}'\in\mathcal{E}_{\text{Alt}}(\mathcal{V})} \sum_{x\in\mathcal{X}\setminus\{0\}} T_x(\tau^*_{\zeta}(1+\epsilon')) \mathbb{D}(\nu_x||\nu_x')(1-\xi_1(\epsilon))<  \Gamma_{\tau^*_{\zeta}(1+\epsilon')}, N_{\mathcal{V}}(\epsilon)<|\log \delta|^{1+\gamma}\Bigg) \nonumber\\
    &\quad+ \mathbb{P}_{\mathcal{V},\Tilde{\pi}}\left( N_{\mathcal{V}}(\epsilon) \geq |\log \delta|^{1+\gamma}\right)
    \label{eqn:sec_6_tau_0}
\end{align}
for some $0<\gamma<1$ when $\delta$ is sufficiently small, where \eqref{eqn:sec_6_tau_0} follows from the law of total probability and the fact that $\tau^*_{\zeta}=\Omega_{\delta\rightarrow 0}(|\log \delta|^2)$ so that there is some $\xi_1(\epsilon)$ satisfying the inequality in \eqref{eqn:avg_divergence_continuity} and the property that $\lim_{\epsilon\rightarrow 0}\xi_1(\epsilon) = 0$. 
For any $x\in\mathcal{X}\setminus\{0\}$, we define 
\begin{align}
    J_t(x) \triangleq \sum_{i=1}^t \mathbf{1}(X_i = x) -  \sum_{i=1}^t P_{X_t|X^{t-1},Y^{t-1}}(x)
\end{align}
for each $t\in\mathbb{N}^+$, so that $\{J_t(x)\}_{t=1}^{\infty}$ is a martingale. 
Then, for any $\epsilon''>0$, it holds that
\begin{align}
    &\mathbb{P}_{\mathcal{V},\Tilde{\pi}}\left( J_{\tau^*_{\zeta}(1+\epsilon')}(x) \leq -\epsilon'' \tau^*_{\zeta}(1+\epsilon')\alpha^{\zeta}_{\mathcal{V}} \Bar{P}^{\zeta}_{\mathcal{V}}(x)\right)\nonumber\\
    &\quad \leq \exp\left(-\frac{\Omega_{\delta\rightarrow 0}(|\log \delta|^2)}{O_{\delta\rightarrow 0}(|\log\delta|)+O_{\delta\rightarrow 0}(|\log \delta|)}\right)\label{eqn:sec_6_tau_1}
    ,
\end{align}
where \eqref{eqn:sec_6_tau_1} follows from Freedman's inequality on martingales and the fact that 
$$\tau^*_{\zeta}\alpha^{\zeta}_{\mathcal{V}} = \Theta_{\delta\rightarrow 0}(|\log \delta|)$$ 
and 
\begin{align}
    \sum_{i=1}^{\tau^*_{\zeta}(1+\epsilon')} \mathbb{E}[(J_i(x)-J_{i-1}(x))^2|\mathcal{F}_{i-1}] &\leq O_{\delta\rightarrow 0}(|\log \delta|). 
\end{align}
Note that $T_x(\tau^*_{\zeta}(1+\epsilon'))=\sum_{i=1}^{\tau^*_{\zeta}(1+\epsilon')} \mathbf{1}(X_i=x)$. 
By the union bound, we have 
\begin{align}
    &\mathbb{P}_{\mathcal{V},\Tilde{\pi}}\left(T_x(\tau^*_{\zeta}(1+\epsilon)) \leq \sum_{i=1}^{\tau^*_{\zeta}(1+\epsilon)} P_{X_i|X^{t-1},Y^{i-1}}(x) -\epsilon'' \tau^*_{\zeta}(1+\epsilon)\alpha^{\zeta}_{\mathcal{V}} \Bar{P}^{\zeta}_{\mathcal{V}}(x) \textnormal{ for some }x\in\mathcal{X}\setminus\{0\}\right) \nonumber\\
    &\quad \leq |\mathcal{X}| \mathbb{P}_{\mathcal{V},\Tilde{\pi}}\left( J_{\tau^*(1+\epsilon)}(x) \leq -\epsilon'' \tau^*_{\zeta}(1+\epsilon)\alpha^{\zeta}_{\mathcal{V}} \Bar{P}^{\zeta}_{\mathcal{V}}(x)\right) \\
    &\quad \leq \exp\left(-\Omega_{\delta \rightarrow 0}(|\log \delta|)\right),
\end{align}
where we use the result in \eqref{eqn:sec_6_tau_1}. 
Therefore, for any $\epsilon''>0$, it holds that 
\begin{align}
    &\mathbb{P}_{\mathcal{V},\Tilde{\pi}}\Bigg(
    \sum_{x\in\mathcal{X}\setminus\{0\}}T_x(\tau^*_{\zeta}(1+\epsilon'))\mathbb{D}(\nu_x\Vert \nu_x') \nonumber\\
    &\hspace{2cm}\geq \sum_{x\in\mathcal{X}\setminus\{0\}} \mathbb{D}(\nu_x\Vert \nu_x') \bigg(\sum_{i=1}^{\tau^*_{\zeta}(1+\epsilon')}P_{X_i|X^{t-1},Y^{i-1}}(x) - \epsilon'' \tau^*_{\zeta}(1+\epsilon')\alpha^{\zeta}_{\mathcal{V}}\Bar{P}^{\zeta}_{\mathcal{V}}(x)\bigg)\Bigg) \nonumber\\
    &\quad \geq 1-\exp\left(-\Omega_{\delta \rightarrow 0}(|\log \delta|)\right)
    \label{eqn:sec_6_tau_2}
\end{align}
for all 
$\mathcal{V}' \in \mathcal{E}_{\text{Alt}}(\mathcal{V}) $.  
Moreover, whenever $N_{\mathcal{V}}(\epsilon)=o_{\delta \rightarrow 0}(|\log \delta|^2)$, for each $x\in\mathcal{X}\setminus\{0\}$, 
\begin{align}
    \sum_{i=1}^{\tau^*_{\zeta}(1+\epsilon')} P_{X_i|X^{t-1},Y^{i-1}}(x) &= \sum_{i=1}^{N_{\mathcal{V}}(\epsilon)} P_{X_i|X^{t-1},Y^{i-1}}(x) + \sum_{i=N_{\mathcal{V}}(\epsilon)+1}^{\tau^*_{\zeta}(1+\epsilon')} P_{X_i|X^{t-1},Y^{i-1}}(x) \\
    &\geq \left(\tau^*_{\zeta}(1+\epsilon')-N_{\mathcal{V}}(\epsilon)\right)\alpha^{\zeta}_{\mathcal{V}}\Bar{P}^{\zeta}_{\mathcal{V}}(x) (1-\epsilon)^2 \label{eqn:sec_6_tau_3} \\
    &\geq \alpha^{\zeta}_{\mathcal{V}}\tau^*_{\zeta}(1+\epsilon')\Bar{P}^{\zeta}_{\mathcal{V}}(x)(1-\xi_2(\epsilon))
    \label{eqn:sec_6_tau_4}
\end{align}
under the dummy policy $\tilde{\pi}$
for some $\xi_2(\epsilon)$ when $\delta$ is sufficiently small, where $\xi_2(\epsilon)$ satisfies $\lim_{\epsilon\rightarrow 0}\xi_2(\epsilon)=0$, and in \eqref{eqn:sec_6_tau_3} we use the definition of $N_{\mathcal{V}}(\epsilon)$. 
By combining \eqref{eqn:sec_6_tau_2} and \eqref{eqn:sec_6_tau_4}, for any 
$\epsilon>0$, whenever $N_{\mathcal{V}}(\epsilon)=o_{\delta \rightarrow 0}(|\log \delta|^2)$, it holds that 
\begin{align}
    &\inf_{\mathcal{V}'\in\mathcal{E}_{\textnormal{Alt}}(\mathcal{V})}\sum_{x\in\mathcal{X}\setminus\{0\}}T_x(\tau^*_{\zeta}(1+\epsilon'))\mathbb{D}(\nu_x\Vert \nu_x') \nonumber\\
    &\quad\geq \tau^*_{\zeta}(1+\epsilon') \alpha^{\zeta}_{\mathcal{V}} \inf_{\mathcal{V}'\in\mathcal{E}_{\textnormal{Alt}}(\mathcal{V})}\sum_{x\in\mathcal{X}\setminus\{0\}} \Bar{P}^{\zeta}_{\mathcal{V}}(x)\mathbb{D}(\nu_x\Vert \nu_x')(1-\xi_2(\epsilon)-\epsilon'') \\
    &\quad = |\log \delta| (1+\epsilon')(1-\xi_2(\epsilon)-\epsilon'')
\end{align}
with probability greater than $1-\exp(-\Omega_{\delta\rightarrow 0}(|\log \delta|))$. For any fixed $\epsilon'>0$, we can choose $\epsilon$ and $\epsilon''$ small enough such that 
\begin{align}
    &|\log \delta| (1+\epsilon')(1-\xi_2(\epsilon)-\epsilon'') \nonumber\\
    &\quad\geq K \log\left(\left(T_{\mathcal{X}\setminus\{0\}}(\tau^*_{\zeta}(1+\epsilon'))\right)^2 + T_{\mathcal{X}\setminus\{0\}}(\tau^*_{\zeta}(1+\epsilon')) \right) + f^{-1}(\delta)
    \label{eqn:sec_6_tau_5}
\end{align}
when $\delta$ is small enough, where \eqref{eqn:sec_6_tau_5} follows from the fact that 
\begin{align}
    K \log\left(\left(T_{\mathcal{X}\setminus\{0\}}(\tau^*_{\zeta}(1+\epsilon'))\right)^2 + T_{\mathcal{X}\setminus\{0\}}(\tau^*_{\zeta}(1+\epsilon')) \right) &\leq K\log ((\tau^*_{\zeta})^2 + \tau^*_{\zeta}) \\
    &= O_{\delta\rightarrow 0}\left(\log \left(|\log \delta|^2\right)\right)
\end{align}
and 
\begin{align}
    f^{-1}(\delta) \leq |\log \delta|(1+o_{\delta\rightarrow 0}(1)). 
\end{align}
Therefore, the first term in the right hand side of \eqref{eqn:sec_6_tau_0} is upper bounded by $\exp(-\Omega_{\delta\rightarrow 0}(|\log \delta|))$, which implies that 
\begin{align}
    \mathbb{P}_{\mathcal{V},\pi}\left(\tau > \tau^*_{\zeta}(1+\epsilon') \right) \leq O_{\delta \rightarrow 0}(|\log \delta|^{-\beta}) + \exp(-\Omega_{\delta\rightarrow 0}(|\log \delta|))
\end{align}
for any $\beta>0$ by applying Lemma~\ref{lem:bandit_0}, 
and 
\begin{align}    \lim_{\delta\rightarrow0}\mathbb{P}_{\mathcal{V},\pi}\left(\tau > \tau^*_{\zeta}(1+\epsilon') \right) = 0 
\end{align}
for any $\epsilon'>0$.
This also implies that when $\delta$ is sufficiently small, it holds that 
\begin{align}
    \tau_{\textnormal{sup}}^{(\delta)} < \tau_{\zeta}^{*}. 
\end{align}

\subsubsection{Confidence Analysis}
The confidence analysis follows from the proof of  \cite[Lemma 33.7]{Lattimore2020}. We summarize the main idea behind the proof as below. The event $x^{*}(\hat{\mathcal{V}}(\tau))\neq x^*(\mathcal{V})$ implies $\mathcal{V}\in\mathcal{E}_{\text{Alt}}(\hat{\mathcal{V}}(\tau))$. Then, 
\begin{align}
    &\mathbb{P}_{\mathcal{V},\pi}\left(x^{*}(\hat{\mathcal{V}}(\tau))\neq x^*(\mathcal{V})\right) \\
    &\quad\leq \mathbb{P}_{\mathcal{V},\pi}\left(\mathcal{V}\in\mathcal{E}_{\text{Alt}}(\hat{\mathcal{V}}(\tau)\right)\\
    &\quad \leq \mathbb{P}_{\mathcal{V},\pi}\left(\sum_{x\in\mathcal{X}\setminus\{0\}} T_x(\tau)\mathbb{D}(\hat{\nu}_x(\tau)||\nu_x)> \Gamma_{\tau}\right).
    \label{eqn:conf_1}
\end{align}
In the case of Gaussian bandits, $\mathbb{D}(\hat{\nu}_x(\tau)||\nu_x) = \frac{1}{2}\left(\mu(\hat{\nu}_x(\tau))-\mu(\nu_x)\right)^2$ for all $x\in\mathcal{X}\setminus\{0\}$. The following Lemma gives a concentration bound on the value of $\mu(\hat{\nu}_x(\tau))$. 
\begin{lemma} \cite[Lemma 33.8]{Lattimore2020}
    Let $\{I_t\}_{t=1}^{\infty}$ be a sequence of Gaussian random variables with mean $\mu$ and unit variance. Let $\hat{\mu}_n = \frac{1}{n} \sum_{t=1}^n I_t$ for any $n\in\mathbb{N}^+$, then 
    \begin{align}
        \mathbb{P}\left(\exists n\in\mathbb{N}^+:\frac{n}{2}(\hat{\mu}_n-\mu)^2 \geq \log(1/\delta)+\log(n(n+1)) \right) \leq \delta
    \end{align}
    for any $\delta>0$. 
    \label{lem:bandit_2}
\end{lemma}
For any $\delta>0$, Lemma~\ref{lem:bandit_2} upper bounds the probability that $$T_x(\tau)\mathbb{D}(\hat{\nu}_x(\tau)||\nu_x)\geq \log(1/\delta)+\log\left(T_x(\tau)(T_x(\tau)+1)\right)$$
by $\delta$ for all $x\in\mathcal{X}\setminus\{0\}$  regardless of the value of $T_k(\tau)$. However, the event on the right hand side of (\ref{eqn:conf_1}) is related to the combination of different relative entropy terms. The lemma below extends the result in Lemma~\ref{lem:bandit_2}. 

\begin{lemma} \cite[Proposition 33.9]{Lattimore2020}
    Let $g:\mathbb{N}\mapsto \mathbb{R}$ be increasing, and for each $x\in\mathcal{X}\setminus\{0\}$, let $S_{x}=\{S_{x1}, S_{x2},\cdots\}$ be an infinite sequence of random variables such that for all $\delta\in(0,1)$, 
    \begin{align}
        \mathbb{P}\left(\exists n\in\mathbb{N}:S_{xn} \geq g(n) + \log(1/\delta)\right) \leq \delta.
        \label{eqn:cond_bandit_2}
    \end{align}
Then, provided that the sequences $\{S_{x}\}_{x\in\mathcal{X}\setminus\{0\}}$ are independent from each other, for any $a>0$, it holds that
    \begin{align}
        &\mathbb{P}\left(\exists \mathbf{t}=\{t_1,t_2,\cdots,t_{|\mathcal{X}|-1}\}\in \mathbb{N}^{|\mathcal{X}|-1}:\sum_{x\in\mathcal{X}\setminus\{0\}} S_{xt_x} \geq (|\mathcal{X}|-1)g\left(\sum_{x\in\mathcal{X}\setminus\{0\}} t_x\right) + a\right) \nonumber\\
        &\hspace{2cm}\leq \left(\frac{a}{|\mathcal{X}|-1}\right)^{|\mathcal{X}|-1} \exp(|\mathcal{X}|-1-a). 
    \end{align}
    \label{lem:bandit_3}
\end{lemma}
We would like to apply Lemma~\ref{lem:bandit_3} by identifying $S_x$ as follows. 
\begin{align}
    S_{xn} = \frac{1}{2} \left(\frac{1}{n}\sum_{i=1}^n U_{xn} - \mu(\nu_x)\right)^2
\end{align}
for all $x\in\mathcal{X}\setminus\{0\}$ and $n\in\mathbb{N}^+$, where for all $n\in\mathbb{N}^+$, $U_{xn}$ is the random variable with the distribution $\mathcal{N}(\mu(\nu_x),1)$. Note that $\{U_{xn}\}_{x\in\mathcal{X},n\in\mathbb{N}^+}$ can be viewed as Gaussian random variables which are prepared and generated before the policy starts, and for each $t\in\mathbb{N}^+$, the observation $Y_t$ is just one specific random variable of $\{U_{xn}\}_{x\in\mathcal{X},n\in\mathbb{N}^+}$. 
Since the sequences of empirical means are independent from each other, i.e, $\frac{1}{n}\sum_{i=1}^n U_{xn}$ is independent of $\frac{1}{m}\sum_{i=1}^m U_{x'm}$ for any $x'\neq x$ and $n,m\in\mathbb{N}^+$ 
, it holds that $S_x$ is independent of $S_{x'}$ for any $x'\neq x$. 
We can define $g(n)\triangleq \log(n^2+n)$ for any $n\in\mathbb{N}$  
such that the assumption in Lemma~\ref{lem:bandit_3}, i.e., \eqref{eqn:cond_bandit_2}, is satisfied by the concentration bound in Lemma~\ref{lem:bandit_2}. Then, we can apply Lemma~\ref{lem:bandit_3} to upper bound the right hand side of \eqref{eqn:conf_1} as follows. 
\begin{align}
    &\mathbb{P}_{\mathcal{V},\pi}\left(\sum_{x\in\mathcal{X}\setminus\{0\}} T_x(\tau)\mathbb{D}(\hat{\nu}_x(\tau)||\nu_x) >  \Gamma_{\tau}\right) \nonumber\\
    &\quad \leq  \mathbb{P}_{\mathcal{V},\pi}\left(\sum_{x\in\mathcal{X}\setminus\{0\}} T_x(\tau)\mathbb{D}(\hat{\nu}_x(\tau)||\nu_x)\geq K \log\left((T_{\mathcal{X}\setminus\{0\}}(\tau))^2 + T_{\mathcal{X}\setminus\{0\}}(\tau) \right) + f^{-1}(\delta) \right) \\
    &\quad\leq f(f^{-1}(\delta))\\
    &\quad = \delta,
\end{align}
where we have applied the definition of $f(\cdot)$.

\subsubsection{Relative Entropy Analysis}
Fix any $\zeta>0$, by our stopping time analysis, we know that 
\begin{align}
    \tau_{\textnormal{sup}}^{(\delta)} \leq \tau^*_{\zeta}
\end{align}\
for all $\delta$ sufficiently small. 
The output distribution of $Z^n$ under the dummy policy $\Tilde{\pi}$ is defined as $\tilde{P}_{Z^n}(z^n) \triangleq \mathbb{P}_{\mathcal{V},\mathcal{Q},\Tilde{\pi}}(z^n)$ for any $z^n\in\mathcal{Z}^n$ and $n\in\mathbb{N}^+$. 
Then, for any $\epsilon>0$ and for all $\delta$ small enough, we can use the property of the dummy policy $\Tilde{\pi}$ to upper bound the relative entropy by 
\begin{align}
    \mathbb{D}\left(P_{Z^{\tau_{\textnormal{sup}}^{(\delta)}}}\middle\Vert(q_0)^{\otimes \tau_{\textnormal{sup}}^{(\delta)}}\right) 
    &\leq \mathbb{D}\left(\Tilde{P}_{Z^{\tau^*_{\zeta}}}\middle\Vert(q_0)^{\otimes \tau^*_{\zeta}}\right) \label{eqn:bandit_divergence_upbound_1}\\
    &\leq \sum_{i=1}^{\tau^*_{\zeta}} \mathbb{E}_{Z^{i-1};\mathcal{V},\mathcal{Q},\Tilde{\pi}}\left[\mathbb{D}(\Tilde{P}_{Z_i|Z^{i-1}}\Vert q_0)\right] \\
    &\leq \sum_{i=1}^{\tau^*_{\zeta}} \mathbb{E}_{Z^{i-1};\mathcal{V},\mathcal{Q},\Tilde{\pi}} \mathbb{E}_{\hat{\mathcal{V}}(i-1)|Z^{i-1};\mathcal{V},\mathcal{Q},\Tilde{\pi} } \left[\mathbb{D}\left(\Tilde{P}_{Z_i|Z^{i-1},\hat{\mathcal{V}}(i-1)}\middle\Vert q_0\right) \right]\label{eqn:bandit_divergence_upbound_2}\\
    &\leq  \sum_{i=1}^{\tau^*_{\zeta}} \mathbb{E}_{\hat{\mathcal{V}}(i-1);\mathcal{V},\mathcal{Q},\tilde{\pi}} \mathbb{D}\left(\tilde{P}_{Z_i|\hat{\mathcal{V}}(i-1)}\middle\Vert q_0\right) 
    \label{eqn:bandit_divergence_upbound_3},
\end{align}
where \eqref{eqn:bandit_divergence_upbound_1} comes from the property of $\Tilde{\pi}$ and monotonicity of relative entropy, \eqref{eqn:bandit_divergence_upbound_2} comes from the convexity of relative entropy, and \eqref{eqn:bandit_divergence_upbound_3} follows because the control policy at each time $t$ only depends on $\hat{\mathcal{V}}(t-1)$ and $Z_i$ is independent of $Z_{i-1}$ conditioned on $X_i$. 
Moreover, for any $\epsilon>0$ and $i\in\mathbb{N}^+$, if $||\hat{\mathcal{V}}(i)-\mathcal{V}||_{\infty}<\epsilon$,  there exists some $\xi_3(\epsilon)$ and $\xi_4(\epsilon)$ such that 
\begin{align}
    ||\Bar{P}_{\hat{\mathcal{V}}(i)} - \Bar{P}^{\zeta}_{\mathcal{V}}||_{\infty} < \xi_3(\epsilon)
    \quad \textnormal{ and }\quad 
    \left|\frac{\alpha_{\hat{\mathcal{V}}(i)}^{\zeta}}{\alpha^{\zeta}_{\mathcal{V}}} - 1\right| < \xi_4(\epsilon). 
\end{align}
For any $\epsilon>0$, $\Bar{P}\in\mathcal{P}_{\mathcal{X}\setminus\{0\}}$ and $\alpha\in\mathbb{R}$, we define 
\begin{align}
    \mathcal{B}(\Bar{P},\epsilon) &\triangleq \left\{\Bar{P}'\in \mathcal{P}_{\mathcal{X}\setminus\{0\}}: \Vert \Bar{P}' - \Bar{P}\Vert_{\infty} \leq \epsilon \right\} \\
    \mathcal{B}(\alpha,\epsilon) &\triangleq \left\{\alpha'\in \mathbb{R}: \left\Vert \frac{\alpha'}{\alpha} -1 \right\Vert \leq \epsilon \right\}.
\end{align}
Then, for all $\delta$ sufficiently small, we have
\begin{align}
    &\sum_{i=1}^{\tau^*_{\zeta}} \mathbb{E}_{\hat{\mathcal{V}}(i-1);\mathcal{V},\mathcal{Q},\Tilde{\pi}} \mathbb{D}\left(\tilde{P}_{Z_i|\hat{\mathcal{V}}(i-1)}\middle\Vert q_0\right) \\
    &\leq \sum_{i=1}^{\tau^*_{\zeta}} \mathbb{P}_{\mathcal{V},\Tilde{\pi}}(||\hat{\mathcal{V}}(i-1)-\mathcal{V}||_{\infty}>\epsilon) O_{\delta\rightarrow 0}(1/|\log \delta|^2)\nonumber\\
    &\quad+ \sum_{i=1}^{\tau^*_{\zeta}} \mathbb{P}_{\mathcal{V},\Tilde{\pi}}(||\hat{\mathcal{V}}(i-1)-\mathcal{V}||_{\infty}\leq \epsilon) \nonumber \\
    &\quad\quad\times \max_{\Bar{P}\in\mathcal{B}(\Bar{P}^{\zeta}_{\mathcal{V}},\xi_3(\epsilon))}\max_{\alpha\in\mathcal{B}(\alpha^{\zeta}_{\mathcal{V}},\xi_4(\epsilon))}\mathbb{D}\left((1-\alpha)q_0 + \alpha\left(\sum_{x\in\mathcal{X}\setminus\{0\}}\Bar{P}(x)q_x\right)\middle\Vert q_0\right) \label{eqn:bandit_divergence_upbound_4}\\
    &\leq O_{\delta\rightarrow 0}(1/|\log \delta|^2) \left(|\log \delta|^{3/2}+\sum_{i=|\log \delta|^{3/2}+1}^{\tau^*_{\zeta}} \mathbb{P}_{\mathcal{V},\Tilde{\pi}}( ||\hat{\mathcal{V}}(i-1)-\mathcal{V}||_{\infty} \geq \epsilon) \right) \nonumber\\
    &\quad+ \tau^*_{\zeta} \max_{\Bar{P}\in\mathcal{B}(\Bar{P}^{\zeta}_{\mathcal{V}},\xi_3(\epsilon))}\max_{\alpha\in\mathcal{B}(\alpha^{\zeta}_{\mathcal{V}},\xi_4(\epsilon))}\mathbb{D}\left((1-\alpha)q_0 + \alpha\left(\sum_{x\in\mathcal{X}\setminus\{0\}}\Bar{P}(x)q_x\right)\middle\Vert q_0\right)
    \\
    &\leq o_{\zeta\rightarrow 0}(1) + \tau^*_{\zeta}\left(\frac{(\alpha^{\zeta}_{\mathcal{V}})^2} {2}\chi_2\left(\sum_{x\in\mathcal{X}\setminus\{0\}}\Bar{P}^{\zeta}_{\mathcal{V}}(x)q_x\middle\Vert q_0\right)\left(1+\xi_5(\epsilon)\right) + o_{\delta\rightarrow 0}((\alpha^{\zeta}_{\mathcal{V}})^2) \right)\label{eqn:bandit_divergence_upbound_5}\\
    &\leq \eta(1+\xi_5(\epsilon))(1+o_{\delta\rightarrow 0}(1)),
\end{align}
where \eqref{eqn:bandit_divergence_upbound_4} follows because non-null actions are chosen with the probability $\Theta_{\delta\rightarrow 0}(|\log\delta|^{-1})$ so that $\mathbb{D}\left(\tilde{P}_{Z_i|\hat{\mathcal{V}}(i-1)}\middle\Vert q_0\right)\leq O_{\delta\rightarrow 0}(1/|\log\delta|^2)$ by \cite[Lemma 1]{Bloch2016} regardless of $\hat{\mathcal{V}}(i-1)$, and in \eqref{eqn:bandit_divergence_upbound_5} we apply Lemma~\ref{lem:bandit_0} and \cite[Lemma 1]{Bloch2016} again and use the continuity property of relative entropy, where $\xi_5(\epsilon)$ is some function such that $\lim_{\epsilon\rightarrow 0}\xi_5(\epsilon)=0$. 
Finally, by making $\epsilon$ arbitrarily small, we claim that 
\begin{align}
    \lim_{\delta\rightarrow 0} \mathbb{D}\left(P_{Z^{\tau_{\textnormal{sup}}^{(\delta)}}}\middle\Vert(q_0)^{\tau_{\textnormal{sup}}^{(\delta)}}\right)  \leq \eta
\end{align}

\subsubsection{Claim of Achievable Exponent}

From the definition of the exponent in the covert best arm identification problem, we have 
\begin{align}
    \gamma_{2}^* &\geq \gamma_{2}(\pi) \\
    &\geq \liminf_{\delta \rightarrow 0}  \frac{-\log \delta}{\sqrt{\tau^{*}_{\zeta}}}\\
    &\geq \sqrt{2\eta} \frac{\inf_{\mathcal{V}'\in \mathcal{E}_{\text{Alt}}(\mathcal{V})}\sum_{x\in\mathcal{X}\setminus\{0\}} \Bar{P}^{\zeta}_{\mathcal{V}}(x)\mathbb{D}(\nu_x||\nu_x') }{\sqrt{\chi_2\left(\sum_{x\in\mathcal{X}\setminus\{0\}}\Bar{P}^{\zeta}_{\mathcal{V}}(k) q_x||q_0\right)}} \\
    &= \sqrt{2\eta}\argmax_{\Bar{P}\in\mathcal{P}_{\mathcal{X}\setminus \{0\}}^{\zeta}} \frac{\inf_{\mathcal{V}'\in \mathcal{E}_{\text{Alt}}(\mathcal{V})}\sum_{x\in\mathcal{X}\setminus\{0\}} \Bar{P}(x)\mathbb{D}(\nu_x||\nu_x') }{\sqrt{\chi_2\left(\sum_{x\in\mathcal{X}\setminus\{0\}}\Bar{P}(k) q_x||q_0\right)}} 
\end{align}
for any $\zeta>0$. 
Finally, we complete the proof by making $\zeta>0$ arbitrarily small and using the continuity property of the relative entropy and the chi-square distance. 

\section{Proof of Theorem~\ref{thm:main_result4}}
\label{sec:prof_thm4}

We define the event $\mathcal{E}\triangleq \{\psi(X^{\tau},Y^{\tau}) \neq x^*(\mathcal{V}) \}$. For any $\delta>0$, by the confidence constraint, we have 
\begin{align}
    \delta \geq \mathbb{P}_{\mathcal{V},\pi}(\mathcal{E})
\end{align}
and 
\begin{align}
    \delta \geq \mathbb{P}_{\mathcal{V}',\pi}(\psi(X^{\tau},Y^{\tau})\neq x^*(\mathcal{V}')) \geq \mathbb{P}_{\mathcal{V}',\pi}(\mathcal{E}^c)
\end{align}
for any $\mathcal{V}'\in\mathcal{E}_{\textnormal{Alt}}(\mathcal{V})$ and for any $\pi \in\Lambda_2(\eta)$. Then, by Bretagnolle–Huber inequality and the relative entropy decomposition lemma, we have 
\begin{align}
    2\delta &\geq \mathbb{P}_{\mathcal{V},\pi}(\mathcal{E}) + \mathbb{P}_{\mathcal{V}',\pi}(\mathcal{E}^c) \\
    &\geq \frac{1}{2}\exp\left(-\sum_{x\in\mathcal{X}\setminus\{0\}} \mathbb{E}_{\mathcal{V},\pi}[T_x(\tau)] \mathbb{D}(\nu_x\Vert \nu_x')\right)
        \label{eqn:prof_thm4_1}
\end{align}
for any $\mathcal{V}'\in\mathcal{E}_{\textnormal{Alt}}(\mathcal{V})$, which is equivalent to 
\begin{align}
    \min_{\mathcal{V}'\in\mathcal{E}_{\textnormal{Alt}}(\mathcal{V})} \sum_{x\in\mathcal{X}\setminus\{0\}} \mathbb{E}_{\mathcal{V},\pi}[T_x(\tau)] \mathbb{D}(\nu_x\Vert \nu_x') \geq |\log 4\delta|. 
    \label{eqn:prof_thm4_divergence_condition}
\end{align}
Note that for any $x\in\mathcal{X}\setminus\{0\}$, we always can find a bandit $\mathcal{V}'\in\mathcal{E}_{\textnormal{Alt}}(\mathcal{V})$ such that $\mu(\nu_{x'})=\mu(\nu_{x'}')$ for all $x'\in\mathcal{X}\setminus\{0,x\}$, i.e., by defining $\mu(\nu'_x) > \max_{x\in\mathcal{X}\setminus\{0\}}\mu(\nu_x)$ 
and $\mu(\nu_{x'}')=\mu(\nu_{x'})$ for all $x'\in\mathcal{X}\setminus\{0\}$.  Therefore, the relative entropy $\mathbb{D}(\nu_{x'}\Vert \nu'_{x'})=0$ under this specific bandit for all $x'\in\mathcal{X}\setminus\{0,x\}$, and this implies that 
\begin{align}
    \mathbb{E}_{\mathcal{V},\pi}[T_x(\tau)] = \Omega_{\delta \rightarrow 0}(|\log \delta|)
    \label{eqn:prof_thm4_0}
\end{align}
for all $x\in\mathcal{X}\setminus\{0\}$ in order to satisfy \eqref{eqn:prof_thm4_divergence_condition}. 
Let the policy $\pi$ satisfy \eqref{eqn:control_policy_assumption_3}-\eqref{eqn:control_policy_assumption_5} with some $0<\alpha<1$ and $\Tilde{\pi}$ be the corresponding dummy policy of $\pi$. 
For any $\epsilon>0$, we rewrite $\mathbb{E}_{\mathcal{V},\pi}[T_x(\tau)]$ as follows.
\begin{align}
    \mathbb{E}_{\mathcal{V},\pi}[T_x(\tau)] = \sum_{i=1}^{|\log \delta|^{\alpha+\epsilon}} \mathbb{E}_{\mathcal{V},\pi}[\mathbf{1}(X_i=x,\tau\geq i)] + \sum_{|\log \delta|^{\alpha+\epsilon}+1}^{\infty} \mathbb{E}_{\mathcal{V},\pi}[\mathbf{1}(X_i=x,\tau\geq i)],
    \label{eqn:prof_thm4_1}
\end{align}
where $\alpha$ is the value assumed in \eqref{eqn:control_policy_assumption_5}. 
The first summation in \eqref{eqn:prof_thm4_1} is upper bounded by 
\begin{align}
    \sum_{i=1}^{|\log \delta|^{\alpha+\epsilon}} \mathbb{E}_{\mathcal{V},\pi}[\mathbf{1}(X_i=x,\tau\geq i)] &\leq |\log \delta|^{\alpha+\epsilon} \max_{1\leq i\leq |\log \delta|^{\alpha+\epsilon}}\mathbb{P}_{\mathcal{V},\pi}(X_i=x,\tau\geq i) \\
    &\leq \Tilde{O}_{\delta\rightarrow 0}(|\log \delta|^{\epsilon}) 
    \label{eqn:prof_thm4_2}
\end{align}
by the assumption in \eqref{eqn:control_policy_assumption_4} and \eqref{eqn:control_policy_assumption_5}. To analyze the second term summation in \eqref{eqn:prof_thm4_1}, for any $1>\epsilon>0$, we first define the event 
\begin{align}
    \mathcal{G}_{t} \triangleq \{||\hat{\mathcal{V}}(t)-\mathcal{V}||_{\infty} \leq \epsilon \}. 
\end{align}
It is shown in Appendix~\ref{apx:bounds_of_G} that 
\begin{align}
    \mathbb{P}_{\mathcal{V},\Tilde{\pi}}(\mathcal{G}_t) \geq 1-\exp\left(-t\Tilde{\Omega}_{\delta\rightarrow 0}(|\log \delta|^{-\alpha-\epsilon/2})\right)
    \label{eqn:prof_thm4_bounds_of_G}
\end{align}
for any $t\geq |\log\delta|^{\alpha+\epsilon}$. 
Then, for all $i > |\log \delta|^{\alpha+\epsilon}$, it holds that 
\begin{align}
    &\mathbb{P}_{\mathcal{V},\pi}\left(X_i=x,\tau\geq i\right) \nonumber\\
    &= \mathbb{P}_{\mathcal{V},\pi}\left(X_i=x,\tau\geq i,\mathcal{G}_{i-1}\right) + \mathbb{P}_{\mathcal{V},\pi}\left(X_i=x,\tau\geq i,\mathcal{G}_{i-1}^c\right) \\
    &\leq  \mathbb{P}_{\mathcal{V},\pi}\left(X_i=x,\tau\geq i \middle \vert \mathcal{G}_{i-1}\right)\mathbb{P}_{\mathcal{V},\pi}(\mathcal{G}_{i-1}) + \mathbb{P}_{\mathcal{V},\Tilde{\pi}}(\mathcal{G}_{i-1}^{c})\\
    &\leq \mathbb{P}_{\mathcal{V},\pi}\left(X_i=x\middle\vert \tau\geq i, \mathcal{G}_{i-1}\right)\mathbb{P}_{\mathcal{V},\pi}(\tau \geq i\vert \mathcal{G}_{i-1}) \mathbb{P}_{\mathcal{V},\pi}(\mathcal{G}_{i-1}) + \mathbb{P}_{\mathcal{V},\Tilde{\pi}}(\mathcal{G}_{i-1}^{c})\\
    &\leq \max_{\{\mathcal{V}'':||\mathcal{V}''-\mathcal{V}||_{\infty}\leq \epsilon\}}P_{\mathcal{V}''}(x) \mathbb{P}_{\mathcal{V},\pi}(\tau \geq i) + \mathbb{P}_{\mathcal{V},\Tilde{\pi}}(\mathcal{G}^{c}_{i-1}), 
\end{align}
where we use the fact that 
\begin{align}
\mathbb{P}_{\mathcal{V},\pi}\left(X_i=x\middle\vert \tau\geq i, \mathcal{G}_{i-1}\right) \leq \max_{\{\mathcal{V}'':||\mathcal{V}''-\mathcal{V}||_{\infty}\leq \epsilon\}}P_{\mathcal{V}''}(x)
\end{align}
for any $x\in\mathcal{X}$
by the assumption on the policy. 
Then, the second summation in \eqref{eqn:prof_thm4_1} can be upper bounded by 
\begin{align}
    &\sum_{|\log \delta|^{\alpha+\epsilon}+1}^{\infty} \mathbb{E}_{\mathcal{V},\pi}[\mathbf{1}(X_i=x,\tau\geq i)] \nonumber\\
    &\quad \leq \sum_{i=|\log \delta|^{\alpha+\epsilon}+1}^{\infty} \left( \max_{\{\mathcal{V}'':||\mathcal{V}''-\mathcal{V}||_{\infty}\leq \epsilon\}}P_{\mathcal{V}''}(x) \mathbb{P}_{\mathcal{V},\pi}(\tau \geq i) + \mathbb{P}_{\mathcal{V},\Tilde{\pi}}(\mathcal{G}^{c}_{i-1})\right)\\
    &\quad \leq  \mathbb{E}_{\mathcal{V},\pi}[\tau]\max_{\{\mathcal{V}'':||\mathcal{V}''-\mathcal{V}||_{\infty}\leq \epsilon\}}P_{\mathcal{V}''}(x) + o_{\delta \rightarrow 0}(1), 
    \label{eqn:prof_thm4_3}
\end{align}
where we use \eqref{eqn:prof_thm4_bounds_of_G} and the fact that $\sum_{i=|\log \delta|^{\alpha+\epsilon}+1}^{\infty} \mathbb{P}_{\mathcal{V},\pi}(\tau \geq i) \leq \sum_{i=1}^{\infty} \mathbb{P}_{\mathcal{V},\pi}(\tau \geq i) = \mathbb{E}_{\mathcal{V},\pi}[\tau]$. 
Therefore, by combining \eqref{eqn:prof_thm4_1}, \eqref{eqn:prof_thm4_2} and \eqref{eqn:prof_thm4_3}, we have that 
\begin{align}
    \mathbb{E}_{\mathcal{V},\pi}[T_x(\tau)] \leq \mathbb{E}_{\mathcal{V},\pi}[\tau]\max_{\{\mathcal{V}'':||\mathcal{V}''-\mathcal{V}||_{\infty}\leq \epsilon\}}P_{\mathcal{V}''}(x)(1+o_{\delta\rightarrow 0}(1))
    \label{eqn:prof_thm4_4}
\end{align}
by the fact that $\mathbb{E}_{\mathcal{V},\pi}[T_x(\tau)] = \Omega_{\delta \rightarrow 0}(|\log \delta|)$ for all $x\in\mathcal{X}\setminus\{0\}$ as mentioned in \eqref{eqn:prof_thm4_0} so that $\Tilde{O}_{\delta\rightarrow 0}(|\log \delta|^{\epsilon})=o_{\delta\rightarrow 0}(\mathbb{E}_{\mathcal{V},\pi}[T_x(\tau)] )$. Equations \eqref{eqn:prof_thm4_4} and \eqref{eqn:prof_thm4_divergence_condition} also imply that 
\begin{align}
     \mathbb{E}_{\mathcal{V},\pi}[\tau](1+o(1)) \min_{\mathcal{V}'\in\mathcal{E}_{\textnormal{Alt}}(\mathcal{V})} \sum_{x\in\mathcal{X}\setminus\{0\}} \max_{\{\mathcal{V}'':||\mathcal{V}''-\mathcal{V}||_{\infty}\leq \epsilon\}}P_{\mathcal{V}''}(x)\mathbb{D}(\nu_x \Vert \nu_x') \geq |\log 4\delta|. 
\end{align}
By the assumption that the stopping time concentrates \eqref{eqn:tau_concentrate}, it holds that for any $\epsilon'>0$, 
\begin{align}
    \lim_{\delta\rightarrow 0} \mathbb{P}_{\mathcal{V},\pi}\left(\tau < \frac{|\log 4\delta|(1-\epsilon')}{\min_{\mathcal{V}'\in\mathcal{E}_{\textnormal{Alt}}(\mathcal{V})} \sum_{x\in\mathcal{X}\setminus\{0\}} \max_{\{\mathcal{V}'':||\mathcal{V}''-\mathcal{V}||_{\infty}\leq \epsilon\}}P_{\mathcal{V}''}(x)\mathbb{D}(\nu_x \Vert \nu_x')}\right) = 0
    \label{eqn:prof_thm4_5}
\end{align}
For any $\epsilon'>0$, we define 
\begin{align}
    \tilde{\tau}^{(\delta)} = \frac{|\log 4\delta|(1-\epsilon')}{\min_{\mathcal{V}'\in\mathcal{E}_{\textnormal{Alt}}(\mathcal{V})} \sum_{x\in\mathcal{X}\setminus\{0\}} \max_{\{\mathcal{V}'':||\mathcal{V}''-\mathcal{V}||_{\infty}\leq \epsilon\}}P_{\mathcal{V}''}(x)\mathbb{D}(\nu_x \Vert \nu_x')}.
\end{align}
Then, \eqref{eqn:prof_thm4_5} and the definition of $\tau^{\delta}_{\textnormal{sup}}$ implies that 
\begin{align}
    \tau_{\textnormal{sup}}^{(\delta)} \geq \tilde{\tau}^{(\delta)}
\end{align}
for all $\delta$ sufficiently small.  Then, the covertness constraint can be written as 
\begin{align}
    \eta &\geq \lim_{\delta \rightarrow 0} \mathbb{D}\left(P_{Z^{\tau_{\textnormal{sup}}^{(\delta)}}}\middle\Vert (q_{}^{0})^{\otimes \tau_{\textnormal{sup}}^{(\delta)}}\right) \\
    &\geq \lim_{\delta \rightarrow 0}\mathbb{D}\left(P_{Z^{\tilde{\tau}^{(\delta)}}}\middle\Vert (q_{0})^{\otimes \tilde{\tau}^{(\delta)}}\right) \\
    &\geq \lim_{\delta\rightarrow 0} \Tilde{\tau}^{(\delta)}\mathbb{D}\left(\Bar{P}_{Z}\middle\Vert q_0 \right),
\end{align}
where for any $z\in\mathcal{Z}$
\begin{align}
    \Bar{P}_Z(z) &\triangleq  \frac{1}{\Tilde{\tau}^{(\delta)}}\sum_{i=1}^{\tilde{\tau}^{(\delta)}} P_{Z_i}(z) \\
    &= \frac{1}{\Tilde{\tau}^{(\delta)}} \sum_{i=1}^{\tilde{\tau}^{(\delta)}} \Bigg( \int_{\mathcal{V}''}\mathbb{P}_{\mathcal{V},\pi}(\hat{\mathcal{V}}(i-1)=\mathcal{V}'',\tau \geq i)\bigg(\sum_{x\in\mathcal{X}} P_{\mathcal{V}'}(x)q_x(z)\bigg) \partial \mathcal{V}'' \nonumber\\
    &\hspace{5cm}+ \mathbb{P}_{\mathcal{V},\pi}(\tau<i) q_0(z)\Bigg). 
\end{align}
For any $\epsilon>0$, let 
\begin{align}
    \mathcal{B}(\mathcal{V},\epsilon) \triangleq \left\{\mathcal{V}'\in\mathcal{E}_{\mathcal{N}}: \Vert \mathcal{V}-\mathcal{V}'\Vert_{\infty} \leq \epsilon \right\}
\end{align}
and 
\begin{align}
    \mathcal{V}_{\textnormal{min}} \triangleq \argmin_{\mathcal{V}''\in \mathcal{B}(\mathcal{V},\epsilon)} \mathbb{D}\left(\sum_{x\in\mathcal{X}\setminus\{0\}} P_{\mathcal{V}''}(x)q_x(z)\middle\Vert q_0\right). 
\end{align}
Then, for any $\epsilon>0$, we also define 
\begin{align}
    \widehat{P}_Z(z) &\triangleq \frac{1}{\Tilde{\tau}^{(\delta)}} \sum_{i=1}^{\tilde{\tau}^{(\delta)}} \Bigg( \mathbb{P}_{\mathcal{V},\pi}(||\hat{\mathcal{V}}(i-1)-\mathcal{V}||_{\infty}\leq \epsilon,\tau \geq i)\bigg(\sum_{x\in\mathcal{X}} P_{\mathcal{V}_{\textnormal{min}}}(x)q_x(z)\bigg)  \nonumber\\
    &\hspace{5cm}+ \mathbb{P}_{\mathcal{V},\pi}(||\hat{\mathcal{V}}(i-1)-\mathcal{V}||_{\infty} > \epsilon \textnormal{ or } \tau<i) q_0(z)\Bigg)
\end{align}
so that 
\begin{align}
    \mathbb{D}\left(\Bar{P}_{Z}\middle\Vert q_0 \right) \geq \mathbb{D}\left(\widehat{P}_{Z}\middle\Vert q_0 \right)
\end{align}
because of the definition of $\mathcal{V}_{\textnormal{min}}$. 
We can further rewrite $\widehat{P}_Z$ as follows.
\begin{align}
    \widehat{P}_Z(z) = \widehat{\alpha}\left(\sum_{x\in\mathcal{X}\setminus\{0\}}\widehat{P}_{X}(x)q_x(z)\right) + (1-\widehat{\alpha})q_0(z),
\end{align}
where 
\begin{align}
    \widehat{\alpha} \triangleq \frac{1}{\Tilde{\tau}^{(\delta)}} \sum_{i=1}^{\Tilde{\tau}^{(\delta)}} \left(\mathbb{P}_{\mathcal{V},\pi}(||\hat{\mathcal{V}}(i-1)-\mathcal{V}||_{\infty}\leq \epsilon,\tau \geq i) \right)\sum_{x\in\mathcal{X}\setminus\{0\}} P_{\mathcal{V}_{\textnormal{min}}}(x)
\end{align}
and for any $x\in\mathcal{X}\setminus\{0\}$
\begin{align}
    \widehat{P}_X(x) &= \frac{\frac{1}{\Tilde{\tau}^{(\delta)}} \sum_{i=1}^{\Tilde{\tau}^{(\delta)}} \left(\mathbb{P}_{\mathcal{V},\pi}(||\hat{\mathcal{V}}(i-1)-\mathcal{V}||_{\infty}\leq \epsilon,\tau \geq i)  P_{\mathcal{V}_{\textnormal{min}}}(x)\right)}{\tilde{\alpha}}\\
    &= \frac{P_{\mathcal{V}_{\textnormal{min}}}(x)}{\sum_{x\in\mathcal{X}\setminus\{0\}} P_{\mathcal{V}_{\textnormal{min}}}(x)} \\
    &\triangleq \Bar{P}_{\mathcal{V}_{\textnormal{min}}}(x). 
    \label{eqn:tilde_P_X}
\end{align}
Note that when $\delta$ is sufficiently small, it holds that
\begin{align}
    &\frac{1}{\Tilde{\tau}^{(\delta)}} \sum_{i=1}^{\Tilde{\tau}^{(\delta)}} \left(\mathbb{P}_{\mathcal{V},\pi}(||\hat{\mathcal{V}}(i-1)-\mathcal{V}||_{\infty}\leq \epsilon,\tau \geq i) \right) \nonumber\\
    &\quad \geq \frac{1}{\Tilde{\tau}^{(\delta)}} \sum_{i=1}^{\Tilde{\tau}^{(\delta)}} \left(1 - \mathbb{P}_{\mathcal{V},\pi}(||\hat{\mathcal{V}}(i-1)-\mathcal{V}||_{\infty} >  \epsilon,\tau \geq i) - \mathbb{P}_{\mathcal{V},\pi}(\tau < i)\right) \\
    &\quad \geq 1 - \frac{1}{\Tilde{\tau}^{(\delta)}}\left(\sum_{i=1}^{|\log \delta|^{\alpha+\epsilon}} 1 + \sum_{i=|\log \delta|^{\alpha+\epsilon}+1}^{\tilde{\tau}^{(\delta)}} \mathbb{P}_{\mathcal{V},\Tilde{\pi}}(\Vert \hat{\mathcal{V}}(i-1)-\mathcal{V}\Vert_{\infty }> \epsilon)\right) - o_{\delta\rightarrow 0}(1) \label{eqn:lowerbound_alpha_tilde_1}\\
    &\quad \geq 1 - \frac{1}{\Tilde{\tau}^{(\delta)}}\left(\sum_{i=1}^{|\log \delta|^{\alpha+\epsilon}} 1 + \sum_{i=|\log \delta|^{\alpha+\epsilon}+1}^{\tilde{\tau}^{(\delta)}} + \exp\left(-i\Tilde{\Omega}(|\log \delta|^{-\alpha})\right)\right) - o_{\delta\rightarrow 0}(1)\label{eqn:lowerbound_alpha_tilde_2}\\
    &\quad \geq 1- o_{\delta\rightarrow 0}(1),
\end{align}
where in \eqref{eqn:lowerbound_alpha_tilde_1} we use \eqref{eqn:prof_thm4_5} and the fact that the average of a sequence which converges to zero is zero, and in \eqref{eqn:lowerbound_alpha_tilde_2} we use the result in \eqref{eqn:prof_thm4_bounds_of_G}. 
Above inequality implies that
\begin{align}
    \widehat{\alpha} = (1-o_{\delta\rightarrow 0}(1))\left(\sum_{x\in\mathcal{X}\setminus\{0\}}P_{\mathcal{V}_{\textnormal{min}}}(x)\right).
    \label{eqn:alpha_tilde_lowerbound}
\end{align}
Then, for any $\delta$ sufficiently small, 
\begin{align}
    \Tilde{\tau}^{(\delta)}\mathbb{D}\left(\widehat{P}_{Z}\middle\Vert q_0 \right) \geq \tilde{\tau}^{(\delta)}\frac{\widehat{\alpha}^2}{2}\chi_2\left(\sum_{x\in\mathcal{X}\setminus\{0\}}\widehat{P}_X(x)q_x\middle\Vert q_0\right)(1-o_{\delta\rightarrow 0}(1)) 
    \label{eqn:relative_entropy_lowerbound}
\end{align}
by using \cite[Lemma 1]{Bloch2016}. 
By the continuity property of $\{P_{\mathcal{V}''}\}_{\mathcal{V}''}$, for any $\epsilon>0$, there exists some $\xi_5(\epsilon)$ 
such that 
\begin{align}
    \left|\frac{P_{\mathcal{V}''}(x)}{P_{\mathcal{V}_{\textnormal{min}}}(x)}
    - 1\right| \leq \xi_5(\epsilon)
\end{align}
for any $x\in\mathcal{X}$ for any $\mathcal{V}''\in\mathcal{B}(\mathcal{V},\epsilon)$. Then, we have
\begin{align}
    \tilde{\tau}^{(\delta)} \geq \frac{|\log 4\delta|(1-\epsilon')}{(1+\xi_5(\epsilon))\min_{\mathcal{V}'\in\mathcal{E}_{\textnormal{Alt}}(\mathcal{V})} \sum_{x\in\mathcal{X}\setminus\{0\}} P_{\mathcal{V}_{\textnormal{min}}}(x)\mathbb{D}(\nu_x \Vert \nu_x')}
    \label{eqn:tilde_tau_lowerbound}
\end{align}
and
\begin{align}
    \widehat{\alpha} \geq \frac{|\log 4\delta|(1-\epsilon')(1-o_{\delta \rightarrow 0}(1))}{(1+\xi_5(\epsilon))\tilde{\tau}^{(\delta)}\min_{\mathcal{V}'\in\mathcal{E}_{\textnormal{Alt}}(\mathcal{V})} \sum_{x\in\mathcal{X}\setminus\{0\}} \Bar{P}_{\mathcal{V}_{\textnormal{min}}}(x)\mathbb{D}(\nu_x \Vert \nu_x')} 
    \label{eqn:alpha_tilde_lowerbound_2}
\end{align}
by using \eqref{eqn:tilde_tau_lowerbound}, \eqref{eqn:alpha_tilde_lowerbound} and the definition of $\Bar{P}_{X;\mathcal{V}_{\textnormal{min}}}$. 
Therefore, for all $\delta$ sufficiently small,
\begin{align}
    \Tilde{\tau}^{(\delta)}\mathbb{D}\left(\widehat{P}_{Z}\middle\Vert q_0 \right) \geq \frac{|\log 4\delta|^2(1-\epsilon')^2 \chi_2\left(\sum_{x\in\mathcal{X}\setminus\{0\}}\Bar{P}_{\mathcal{V}_{\textnormal{min}}}(x)q_x\middle\Vert q_0\right) (1-o_{\delta\rightarrow 0}(1))^3}{2(1+\xi_5(\epsilon))^2\tilde{\tau}^{(\delta)} \left(\min_{\mathcal{V}'\in\mathcal{E}_{\textnormal{Alt}}(\mathcal{V})} \sum_{x\in\mathcal{X}\setminus\{0\}} \Bar{P}_{\mathcal{V}_{\textnormal{min}}}(x)\mathbb{D}(\nu_x \Vert \nu_x')\right)^2}
\end{align}
by combining \eqref{eqn:tilde_P_X}, \eqref{eqn:alpha_tilde_lowerbound}, \eqref{eqn:relative_entropy_lowerbound}, \eqref{eqn:tilde_tau_lowerbound} and \eqref{eqn:alpha_tilde_lowerbound_2}. 
By the covertness constraint, we have for any $\epsilon>0$ and $\epsilon'>0$, 
\begin{align}
    \lim_{\delta\rightarrow 0} \frac{|\log 4\delta|}{\sqrt{\tilde{\tau}^{\delta}}} \leq \sqrt{2\eta} \frac{\min_{\mathcal{V}'\in\mathcal{E}_{\textnormal{Alt}}(\mathcal{V})} \sum_{x\in\mathcal{X}\setminus\{0\}} \Bar{P}_{\mathcal{V}_{\textnormal{min}}}(x)\mathbb{D}(\nu_x \Vert \nu_x')}{\sqrt{\chi_2\left(\sum_{x\in\mathcal{X}\setminus\{0\}}\Bar{P}_{\mathcal{V}_{\textnormal{min}}}(x)q_x\middle\Vert q_0\right)}} \frac{(1+\xi_5(\epsilon))}{(1-\epsilon')}
\end{align}
Finally, we have $\tau_{\textnormal{sup}}^{(\delta)} \geq \tilde{\tau}^{(\delta)}$ for all $\delta$ sufficiently small, and we can choose $\epsilon>0$ and $\epsilon'>0$ arbitrarily small. Then, 
\begin{align}
    \gamma_2(\pi) 
    &\leq \lim_{\delta\rightarrow 0} \frac{|\log \delta|}{\sqrt{\tilde{\tau}^{\delta}}} \\
    &\leq \sqrt{2\eta} \frac{\min_{\mathcal{V}'\in\mathcal{E}_{\textnormal{Alt}}(\mathcal{V})} \sum_{x\in\mathcal{X}\setminus\{0\}} \Bar{P}_{\mathcal{V}_{\textnormal{min}}}(x)\mathbb{D}(\nu_x \Vert \nu_x')}{\sqrt{\chi_2\left(\sum_{x\in\mathcal{X}\setminus\{0\}}\Bar{P}_{\mathcal{V}_{\textnormal{min}}}(x)q_x\middle\Vert q_0\right)}}\\
    &\leq \sqrt{2\eta} \max_{\Bar{P}_X\in\mathcal{P}_{\mathcal{X}\setminus\{0\}}}\frac{\min_{\mathcal{V}'\in\mathcal{E}_{\textnormal{Alt}}(\mathcal{V})} \sum_{x\in\mathcal{X}\setminus\{0\}} \Bar{P}_{X}(x)\mathbb{D}(\nu_x \Vert \nu_x')}{\sqrt{\chi_2\left(\sum_{x\in\mathcal{X}\setminus\{0\}}\Bar{P}_{X}(x)q_x\middle\Vert q_0\right)}}. 
\end{align}

\bibliographystyle{IEEEtran}
\bibliography{main.bib}

\newpage

\appendix
\subsection{Proof of Lemma~\ref{lem:1}}
\label{apx:A}
Let $\sum_{x\neq 0}P_{X;\theta}(x)=\Tilde{\Theta}_{n\rightarrow \infty}(n^{-\alpha})$ for all $\theta\in\Theta$ for some $0<\alpha<1$. Then, we can choose some $\epsilon>0$ small enough such that $\alpha+\epsilon<1$, and 
\begin{align}
    \mathbb{P}_{\mathcal{V}_{\theta},\Tilde{\pi}}(N_{\theta}>n^{\alpha+\epsilon}) &\leq \sum_{t=n^{\alpha+\epsilon}}^{\infty}\sum_{\theta'\neq \theta} \mathbb{P}_{\mathcal{V}_{\theta},\Tilde{\pi}}\left(A_{\theta,\theta'}(t) < 0\right)\\
    &\leq \sum_{t=n^{\alpha+\epsilon}}^{\infty}\sum_{\theta'\neq \theta} \Bigg(\mathbb{P}_{\mathcal{V}_{\theta},\Tilde{\pi}}\left(A_{\theta,\theta'}(t) < 0, \sum_{i=1}^t \mathbf{1}(X_i\neq 0) \geq t n^{-\alpha-\epsilon/4}\right) \nonumber\\
    &\quad+ \mathbb{P}_{\mathcal{V}_{\theta},\Tilde{\pi}}\left( \sum_{i=1}^t \mathbf{1}(X_i\neq 0) < t n^{-\alpha-\epsilon/4} \right) \Bigg),
\end{align}
where we use the union bound and the law of total probability. 
For any $t\in\mathbb{N}^+$, we define 
\begin{align}
    V_t = \sum_{i=1}^t \mathbf{1}\left(X_i\neq 0\right) - \sum_{i=1}^t \sum_{x\neq 0}\mathbb{P}_{X_i|X^{t-1},Y^{i-1}}(x),
\end{align}
and the sequence $\{V_t\}$ is a martingale, where for all $i\in\mathbb{N}^+$, 
\begin{align}
    &\mathbb{E}_{\mathcal{V}_{\theta},\Tilde{\pi}}\left[\left(\mathbf{1}(X_i\neq 0) - \sum_{x\neq 0} \mathbb{P}_{X_i|X^{t-1},Y^{i-1}}(x) \right)^2\middle|\mathcal{F}_{i-1}\right] \nonumber\\
    &=  \sum_{x\neq 0} \mathbb{P}_{X_i|X^{t-1},Y^{i-1}}(x) - \left(\sum_{x\neq 0} \mathbb{P}_{X_i|X^{t-1},Y^{i-1}}(x)\right)^2\\
    &= \Tilde{O}_{n\rightarrow \infty}(n^{-\alpha}).
    \label{eqn:prof_lem10_1}
\end{align}
Then, for all $t\geq n^{\alpha+\epsilon}$, we have
\begin{align*}
    \mathbb{P}_{\mathcal{V}_{\theta},\Tilde{\pi}}\left(\sum_{i=1}^t \mathbf{1}\left(X_i\neq 0\right)< tn^{-\alpha-\epsilon/4}\right) &= \mathbb{P}_{\mathcal{V}_{\theta},\Tilde{\pi}}\left(V_t < tn^{-\alpha-\epsilon/4} - \sum_{i=1}^t \sum_{x\neq 0}\mathbb{P}_{X_i|X^{t-1},Y^{i-1}}(x)\right)\\
    &\leq \exp\left(-\frac{t^2\Tilde{\Omega}_{n\rightarrow \infty}(n^{-2\alpha})}{t\Tilde{O}_{n\rightarrow \infty}(n^{-\alpha})+t\Tilde{O}_{n\rightarrow \infty}(n^{-\alpha})}\right)\\
    &\leq \exp(-t\Tilde{\Omega}_{n\rightarrow \infty}(n^{-\alpha})),
\end{align*}
where we use the fact that $tn^{-\alpha-\epsilon/4}-\sum_{i=1}^t \sum_{x\neq 0}\mathbb{P}_{X_i|X^{t-1},Y^{i-1}}(x)= -t\Tilde{\Omega}(n^{-\alpha})$, the result in \eqref{eqn:prof_lem10_1} and Freedman's inequality. 
Moreover, for all $t\geq n^{\alpha+\epsilon}$, 
\begin{align}
    &\mathbb{P}_{\mathcal{V}_{\theta},\Tilde{\pi}}\left(A_{\theta,\theta'}(t) < 0, \sum_{i=1}^t \mathbf{1}(X_i\neq 0) \geq t n^{-\alpha-\epsilon/4}\right) \nonumber\\
    &= \mathbb{P}_{\mathcal{V}_{\theta},\Tilde{\pi}}\left(\sum_{i=1}^t L_{\theta,\theta'}(i) - \sum_{i=1}^t \mathbb{D}(\nu_{\theta}^{X_i}\Vert \nu_{\theta'}^{X_i}) < - \sum_{i=1}^t \mathbb{D}(\nu_{\theta}^{X_i}\Vert \nu_{\theta'}^{X_i}), \sum_{i=1}^t \mathbf{1}(X_i\neq 0) \geq tn^{-\alpha-\epsilon/4}\right)\nonumber\\
    &\leq \mathbb{P}_{\mathcal{V}_{\theta},\Tilde{\pi}}\left(\sum_{i=1}^t L_{\theta,\theta'}(i) - \sum_{i=1}^t \mathbb{D}(\nu_{\theta}^{X_i}\Vert \nu_{\theta'}^{X_i}) < -tn^{-\alpha-\epsilon/4} \min_{x\in\mathcal{X}\setminus\{0\}} \mathbb{D}(\nu_{\theta}^x\Vert \nu_{\theta'}^x)\right)\\
    &\leq \exp\left(- \frac{t^2\Tilde{\Omega}_{n\rightarrow \infty}(n^{-2\alpha-\epsilon/2})}{t\Tilde{O}_{n\rightarrow \infty}(n^{-\alpha})+t\Tilde{O}_{n\rightarrow \infty}(n^{-\alpha-\epsilon/4})}\right)\label{eqn:prof_lem10_2}\\
    &\leq \exp\left(t \Tilde{\Omega}_{n\rightarrow \infty}(n^{-\alpha - \epsilon/2})\right),
\end{align}
where in \eqref{eqn:prof_lem10_2} we use Freedman's inequality and the fact that 
\begin{align}
    W_t \triangleq \sum_{i=1}^t L_{\theta,\theta'}(i) - \sum_{i=1}^t \mathbb{D}(\nu_{\theta}^{X_i}\Vert \nu_{\theta'}^{X_i}) 
\end{align}
is a martingale, and 
\begin{align}
    \mathbb{E}_{\mathcal{V}_{\theta},\Tilde{\pi}}\left[\left(L_{\theta,\theta'}(t)-\mathbb{D}(\nu_{\theta}^{X_t}\Vert \nu_{\theta'}^{X_t})\right)^2\middle|\mathcal{F}_{t-1}\right] &\leq \mathbb{E}_{\mathcal{V},\Tilde{\pi}}[(L_{\theta,\theta'}(t))^2]\\
    &\leq \Tilde{O}_{n\rightarrow \infty}(n^{-\alpha}). 
\end{align}
Then, for any $\epsilon>0$, we have 
\begin{align}
    \mathbb{P}_{\mathcal{V}_{\theta},\Tilde{\pi}}(N_{\theta}>n^{\alpha+\epsilon}) &\leq \sum_{t=n^{\alpha+\epsilon}}^{\infty}\sum_{\theta'\neq \theta} \left(\exp\left(-t\Tilde
    {\Omega}_{n\rightarrow \infty}(n^{-\alpha})\right) + \exp\left(-t\Tilde{\Omega}_{n\rightarrow \infty}(n^{-\alpha-\epsilon/2})\right)\right)\\
    &\leq O_{n\rightarrow \infty}(n^{-\beta})  
\end{align}
for any $\beta>0$.

\subsection{Proof of Lemma~\ref{lem:2}}
\label{apx:prof_lem6} 

Assuming the policy $\pi$ can achieve the detection error exponent $\gamma$, then for any $0<\kappa<1$, we have from \eqref{eqn:conv_4} that
\begin{align}
    \lim_{n\rightarrow\infty} \mathbb{P}_{\mathcal{V}_{\theta},\pi}\left(A_{\theta,\theta''}(\tau) \geq \kappa \gamma\sqrt{n}\right) = 1
    \label{eqn:prof_lem_2_1}
\end{align}
for any $\theta\in\Theta$ and $\theta''\neq\theta$ . By the fact that $\Theta$ is a finite set, we have
\begin{align}
    \lim_{n\rightarrow\infty} \mathbb{P}_{\mathcal{V}_{\theta},\pi}\left(\min_{\theta''\neq \theta}A_{\theta,\theta''}(\tau) \geq \sqrt{n}\kappa\gamma\right) &= 1 - \lim_{n\rightarrow\infty} \mathbb{P}_{\mathcal{V}_{\theta},\pi}\left(\min_{\theta''\neq \theta}A_{\theta,\theta''}(\tau) < \sqrt{n}\kappa\gamma\right) \\
    &\geq 1 - \lim_{n\rightarrow\infty}\sum_{\theta''\neq \theta} \mathbb{P}_{\mathcal{V}_{\theta},\pi}\left(A_{\theta,\theta''}(\tau) < \sqrt{n}\kappa\gamma\right) \label{eqn:prof_lem_2_2}\\
    &= 1, 
    \label{eqn:prof_lem_2_3}
\end{align}
where \eqref{eqn:prof_lem_2_2} follows from \eqref{eqn:prof_lem_2_1}. 
By the law of total probability, for any $0<\zeta<1$, $\theta\in\Theta$ and any $\theta'\neq\theta$,  
\begin{align}
    \mathbb{P}_{\mathcal{V}_{\theta},\pi}\left(\min_{\theta''\neq \theta}A_{\theta,\theta''}(\tau) \geq \sqrt{n}\kappa\gamma\right) &\leq \mathbb{P}_{\mathcal{V}_{\theta},\pi}\left(\tau > \frac{ \kappa\zeta \sqrt{n}\gamma}{\sum_{x}P_{X;\theta}(x)\mathbb{D}(\nu_{\theta}^x\Vert \nu_{\theta'}^x)}\right) \nonumber\\
    &+ \mathbb{P}_{\mathcal{V}_{\theta},\Tilde{\pi}}\left( \max_{1\leq t\leq \frac{ \kappa\zeta\sqrt{n}\gamma}{\sum_{x}P_{X;\theta}(x)\mathbb{D}(\nu_{\theta}^x\Vert \nu_{\theta'}^x)}} \min_{\theta''\neq\theta} A_{\theta,\theta''}(t) \geq \sqrt{n}\kappa\gamma \right),
    \label{eqn:prof_lem_2_4}
\end{align}
where we have replace the policy $\pi$ by its dummy policy $\tilde{\pi}$ in the second term of the right hand side \eqref{eqn:prof_lem_2_4}. 
Combining \eqref{eqn:prof_lem_2_3} and \eqref{eqn:prof_lem_2_4}, we have that 
\begin{align}
    &\lim_{n\rightarrow\infty} \mathbb{P}_{\mathcal{V}_{\theta},\pi}\left(\tau > \frac{ \kappa\zeta \sqrt{n}\gamma}{\sum_{x}P_{X;\theta}(x)\mathbb{D}(\nu_{\theta}^x\Vert \nu_{\theta'}^x)}\right) \nonumber\\
    &\quad\geq 1- \lim_{n\rightarrow\infty} \mathbb{P}_{\mathcal{V}_{\theta},\Tilde{\pi}}\left( \max_{1\leq t\leq \frac{ \kappa\zeta\sqrt{n}\gamma}{\sum_{x}P_{X;\theta}(x)\mathbb{D}(\nu_{\theta}^x\Vert \nu_{\theta'}^x)}} \min_{\theta''\neq\theta} A_{\theta,\theta''}(t) \geq \sqrt{n}\kappa\gamma \right)
    \label{eqn:prof_lem_2_9}
\end{align}
for any $\theta\in\Theta$ and $\theta'\neq \theta$. 
Note that for any $t\in\mathbb{N}^+$, the event $\min_{\theta''\neq\theta} A_{\theta,\theta''}(t) > 0$ implies that the ML estimation $\hat{\theta}_{\textnormal{ML}}(t)$ is correct when the true hypothesis is $\theta$. Therefore, the event $\min_{\theta''\neq \theta}A_{\theta,\theta''}(t)>\sqrt{n}\kappa \gamma$ implies that there exists some time $\ell<t$ such that $\min_{\theta''\neq\theta}A_{\theta,\theta''}(\ell) < C$ for some $C=O_{n\rightarrow \infty}(1)$ 
and $\hat{\theta}_{\textnormal{ML}}(i)=\theta$ for all $ \ell\leq i\leq t$. In another word, there exists some $\ell$ such that the ML estimate is correct for all $\ell\leq i\leq t$ and the minimum log likelihood ratio $\min_{\theta''\neq\theta}A_{\theta,\theta''}(\ell)$ is upper bounded by some constant at the time $\ell$. 
Then, we can upper bound the second probability term on the right hand side of \eqref{eqn:prof_lem_2_4} as follow. 
\begin{align}
    &\mathbb{P}_{\mathcal{V}_{\theta},\Tilde{\pi}}\left( \max_{1\leq t\leq \frac{ \kappa\zeta\sqrt{n}\gamma}{\sum_{x}P_{X;\theta}(x)\mathbb{D}(\nu_{\theta}^x\Vert \nu_{\theta'}^x)}} \min_{\theta''\neq \theta}A_{\theta,\theta''}(t) \geq \sqrt{n}\kappa\gamma \right) \nonumber\\
    &\quad \leq \sum_{t=1}^{\frac{ \kappa\zeta\sqrt{n}\gamma} {\sum_{x}P_{X;\theta}(x)\mathbb{D}(\nu_{\theta}^x\Vert \nu_{\theta'}^x)}} \sum_{\ell=1}^t \mathbb{P}_{\mathcal{V}_{\theta},\Tilde{\pi}}\left( \min_{\theta''\neq \theta} \sum_{i=\ell}^t L_{\theta,\theta''}(i) \geq \sqrt{n}\kappa\gamma - C, \hat{\theta}_{\textnormal{ML}}(i)=\theta \textnormal{ for all } i \geq \ell\right)\\
    &\quad \leq \sum_{t=1}^{\frac{ \kappa\zeta\sqrt{n}\gamma} {\sum_{x}P_{X;\theta}(x)\mathbb{D}(\nu_{\theta}^x\Vert \nu_{\theta'}^x)}} \sum_{\ell=1}^t \mathbb{P}_{\mathcal{V}_{\theta},\tilde{\pi}}\Bigg(  \sum_{i=\ell}^t L_{\theta,\theta'}(i) - \sum_{i=\ell}^t\mathbb{D}(\nu_{\theta}^{X_i}\Vert \nu_{\theta'}^{X_i})\nonumber\\
    &\hspace{4cm}\geq \sqrt{n}\kappa\gamma - \sum_{i=\ell}^t\mathbb{D}(\nu_{\theta}^{X_i}\Vert \nu_{\theta'}^{X_i}) - C 
    , \hat{\theta}_{\textnormal{ML}}(i)=\theta \textnormal{ for all } i \geq \ell\Bigg)
    \label{eqn:prof_lem_2_8}
\end{align}
Note that for any  $x\in\mathcal{X}\setminus\{0\}$, $1\leq t \leq \frac{ \kappa\zeta\sqrt{n}\gamma} {\sum_{x}P_{X;\theta}(x)\mathbb{D}(\nu_{\theta}^x\Vert \nu_{\theta'}^x)}$, $\ell\leq t$ and $\epsilon''>0$, it holds that
\begin{align}
    &\mathbb{P}_{\mathcal{V}_{\theta},\Tilde{\pi}}\left(\sum_{i=\ell}^t \mathbb{D}(\nu_{\theta}^{X_i}\Vert \nu_{\theta'}^{X_i}) > (1+\epsilon'')\kappa\zeta\sqrt{n}\gamma,\hat{\theta}_{\textnormal{ML}}(i)=\theta \textnormal{ for all }\ell \leq i \leq t \right) \nonumber\\
    &\quad \leq \mathbb{P}_{\mathcal{V}_{\theta},\Tilde{\pi}}\Bigg(\sum_{i=\ell}^t \mathbb{D}(\nu_{\theta}^{X_i}\Vert \nu_{\theta'}^{X_i}) - \sum_{i=\ell}^{t} \sum_{x} P_{X;\theta}(x) \mathbb{D}(\nu_{\theta}^{x}\Vert \nu_{\theta'}^x)> \epsilon'' \kappa\zeta\sqrt{n}\gamma \nonumber\\
    &\hspace{5cm}, \hat{\theta}_{\textnormal{ML}}(i)=\theta \textnormal{ for all }\ell \leq i \leq t\Bigg) \label{eqn:prof_lem_2_5}\\
    &\quad\leq \exp\left(-\frac{ \left(\epsilon''  \kappa\zeta\sqrt{n}\gamma \right)^2}{\frac{ \kappa\zeta\sqrt{n}\gamma} {\sum_{x'}P_{X;\theta}(x')\mathbb{D}(\nu_{\theta}^x\Vert \nu_{\theta'}^x)} O_{n\rightarrow \infty}\left(\sum_{x\neq 0}P_{X;\theta}(x)\right) + O_{n\rightarrow \infty}\left(\kappa\zeta\sqrt{n}\gamma\right)}\right)\label{eqn:prof_lem_2_6}\\
    &\quad \leq \exp(-\Omega_{n\rightarrow \infty}(n^{1/2}))  \label{eqn:prof_lem_2_7}
\end{align}
where \eqref{eqn:prof_lem_2_5} comes from the fact that $t \leq \frac{ \kappa\zeta\sqrt{n}\gamma} {\sum_{x}P_{X;\theta}(x)\mathbb{D}(\nu_{\theta}^x\Vert \nu_{\theta'}^x)}$, and \eqref{eqn:prof_lem_2_6} follows from Bernstein's inequality and the similar change of measure technique as in \eqref{eqn:4_A_2_2}-\eqref{eqn:4_A_2_5} so that $P_{X_i|X^{i-1},Y^{i-1}}$ are replaced by $P_{X;\theta}$ for all $\ell\leq i\leq t$ under the joint event $\hat{\theta}_{\textnormal{ML}}(i)=\theta \textnormal{ for all }\ell \leq i \leq t$.
When 
\begin{align}
    \sum_{i=\ell}^t \mathbf{1}(X_i=x) \leq (1+\epsilon'')\frac{ \kappa\zeta\sqrt{n}\gamma} {\sum_{x}P_{X;\theta}(x)\mathbb{D}(\nu_{\theta}^x\Vert \nu_{\theta'}^x)}P_{X;\theta} (x),
\end{align}
it holds that 
\begin{align}
    \sqrt{n}\kappa\gamma - \sum_{i=\ell}^t\mathbb{D}(\nu_{\theta}^{X_i}\Vert \nu_{\theta'}^{X_i}) - C  \geq \sqrt{n}\kappa\gamma - C - (1+\epsilon'')\kappa\zeta\gamma \sqrt{n}.
\end{align}
For each $0<\kappa<1$ and $0<\zeta<1$, we can choose $\epsilon''$ small enough such that $\sqrt{n}\kappa\gamma - C - (1+\epsilon'')\kappa\zeta\gamma \sqrt{n}>\sqrt{n}\epsilon'''$ for some $\epsilon'''>0$. Then, by the law of total probability, each term in the summation of \eqref{eqn:prof_lem_2_8} can be upper bounded by 
\begin{align}
    &\exp(-\Omega_{n\rightarrow \infty}(n^{1/2})) + \mathbb{P}_{\mathcal{V}_{\theta},\Tilde{\pi}}\Bigg(  \sum_{i=\ell}^t L_{\theta,\theta'}(i) - \sum_{i=\ell}^t\mathbb{D}(\nu_{\theta}^{X_i}\Vert \nu_{\theta'}^{X_i}) \geq \epsilon'''\sqrt{n}\Bigg) \nonumber\\
    &\quad \leq \exp(-\Omega_{n\rightarrow \infty}(n^{1/2})) + \exp\left(-\frac{\Omega_{n\rightarrow \infty}(\epsilon'''\sqrt{n})^2}{ \frac{ \kappa\zeta\sqrt{n}\gamma} {\sum_{x}P_{X;\theta}(x)\mathbb{D}(\nu_{\theta}^x\Vert \nu_{\theta'}^x)} O_{n\rightarrow \infty}\left(\sum_{x\neq 0}P_{X;\theta}(x)\right) + O_{n\rightarrow \infty}(\sqrt{n})}\right) \nonumber \\
    &\quad \leq \exp(-\Omega_{n\rightarrow \infty}(n^{1/2}))
\end{align}
by Freedman's inequality. 
Then, for any $\theta\in\Theta$ and $\theta'\neq \theta$, we have 
\begin{align}
    \lim_{n\rightarrow \infty} \mathbb{P}_{\mathcal{V}_{\theta},\Tilde{\pi}}\left( \max_{1\leq t\leq \frac{ \kappa\zeta\sqrt{n}\gamma}{\sum_{x}P_{X;\theta}(x)\mathbb{D}(\nu_{\theta}^x\Vert \nu_{\theta'}^x)}} \min_{\theta''\neq \theta}A_{\theta,\theta''}(t) \geq \sqrt{n}\kappa\gamma \right) = 0,
\end{align}
which completes the proof by plugging above inequality into \eqref{eqn:prof_lem_2_9}. 

\subsection{Proof of Lemma~\ref{lem:3}}
\label{apx:proof_lem3}
We assume that $\sum_{x\neq 0}P_{X;\theta}(x)=\omega_{n \rightarrow \infty}(n^{-1/2})$ for some $\theta\in\Theta$ and show that the covertness constraint is violated. Let $\sum_{x\neq 0}P_{X;\theta}(x)=\Tilde{\Theta}_{n\rightarrow \infty}(n^{-\alpha})$ for some $\alpha < 1/2$. 
Note that for any $\theta$, $\theta'\neq\theta$ and some $\epsilon>0$ small enough, it holds that 
\begin{align}
    &\frac{1}{\Tilde{n}_{\theta,\theta'}}\sum_{i=1}^{\Tilde{n}_{\theta,\theta'}}\mathbb{P}_{\mathcal{V}_{\theta},\pi}(\hat{\theta}(i-1)=\theta,\tau \geq i) 
    \nonumber\\
    &\quad \geq 1 - \frac{1}{\Tilde{n}_{\theta,\theta'}}\sum_{i=1}^{\Tilde{n}_{\theta,\theta'}} \mathbb{P}_{\mathcal{V}_{\theta},\pi}(\tau < i) - \frac{1}{\Tilde{n}_{\theta,\theta'}}\sum_{i=1}^{\Tilde{n}_{\theta,\theta'}} \mathbb{P}_{\mathcal{V}_{\theta},\pi}(\hat{\theta}(i-1)\neq \theta, \tau \geq i)\\
    &\quad \geq 1 - o_{n\rightarrow \infty}(1) - \frac{1}{\Tilde{n}_{\theta,\theta'}}\sum_{i=1}^{\Tilde{n}_{\theta,\theta'}} \mathbb{P}_{\mathcal{V}_{\theta},\Tilde{\pi}}(N_{\theta} \geq i-1)\\
    &\quad \geq 1 - o_{n\rightarrow \infty}(1) - \frac{1}{\Tilde{n}_{\theta,\theta'}} \left(n^{\alpha+\epsilon} + \sum_{i=n^{\alpha+\epsilon}+1}^{\tilde{n}_{\theta,\theta'}} \mathbb{P}_{\mathcal{V}_{\theta},\Tilde{\pi}}(N_{\theta} \geq i-1)\right)\\
    &\quad \geq 1 - o_{n\rightarrow \infty}(1)
\end{align}
where we have use the result in Lemma~\ref{lem:1}, Lemma~\ref{lem:2} and the fact that $\tilde{n}_{\theta,\theta'} = \Tilde{\Omega}_{n\rightarrow\infty}(n^{1/2+\alpha})$
Therefore, there exists some $C>0$ such that 
\begin{align}
    \widehat{\alpha} &> C \left(\sum_{x\neq 0}P_{X;\theta}(x)\right),
\end{align}
where $\widehat{\alpha}$ is defined in \eqref{eqn:def_alpha_tilde}. 
Then, by using \eqref{eqn:prof_thm2_divergencebound2} and \eqref{eqn:prof_thm2_divergencebound3}, we have for any $\theta\in\Theta$ and $\theta'\neq \theta$ 
\begin{align}
    \mathbb{D}(P_{Z^n;\theta}\Vert (q_{\theta}^{0})^{\otimes n}) &\geq 
    \Tilde{n}_{\theta,\theta'} \mathbb{D}(\widehat{P}_{Z;\theta}||q_{\theta}^{0}) \\
    &\geq \Tilde{n}_{\theta,\theta'} \frac{\widehat{\alpha}^2}{2} \chi_2\left(\sum_{x\neq 0}\widehat{P}_{X;\theta}(x)q_{\theta}^x(z)\middle\Vert q_{\theta}^0\right)(1-o(1)) \\
    &\geq \Omega_{n\rightarrow \infty}\left(n^{1/2}\times \left(\sum_{x\neq 0}P_{X;\theta}(x)\right)\right) \\
    &= \omega_{n\rightarrow \infty}(1),
\end{align}
which violates the covertness constraint.

\subsection{Proof of Lemma~\ref{lem:bandit_0}}
\label{apx:D}
From \cite[Problem 33.4 (a)]{Lattimore2020}, one can observe that 
\begin{align*}
    \inf_{\mathcal{V}'\in \mathcal{E}_{\text{Alt}}(\hat{\mathcal{V}}(t))}\sum_{x\in\mathcal{X}\setminus\{0\}} \Bar{P}(x)\mathbb{D}(\hat{\nu}_x(t)||\nu_x') 
\end{align*}
is a continuous function of $\hat{\mathcal{V}}(t)$ and $\Bar{P}$.  By combining the fact that the Chi-square distance 
\begin{align*}
\chi_2\left(\sum_{x\in\mathcal{X}\setminus\{0\}}\Bar{P}(k) q_x||q_0\right)
\end{align*}
is also a continuous function of $\Bar{P}$, we claim that $\Bar{P}_{\mathcal{V}'}^{\zeta}$ is a continuous function of $\mathcal{V}'$ so as $\alpha_{\mathcal{V}'}^{\zeta}$. 
Therefore, for any $\epsilon>0$, there exists some $\hat{\epsilon}(\epsilon)$ related to the value of $\epsilon$ such that 
\begin{align}
    \left\Vert\hat{\mathcal{V}}(t)-\mathcal{V}\right\Vert_{\infty} \leq  \epsilon \textnormal{ and } \left\Vert \Bar{P}_{\hat{\mathcal{V}}(t)}^{\zeta} -\Bar{P}_{\mathcal{V}}^{\zeta}\right\Vert_{\infty} \leq \epsilon \textnormal{ and } \left\vert\frac{\alpha_{\hat{\mathcal{V}}(t)}^{\zeta}}{\alpha_{\mathcal{V}}^{\zeta}}-1\right\vert \leq \epsilon
\end{align}
whenever 
\begin{align}
    \left\Vert\hat{\mathcal{V}}(t)-\mathcal{V}\right\Vert_{\infty} \leq  \hat{\epsilon}(\epsilon). 
\end{align}
For any $\epsilon>0$ and $\gamma>0$, 
\begin{align}
    &\mathbb{P}_{\mathcal{V},\tilde{\pi}}\left( N_{\nu}(\epsilon) \geq |\log \delta|^{\alpha+\gamma}\right) \nonumber\\
    &\leq \mathbb{P}_{\mathcal{V},\tilde{\pi}}\left(\exists t\geq |\log \delta|^{\alpha+\gamma} \textnormal{ s.t  } ||\hat{\mathcal{V}}(t)-\mathcal{V}||_{\infty}>\hat{\epsilon}(\epsilon) \right)\\
    &\leq \sum_{t=|\log \delta|^{\alpha+\gamma}}^{\infty} \sum_{x\in\mathcal{X}\setminus\{0\}}\mathbb{P}_{\mathcal{V},\tilde{\pi}}\left( \left|\mu(\hat{\nu}_x(t))-\mu(\nu_x)\right| >  \hat{\epsilon}(\epsilon)\right) \\
    &= \sum_{t=|\log \delta|^{\alpha+\gamma}}^{\infty} \sum_{x\in\mathcal{X}\setminus\{0\}}\mathbb{P}_{\mathcal{V},\tilde{\pi}}\left( \left|\frac{1}{T_x(t)}\sum_{i=1}^t Y_i \mathbf{1}(X_i= x)-\mu(\nu_x)\right| >  \hat{\epsilon}(\epsilon) \right).
\end{align}
When $\delta$ is sufficiently small, for any $t\geq |\log\delta|^{\alpha+\gamma}$ and $x\in\mathcal{X}\setminus\{0\}$, 
it holds that 
\begin{align}
    &\mathbb{P}_{\mathcal{V},\tilde{\pi}} \left(T_x(t)< t|\log \delta|^{-\alpha-\gamma/2}\right) \nonumber\\
    &= \mathbb{P}_{\mathcal{V},\tilde{\pi}} \left(\sum_{i=1}^t \mathbf{1}(X_i=x) - \sum_{i=1}^t P_{X_i|Y^{i-1},X^{i-1}}(x) < t|\log \delta|^{-\alpha-\gamma/2}- \sum_{i=1}^t P_{X_i|Y^{i-1},X^{i-1}}(x)\right) \nonumber\\
    &\leq \mathbb{P}_{\mathcal{V},\tilde{\pi}} \left(\sum_{i=1}^t \mathbf{1}(X_i=x) - \sum_{i=1}^t P_{X_i|Y^{i-1},X^{i-1}}(x) < - t\Tilde{\Omega}_{\delta \rightarrow 0}(|\log \delta|^{-\alpha})\right) \\
    &\leq \exp\left(-\frac{ t^2 \Tilde{\Omega}_{\delta \rightarrow 0}(|\log \delta|^{-2\alpha})}{t \Tilde{O}_{\delta \rightarrow 0}(|\log \delta|^{-\alpha})+t \Tilde{O}_{\delta \rightarrow 0}(|\log \delta|^{-\alpha})}\right)\\
    &\leq \exp\left(-t \Tilde{\Omega}_{\delta \rightarrow 0}(|\log \delta|^{-\alpha})\right)
\end{align}
by Freedman's inequality. 
Therefore, when $\delta$ is sufficiently small, for any $t\geq |\log\delta|^{\alpha+\gamma}$ and $x\in\mathcal{X}\setminus\{0\}$, we have 
\begin{align}
    &\mathbb{P}_{\mathcal{V},\Tilde{\pi}}\left( \left|\frac{1}{T_x(t)}\sum_{i=1}^t Y_i \mathbf{1}(X_i= x)-\mu(\nu_x)\right| >  \hat{\epsilon}(\epsilon) \right) \nonumber\\
    &\leq \mathbb{P}_{\mathcal{V},\Tilde{\pi}} \left(T_x(t)< t |\log \delta|^{-\alpha-\gamma/2}\right) \nonumber\\
    &\quad + \sum_{k=t|\log \delta|^{-\alpha-\gamma/2}}^t \mathbb{P}_{\mathcal{V},\Tilde{\pi}}\left( \left|\frac{1}{T_x(t)}\sum_{i=1}^t Y_i \mathbf{1}(X_i= x)-\mu(\nu_x)\right| >  \hat{\epsilon}(\epsilon) , T_x(t) = k\right) \\
    &\leq \exp\left(-t \Tilde{\Omega}_{\delta\rightarrow 0}(|\log \delta|^{-\alpha})\right) + \sum_{k=t|\log \delta|^{-\alpha-\gamma/2}}^t\exp\left(-\frac{k(\hat{\epsilon}(\epsilon))^2}{2}\right) \label{eqn:prof_lem13_1}\\
    &\leq \exp\left(-t \Tilde{\Omega}_{\delta\rightarrow 0}(|\log \delta|^{-\alpha})\right) + t \exp\left(-\frac{t|\log \delta|^{-\alpha-\gamma/2}(\hat{\epsilon}(\epsilon))^2}{2}\right),
\end{align}
where in \eqref{eqn:prof_lem13_1} we use the Chernoff bound for sub-Gaussian random variables. 
Finally, we have for all $\delta$ sufficiently small, 
\begin{align} 
    \mathbb{P}_{\mathcal{V},\Tilde{\pi}}\left( N_{\nu}(\epsilon) \geq |\log \delta|^{\alpha+\gamma}\right) &\leq 
|\mathcal{X}|
    \sum_{t=|\log \delta|^{\alpha+\gamma}}^{\infty} \exp\left(-t \Tilde{\Omega}_{\delta\rightarrow 0}(|\log \delta|^{-\alpha})\right) + t \exp\left(-\frac{t|\log \delta|^{-\alpha-\gamma/2}\epsilon^2}{2}\right) \nonumber\\
    &\leq O_{\delta\rightarrow 0}(|\log \delta|^{-\beta})
\end{align}
for any $\beta>0$ and $\epsilon>0$.

\subsection{Proof of Lemma~\ref{lem:bandit_1}}

Since $\Gamma_t = \Omega_{\delta\rightarrow 0}(|\log \delta|)$, there exists some $a>0$ such that $\Gamma_t > a|\log \delta|$ when $\delta$ is small enough. 
Then, fixed any $b>0$ such that for all $\delta$ small enough, 
\begin{align}
    &\mathbb{P}_{\mathcal{V},\pi}\left(\tau \leq b|\log \delta|^2\right) \nonumber\\
    &\leq  \mathbb{P}_{\mathcal{V},\pi}\left(R_{b|\log \delta|^2} > \Gamma_t \right) \\
    &\leq \mathbb{P}_{\mathcal{V},\pi}\left(R_{b|\log \delta|^2} > a|\log \delta| \right)\\
    &\leq \mathbb{P}_{\mathcal{V},\pi}\left(\inf_{\nu'\in\mathcal{E}_{\text{Alt}}(\hat{\mathcal{V}}(b|\log \delta|^2))} \sum_{x\in\mathcal{X}\setminus\{0\}} T_x(b|\log \delta|^2) \mathbb{D}(\hat{\nu}_x(b|\log \delta|^2)||\nu_x') > a|\log \delta| \right)\\
    &\leq \mathbb{P}_{\mathcal{V},\pi}\left(   T_{\mathcal{X}\setminus\{0\}}(b|\log \delta|^2) \times c > a|\log \delta| \right)
    \label{eqn:prof_lem_bandit_1_1}, 
\end{align}
where \eqref{eqn:prof_lem_bandit_1_1} follows from the fact that the relative entropy $\mathbb{D}(\nu_x'\Vert v_x'')$ is bounded for all $\nu',\nu''$ and $x\in\mathcal{X}\setminus\{0\}$ so that we upper bound $\mathbb{D}(\hat{\nu}_x(b|\log \delta|^2)||\nu_x')$ by some $c>0$ for all $x\in\mathcal{X}\setminus\{0\}$. 
By choosing $b$ sufficiently small, there exists some $\epsilon'>0$ such that 
\begin{align}
    \frac{a}{c}|\log \delta| - \sum_{i=1}^{b|\log \delta|^2}\sum_{x\neq 0} P_{X_i|X^{i-1},Y^{i-1}}(x) \geq \epsilon' |\log \delta|
\end{align}
because 
\begin{align}
    \sum_{x\neq 0} P_{X_i|X^{i-1},Y^{i-1}}(x) = O_{\delta\rightarrow 0}(1/|\log \delta|)
\end{align}
for any $i\in\mathbb{N}$ by out construction of the policy. 
Then, we can use Freedman's concentration bound on martingales to bound $\eqref{eqn:prof_lem_bandit_1_1}$ as follows. 
\begin{align*}
     &\mathbb{P}_{\mathcal{V},\pi}\left(   T_{\mathcal{X}\setminus\{0\}}(b|\log \delta|^2) \times c > a|\log \delta| \right) \nonumber\\
     & \leq \mathbb{P}_{\mathcal{V},\pi}\left(  \sum_{i=1}^{b|\log \delta|^2} \mathbf{1}(X_i\neq 0)  > \frac{a}{c}|\log \delta| \right)\\
     & \leq \mathbb{P}_{\mathcal{V},\pi}\left(  \sum_{i=1}^{b|\log \delta|^2} \mathbf{1}(X_i\neq 0) - \sum_{i=1}^{b|\log \delta|^2}\sum_{x\neq 0} P_{X_i|X^{i-1},Y^{i-1}}(x) > \epsilon'|\log \delta|\right)\\
     &\leq \exp\left(-\Omega_{\delta\rightarrow 0}(|\log \delta|)\right) \\
     &\leq O_{\delta\rightarrow 0}(|\log \delta|^{-\beta})
\end{align*}
for any $\beta>0$, where we have use the fact that 
\begin{align}
    \left\{\sum_{i=1}^{n} \mathbf{1}(X_i\neq 0) - \sum_{i=1}^{n}\sum_{x\neq 0} P_{X_i|X^{i-1},Y^{i-1}}(x)\right\}_{n=1}^{\infty}
\end{align}
is a martingale sequence.

\subsection{Proof of \eqref{eqn:prof_thm4_bounds_of_G}}
\label{apx:bounds_of_G}
For any $t\geq |\log \delta|^{\alpha+\epsilon}$, it holds that 
\begin{align*}
    \mathbb{P}_{\mathcal{V},\Tilde{\pi}}\left(\mathcal{G}^c\right) &= \mathbb{P}_{\mathcal{V},\Tilde{\pi}}\left(||\hat{\mathcal{V}}(t)-\mathcal{V}||_{\infty} > \epsilon\right) \\
    &\leq \sum_{x\in\mathcal{X}\setminus\{0\}}\mathbb{P}_{\mathcal{V},\Tilde{\pi}}\left(|\mu(\hat{\nu}_x(t))-\mu(\nu_x)| > \epsilon\right) \\
    &\leq \sum_{x\in\mathcal{X}\setminus\{0\}} \Bigg( \mathbb{P}_{\mathcal{V},\Tilde{\pi}}\left(|\mu(\hat{\nu}_x(t))-\mu(\nu_x)| > \epsilon, T_x(t) \geq t|\log \delta|^{-\alpha-\epsilon/2} \right) \nonumber\\
    &\hspace{4cm}+ \mathbb{P}_{\mathcal{V},\Tilde{\pi}}\left(T_x(t)<t|\log \delta|^{-\alpha-\epsilon/2}\right)\Bigg)\\
    &\leq \sum_{x\in\mathcal{X}\setminus\{0\}} \left(\exp\left(-\frac{\epsilon^2 t |\log\delta|^{-\alpha-\epsilon/2}}{2}\right)+\exp\left(-\frac{t^2\Omega_{\delta\rightarrow 0}(|\log\delta|^{-2\alpha})}{t\Tilde{O}_{\delta\rightarrow 0}(|\log \delta|^{-\alpha})+t\Tilde{O}_{\delta\rightarrow 0}(|\log \delta|^{-\alpha})}\right)\right)\\
    &\leq \exp\left(-t\Tilde{\Omega}_{\delta\rightarrow 0}(|\log \delta|^{-\alpha-\epsilon/2})\right) 
\end{align*}
by the Chernoff bound of Gaussian random variables and Freeman's inequality that have been used frequently in this paper, and we omit the details. 
\end{document}